\documentclass[acmsmall]{acmart}

\usepackage[ruled, vlined, linesnumbered]{algorithm2e}
\DontPrintSemicolon
\usepackage{bm}
\usepackage{mathtools}
\usepackage{amsmath}
\usepackage{amsthm}
\usepackage{todonotes}
\usepackage{wrapfig}
\usepackage{graphicx}
\usepackage{subcaption} 
\usepackage{enumitem}
\usepackage{tabularx}
\usepackage{booktabs}
\usepackage{multirow}
\usepackage{makecell}
\usepackage{array}
\usepackage{pifont}
\usepackage[dvipsnames]{xcolor}
\usepackage[normalem]{ulem}
\usepackage{tikz}
\usepackage{float}
\usepackage[commandnameprefix=always,authormarkup=none,markup=underlined]{changes}
\usepackage{cancel}
\usepackage{marginnote}
\usepackage{tcolorbox}
\usepackage{geometry}

\usetikzlibrary{positioning, arrows.meta, calc, decorations.pathreplacing}
\makeatletter
\newtheorem*{rep@theorem}{\rep@title}
\newcommand{\newreptheorem}[2]{%
\newenvironment{rep#1}[1]{%
 \def\rep@title{#2 \ref{##1}}%
 \begin{rep@theorem}}%
 {\end{rep@theorem}}}
\makeatother

\newtheorem*{remark}{Remark}

\theoremstyle{plain}
\newtheorem{fact}{Fact}
\newtheorem*{lem*}{Lemma}
\newreptheorem{lem}{Lemma}
\newtheorem*{thm*}{Theorem}
\newreptheorem{theorem}{Theorem}
\newreptheorem{corollary}{Corollary}

\theoremstyle{definition}
\newtheorem{prob}{\textbf{Problem}}

\newcommand{\tr}{\operatorname{tr}}

\newcommand{\sidsfd}{\textsc{AeroSketch}\xspace}
\newcommand{\sidscod}{\textsc{AeroSketchCOD}\xspace}

\newcommand{\pfive}{\textsf{P5}\xspace}
\newcommand{\cmark}{{\color{green} \textbf{\checkmark}}}
\newcommand{\xmark}{{\color{red} \textbf{---} }}

\AtBeginDocument{%
  }

\setcopyright{cc}
\setcctype{by}
\acmJournal{PACMMOD}
\acmYear{2026} \acmVolume{4} \acmNumber{1 (SIGMOD)} \acmArticle{8} \acmMonth{2} \acmPrice{}\acmDOI{10.1145/3786622}

\begin{document}

\title[Near-Optimal Time Matrix Sketch Framework for Persistent, Sliding Window, and Distributed Streams]{\sidsfd: Near-Optimal Time Matrix Sketch Framework for Persistent, Sliding Window, and Distributed Streams}

\author{Hanyan Yin}
\email{yinhanyan@ruc.edu.cn}
\orcid{0000-0003-2262-4386}
\affiliation{%
  \institution{Renmin University of China}
  \state{Beijing}
  \country{China}
}
\author{Dongxie Wen}
\email{2019202221@ruc.edu.cn}
\orcid{0009-0004-8359-9182}
\affiliation{%
  \institution{Renmin University of China}
  \state{Beijing}
  \country{China}
}
\author{Jiajun Li}
\email{2015201613@ruc.edu.cn}
\orcid{0009-0008-4149-5165}
\affiliation{%
  \institution{Renmin University of China}
  \state{Beijing}
  \country{China}
}

\author{Zhewei Wei}
\email{zhewei@ruc.edu.cn}
\orcid{0000-0003-3620-5086}
\authornote{Zhewei Wei is the corresponding author. The work was partially done at Gaoling School of Artificial Intelligence, Beijing Key Laboratory of Research on Large Models and Intelligent Governance, Engineering Research Center of Next-Generation Intelligent Search and Recommendation, MOE, and Pazhou Laboratory (Huangpu), Guangzhou, Guangdong 510555, China. }
\affiliation{%
  \institution{Renmin University of China}
  \state{Beijing}
  \country{China}
}

\author{Xiao Zhang}
\email{zhangx89@ruc.edu.cn}
\orcid{0009-0001-1857-1368}
\affiliation{%
  \institution{Renmin University of China}
  \state{Beijing}
  \country{China}
}

\author{Peng Zhao}
\email{zhaop@lamda.nju.edu.cn}
\orcid{0000-0001-7925-8255}
\affiliation{%
  \institution{Nanjing University}
  \city{Nanjing}
  \state{Jiangsu}
  \country{China}}

\author{Zhi-Hua Zhou}
\email{zhouzh@lamda.nju.edu.cn}
\orcid{0000-0003-0746-1494}
\affiliation{%
  \institution{Nanjing University}
  \city{Nanjing}
  \state{Jiangsu}
  \country{China}}

\begin{abstract}
  Many real-world matrix datasets arrive as high-throughput vector streams, making it impractical to store or process them in their entirety. To enable real-time analytics under limited computational, memory, and communication resources, matrix sketching techniques have been developed over recent decades to provide compact approximations of such streaming data. Some algorithms have achieved optimal space and communication complexity. However, these approaches often require frequent time-consuming matrix factorization operations. In particular, under tight approximation error bounds, each matrix factorization computation incurs cubic time complexity, thereby limiting their update efficiency.

  In this paper, we introduce \textsc{AeroSketch}, a novel matrix sketching framework that leverages recent advances in randomized numerical linear algebra (RandNLA). \textsc{AeroSketch} achieves optimal communication and space costs while delivering near-optimal update time complexity (within logarithmic factors) across persistent, sliding window, and distributed streaming scenarios. Extensive experiments on both synthetic and real-world datasets demonstrate that \textsc{AeroSketch} consistently outperforms state-of-the-art methods in update throughput. In particular, under tight approximation error constraints, \textsc{AeroSketch} reduces the cubic time complexity to the quadratic level. Meanwhile, it maintains comparable approximation quality while retaining optimal communication and space costs.
\end{abstract}



\received{July 2025}
\received[revised]{October 2025}
\received[accepted]{November 2025}

\maketitle

\section{Introduction}

Many types of real-world data—such as computer network traffic~\cite{xia2014survey,kim2020tea}, online machine learning~\cite{kuzborskij2019efficient,chen2020efficient,wen2024fast,chen2022efficient,wen2024matrix}, and sensor data~\cite{Xie2023OptimalSD,Aquino2007ASD}—are continuously generated as vast, high-velocity streams of vectors~\cite{zhou2024learnability,ICML'24:LARA,muthukrishnan2005data}. Storing even a reasonable portion of these streams in memory is impractical. Therefore, it is necessary to design online algorithms that process one vector at a time and immediately discard it, enabling real-time processing with reduced memory usage at the cost of some precision. Such compact representations are known as \textit{matrix sketching}~\cite{cormode2020small,liberty2013simple,woodruff2014sketching,desai2016improved,kacham2024efficient}.

\begin{table*}[h]
    \centering
    \scriptsize
    \caption{Given the dimension $d$ of each row vector (for tracking approximate matrix multiplication over sliding window (Problem~\ref{prob:sw-amm}), there are two vector streams with dimensions $d_x$ and $d_y$, respectively), the upper bound on the relative covariance error $\varepsilon$, the norm of each row $\left\|\bm{a}_i\right\|_2^2 \in \left[1, R\right]$, and the number of sites $m$ (in distributed scenarios), this table presents an overview of the communication/space and amortized update time complexities of algorithms addressing a wide range of matrix sketching problems.  Here, $\delta \in (0,0.01)$ denotes the upper bound on failure probability; that is, each listed complexity bound holds for an algorithm that returns a matrix sketch satisfying the stated error guarantees with probability at least $1-\delta$.  \cmark\xspace indicates that the asymptotic complexity is optimal; \xmark\xspace indicates that it is not optimal; and \cmark* denotes near-optimal complexity (up to a logarithmic factor).
    } \label{tab:comparison}
    \begin{tabularx}{\textwidth}{@{}lX||X|p{1.8cm}r|p{2.8cm}r@{}}
        \toprule
        \multicolumn{2}{l||}{\textbf{Problem Scenarios}} & \textbf{Methods}                                    & \textbf{Space / Communication}          &                                                                      & \textbf{Amortized Update Time Complexity} &                                                                                            \\ \midrule \midrule
        \multicolumn{2}{l||}{\makecell[l]{Full Stream Matrix Sketch                                                                                                                                                                                                                                                                                                      \\(Problem~\ref{prob:full})}}     & Frequent Directions~\cite{liberty2013simple,ghashami2016frequent}               & $O\left(\frac{d}{\varepsilon}\right)$ &\cmark                                 & $O\left(\frac{d}{\varepsilon}\right)$ & \cmark                          \\ \midrule
        \multicolumn{1}{l|}{\multirow{7}{*}{\makecell[l]{Tracking                                                                                                                                                                                                                                                                                                        \\Matrix\\Sketch\\Problem\\(Scenario~1)}}}                             & \multirow{3}{*}{\makecell[l]{Tracking Matrix Sketch                                                                                                                                                                                                                               \\over Sliding Window\\(Problem~\ref{prob:sw-fd})}}\vspace{2.5ex}                                               & Fast-DS-FD~\cite{yin2024optimal}                                   & $O\left(\frac{d}{\varepsilon}\log R\right)$  &\cmark& $O\left(\left(\frac{d}{\varepsilon}+\frac{1}{\varepsilon^3}\right)\log R\right)$                       & \xmark \\
        \multicolumn{1}{l|}{}                            &                                                     & \textbf{[Theorem~\ref{thm:sw}]}         & $O\left(\frac{d}{\varepsilon}\log R\right)$                          & \cmark                                    & $O\left(\frac{d}{\varepsilon}\log d \log R\right)$                              & \cmark * \\
        \multicolumn{1}{l|}{}                            &                                                     & \textbf{[Corollary~\ref{coro:amplify}]} & $O\left(\frac{d}{\varepsilon}\log R\right)$                          & \cmark                                    & $O\left(\frac{d}{\varepsilon}\log d \log R \log^2\frac{1}{\delta}\right)$       & \cmark * \\\cmidrule(l){2-7}
        \multicolumn{1}{l|}{}                            & \multirow{2}{*}{\makecell[lt]{Tracking Persistent                                                                                                                                                                                                                                                             \\Matrix Sketch\\(ATTP, Problem~\ref{prob:attp})}}\vspace{2.5ex} & ATTP-FD~\cite{shi2021time}              & $O\left(\frac{d}{\varepsilon}\log \|\bm{A}\|_F^2\right)$               & \cmark   & $O\left(\frac{d}{\varepsilon^2}\right)$                   & \xmark \\
        \multicolumn{1}{l|}{}                            &                                                     & \textbf{[Theorem~\ref{thm:attp}]}       & $O\left(\frac{d}{\varepsilon}\log \|\bm{A}\|_F^2\right)$             & \cmark                                    & $O\left(\frac{d}{\varepsilon}\log d \log^2 \frac{1}{\delta}\right)$             & \cmark * \\\cmidrule(l){2-7}
        \multicolumn{1}{l|}{}                            & \multirow{2}{*}{\makecell[l]{Tracking Approximate                                                                                                                                                                                                                                                             \\ Matrix Multiplication\\ over Sliding Window\\(Problem~\ref{prob:sw-amm})}}\vspace{4.5ex}                                    & {\begin{tabular}{l}ML-SO-COD~\cite{xian2025optimal}\\hDS-COD~\cite{yao2025optimal}\end{tabular}} &$O\left(\frac{d_x+d_y}{\varepsilon} \log R\right)$& \cmark & $O\left(\left(\frac{d_x+d_y}{\varepsilon}+\frac{1}{\varepsilon^3}\right)\log R\right)$                         & \xmark \\
        \multicolumn{1}{l|}{}                            &                                                     & \textbf{[Theorem~\ref{thm:sw-amm}]}     & $O\left(\frac{d_x+d_y}{\varepsilon}\log R\right)$                    & \cmark                                    & $O\left(\frac{d_x+d_y}{\varepsilon}\log d\log R \log^2 \frac{1}{\delta}\right)$ & \cmark * \\ \cmidrule(l){1-7}
        \multicolumn{1}{l|}{\multirow{4}{*}{\makecell[l]{Tracking                                                                                                                                                                                                                                                                                                        \\Distributed                                                                                                                                                                                                                                                                           \\Matrix\\Sketch\\Problem\\(Scenario~2)}}}                            & \multirow{2}{*}{\makecell[l]{Tracking Matrix Sketch                                                                                                                                                                                                                               \\over Distributed Sites\\(Problem~\ref{prob:dsfd})}}\vspace{2.5ex}                                               & P2~\cite{ghashami2014continuous}                                  & $O\left(\frac{md}{\varepsilon}\log \left\|\bm{A}\right\|_F^2\right)$ &\cmark& $O\left(\frac{d}{\varepsilon^2}\right)$                       & \xmark \\
        \multicolumn{1}{l|}{}                            &                                                     & \textbf{[Theorem~\ref{thm:dist}]}       & $O\left(\frac{md}{\varepsilon}\log \left\|\bm{A}\right\|_F^2\right)$ & \cmark                                    & $O\left(\frac{d}{\varepsilon}\log d \log^2 \frac{1}{\delta}\right)$             & \cmark * \\\cmidrule(l){2-7}
        \multicolumn{1}{l|}{}                            & \multirow{2}{*}{\makecell[l]{Tracking Matrix Sketch                                                                                                                                                                                                                                                           \\over Distributed Sliding\\Window (Problem~\ref{prob:dswfd})}}\vspace{2.5ex}                                              & DA2~\cite{zhang2017tracking}                             &$O\left(\frac{md}{\varepsilon}\log \left\|\bm{A}\right\|_F^2\right)$ &  \cmark    & $O\left(\frac{d}{\varepsilon^2} \right) $                         & \xmark \\
        \multicolumn{1}{l|}{}                            &                                                     & \textbf{[Theorem~\ref{thm:dist-sw}]}    & $O\left(\frac{md}{\varepsilon}\log \left\|\bm{A}\right\|_F^2\right)$ & \cmark                                    & $O\left(\frac{d}{\varepsilon}\log d \log^2 \frac{1}{\delta}\right)$             & \cmark * \\\bottomrule
    \end{tabularx}
\end{table*}

We observe that a range of matrix sketching problems over vector streams—including tracking persistent matrix sketches~\cite{shi2021time}, matrix sketches over sliding windows~\cite{wei2016matrix,yin2024optimal}, distributed matrix sketches~\cite{ghashami2014continuous,zhang2017tracking}, and approximate matrix multiplication~\cite{xian2025optimal,yao2025optimal}—can be reduced to two fundamental tasks: \textit{tracking matrix sketch problem with limited memory} (Scenario~1) and \textit{tracking distributed matrix sketch problem over bandwidth-constrained channels} (Scenario~2), as categorized in Table~\ref{tab:comparison}. The primary optimization metrics of matrix sketch algorithms are space cost (for Scenario~1), communication cost (for Scenario~2), and time cost (for both Scenarios~1 and 2).

Extensive research has been dedicated to developing efficient matrix sketch algorithms for various problem settings~\cite{liberty2013simple,ghashami2016frequent,ghashami2014relative,shi2021time,zhang2017tracking,ghashami2014continuous,wei2016matrix,yin2024optimal,xian2025optimal,yao2025optimal}, as well as theoretically analyzing the relationships among approximation quality, space/communication cost, and time complexity~\cite{huang2017efficient,huang2021communication,ghashami2016frequent,liberty2022even,yin2024optimal}. \cite{huang2017efficient, huang2021communication} proved a lower bound on the communication cost for the distributed deterministic matrix sketch problem. In other words, any deterministic matrix sketching algorithm that solves the problem under Scenario~2 must incur a communication cost of at least a certain number of bits. The corresponding lower bounds for the space complexities under Scenario~1 have also been proved by other researchers~\cite{ghashami2016frequent,yin2024optimal,xian2025optimal,yao2025optimal}. In addition to the theoretical results, numerous specific matrix sketch algorithms have been proposed. Some of them have the same asymptotic communication/space complexity as the theoretical lower bound, indicating that the sketch algorithms have achieved optimal communication/space costs, as listed in Table~\ref{tab:comparison} and marked by \cmark .

Despite their superiority in communication/space complexity, these algorithms often suffer from high computational costs, especially under tight approximation error bounds and over high-dimensional vector streams of dimension $d$. Their asymptotic time complexities are suboptimal compared with the time lower bound proved by Huang~\cite{huang2019near} (marked by \xmark in Table~\ref{tab:comparison}). This is mainly because they continuously perform time-consuming matrix factorizations to monitor changes in the stream and update the matrix sketch. In particular, under tight approximation error bounds, each matrix factorization incurs a cubic time complexity of $O(d^3)$, which is prohibitive for high-dimensional vector streams. As a result, the prior communication- and space-optimal algorithms are far from optimal in terms of computational complexity. This underperformance naturally raises the question: \textit{Can we design matrix sketching algorithms that are simultaneously optimal in communication, space, and time complexity?}

Fortunately, we find that some more efficient RandNLA (randomized numerical linear algebra) operators, including Power Iteration and Simultaneous Iteration, can replace time-consuming deterministic matrix factorization operators~\cite{halko2011finding, martinsson2016randomized, martinsson2020randomized,tropp2023randomized, murray2023randomized}. Although substituting with RandNLA operators may seem straightforward, there is a cost: RandNLA operators introduce probabilistic errors, but deterministic matrix factorization operators do not. As a result, key steps in prior algorithms must also be adapted. Moreover, randomness inherent in RandNLA makes it difficult to determine appropriate parameter settings. It is also challenging to rigorously analyze how probabilistic errors affect the final matrix sketch's accuracy, as well as the communication, space, and time complexity. Another obstacle is that RandNLA operators may not directly meet requirements. For example, we need to identify singular values and vectors that exceed a given threshold $\theta$, but simultaneous iteration only supports computing the top-$k$ singular components for a specified $k$. However, we do not know in advance how many of the top-$k$ singular values exceed the threshold $\theta$. This adds further difficulty to both algorithm design and theoretical complexity analysis.

In this work, we overcome these obstacles through proper algorithm design and parameter selection. The stochastic error from RandNLA operators can be analytically regulated, while maintaining optimal space/communication complexity on par with the best baseline methods. The amortized time complexity per update also approaches near-optimality (within a logarithmic factor). The new generalized matrix sketching framework is termed \sidsfd, which accelerates a wide range of tracking matrix sketching algorithms to near-optimal update time complexity (marked as \cmark* in Table~\ref{tab:comparison}) along with optimal communication/space complexity.
The specific contributions of this paper are outlined below:

\begin{itemize}[leftmargin=*]
    \item \textbf{Novel Framework:} For a wide range of tracking matrix sketching problems listed in Table~\ref{tab:comparison}, we design a new generalized matrix sketching framework, termed \sidsfd, which improves the time efficiency of the previous state-of-the-art methods, while maintaining their optimal communication and space overhead. In particular, under tight approximation error constraints $\varepsilon$, baseline algorithms degrade to $O(d^3)$ time complexity while our \sidsfd framework achieves $O(d^2 \log d)$ time complexity per update.
    \item \textbf{Rigorous Theoretical Analysis:} We provide a rigorous end-to-end proof of \sidsfd, demonstrating its competitive approximation quality, communication and space complexities, and superior update time complexities compared to state-of-the-art methods, especially under tight approximation error constraints $\varepsilon$. To the best of our knowledge, \sidsfd is the first implementable algorithm with theoretical guarantees of explicit probabilistic error bounds, near-optimal update time complexity, and optimal space/communication complexity.
    \item \textbf{Extensive Experiments:} We implement the \sidsfd framework and conduct comprehensive experiments to verify its superiority over state-of-the-art algorithms, particularly in terms of empirical approximation error, communication/space cost, and time consumption. Experimental results show that the trade-off between error and computational cost for the \sidsfd-optimized algorithm is more favorable than that for the best baseline algorithms on both synthetic and real-world datasets. The communication/space cost of our framework is also competitive. Furthermore, optimizing time cost becomes increasingly significant as the upper bound on covariance relative error $\varepsilon$ tightens and the dimensionality of the vectors $d$ increases substantially. These experimental findings strongly align with our theoretical analyses.
\end{itemize}

\section{Preliminaries and Related Work}

This section introduces problem definitions and relevant tools, including Power Iteration, Simultaneous Iteration, and Frequent Directions, on which our \sidsfd framework is based.

\subsection{Problem Definition: Matrix Sketching}
\label{sec:prob-defn}

We now provide formal definitions of the different variants of matrix sketching problems listed in Table~\ref{tab:comparison}. A stream $\bm{A}$ consists of row vectors $\bm{a}_i\in \mathbb{R}^{1\times d}$, where the norm of each row satisfies $\left\|\bm{a}_i\right\|_2^2 \in \left[1, R\right]$. Given two time values $s < t$, define $\bm{A}_{s,t}\in\mathbb{R}^{(t-s)\times d}$ as the portion of the stream that arrived during the time interval $(s, t]$. Let $0$ denote a time point before any elements of the stream have arrived, and $T$ denote the current time. We further define the specific stream subsets $\bm{A}_t = \bm{A}_{0, t}$ and $\bm{A}_{-t} = \bm{A}_{t,T}$. The sizes of both the sketching state and the output sketch matrix are significantly smaller than the matrix to be approximated. In distributed settings, the communication cost between the \(m\) sites and the coordinator, as well as the size of the output matrix, is significantly smaller than that of the matrix to be approximated.

\subsubsection{Full Stream Matrix Sketch}
\label{sec:full-stream-matrix-sketch}

The goal of the \textit{full stream matrix sketch} problem is to maintain a matrix sketch of a stream from its beginning up to the current time $T$. The formal definition is~\cite{liberty2013simple,ghashami2016frequent}:

\begin{prob}[Full Stream Matrix Sketching]
    \label{prob:full}
    Given an error parameter $\varepsilon$, at any current time $T$, we must maintain a matrix sketch $\bm{B}_T$ such that $\bm{B}_T$ approximates the matrix $\bm{A}_T$. The approximation quality is assured by the \textit{covariance error} bound, defined as follows:
    \begin{equation*}
        \textup{\textbf{cov-error}}(\bm{A}_T, \bm{B}_T) = \lVert \bm{A}_T^\top\bm{A}_T - \bm{B}_T^\top\bm{B}_T\rVert_2 \le O(1) \varepsilon \lVert\bm{A}_T\rVert_F^2.
    \end{equation*}
\end{prob}

\subsubsection{Matrix Sketching over Sliding Window}
\label{sec:matrix-sketch-over-sliding-windows}

In real-world scenarios, the focus often lies on the most recently arrived elements rather than outdated items within data streams. Datar et al.~\cite{datar2002maintaining} considered the problem of maintaining aggregates and statistics from the most recent portion of a data stream and referred to such a model as the \textit{sliding window model}. The \textit{matrix sketching over sliding window} problem is formally stated as follows~\cite{wei2016matrix,yin2024optimal}:

\begin{prob}[Matrix Sketching over Sliding Window]
    \label{prob:sw-fd}
    Given an error parameter $\varepsilon$ and a window size $N$, the goal is to maintain a matrix sketch $\bm{B}_\text{W}$ such that, at the current time $T$, $\bm{B}_\text{W}$ approximates the matrix $\bm{A}_\text{W} = \bm{A}_{T-N,T} \in \mathbb{R}^{N \times d}$. The approximation quality is assured by the \textit{covariance error} bound:
    \begin{equation*}
        \textup{\textbf{cov-error}}(\bm{A}_\text{W}, \bm{B}_\text{W}) = \lVert \bm{A}^\top_\text{W}\bm{A}_\text{W} - \bm{B}^\top_\text{W}\bm{B}_\text{W}\rVert_2 \le O(1) \varepsilon \lVert\bm{A}_\text{W}\rVert_F^2.
    \end{equation*}
\end{prob}

\begin{remark}
    Besides the space-optimal Fast-DS-FD algorithm proposed by Yin et al. in~\cite{yin2024optimal} as we listed in Table~\ref{tab:comparison}, they also mentioned a time-optimized version called \textbf{probabilistic Fast-DS-FD} in Theorem~3.1, claiming that the per-update time complexity is $O(d\ell)$, by setting $\ell \gets \min\left(\lceil \frac{1}{\varepsilon} \rceil, d\right)$, so that the time complexity becomes $O\left(\frac{d}{\varepsilon}\right)$. Unfortunately, probabilistic Fast-DS-FD lacks a concrete implementation and therefore remains a potential direction for future work. It also lacks strict proofs of error bounds and complexities. Important parameters are not defined explicitly, such as how to set the number of iterations for the RandNLA algorithm to meet the error bound. Therefore, we did not list it in Table~\ref{tab:comparison} as a competitor to our~\sidsfd framework. Additionally, we found that Lemma~1 in \cite{yin2024optimal}, which involves the key operation of dumping and restoring snapshots and the associated correctness proof, holds only under the premise that $\bm{v}_j$ is an exact singular vector. However, the probabilistic randomized algorithm used by probabilistic Fast-DS-FD to optimize time can obtain only an approximate estimate of $\bm{v}_j$. Therefore, directly replacing the SVD in Fast-DS-FD with RandNLA will introduce cumulative restoring norm errors, as shown in Figure~\ref{fig:norm-error}. Designing the dumping and restoring snapshots operation and correctness proof under this premise is far from straightforward. In this paper, our~\sidsfd framework will overcome these limitations.
\end{remark}

\subsubsection{(At-the-time) Persistent Matrix Sketching}
\label{sec:persistent-matrix-sketch}

In this paper, we focus on the ATTP (\textit{at-the-time persistent}) sketch task proposed by Shi et al.~\cite{shi2021time}, which requires supporting queries on arbitrary historical versions of the full stream matrix sketch. The formal definition is as follows:

\begin{prob}[ATTP Sketch]
    \label{prob:attp}
    Given an error parameter $\varepsilon$, we maintain a matrix sketch such that, at the current time $T$, it can return an approximation $\bm{B}_{t}$ for any matrix $\bm{A}_t = \bm{A}_{0, t} \in \mathbb{R}^{t \times d}$ with $t \le T$. The approximation quality is assured by the \textit{covariance error} bound:
    \begin{equation*}
        \forall t<T,\quad\textup{\textbf{cov-error}}(\bm{A}_t, \bm{B}_t) = \lVert \bm{A}_t^\top\bm{A}_t - \bm{B}_t^\top\bm{B}_t\rVert_2 \le O(1) \varepsilon \lVert\bm{A}_t\rVert_F^2.
    \end{equation*}
\end{prob}

\subsubsection{Approximate Matrix Multiplication over Sliding Window}
\label{sec:approximated-matrix-multiplication-over-sliding-windows}

We follow the problem definition of \textit{approximate matrix multiplication over sliding window} as in~\cite{yao2024approximate,xian2025optimal,yao2025optimal}:

\begin{prob}[Approximate Matrix Multiplication over Sliding Window]
    \label{prob:sw-amm}
    Let \(\{(\boldsymbol{x}_t, \boldsymbol{y}_t)\}_{t \ge 0}\) be sequences of two vector streams, where for each time \(t\), we have \(\boldsymbol{x}_t \in \mathbb{R}^{d_x}\) and \(\boldsymbol{y}_t \in \mathbb{R}^{d_y}\). For a fixed window size \(N\) and for any time \(T \ge N\), define the sliding window matrices:
    \[
        \boldsymbol{X}_\text{W} = \begin{bmatrix} \boldsymbol{x}_{T-N+1} & \boldsymbol{x}_{T-N+2} & \cdots & \boldsymbol{x}_T \end{bmatrix} \in \mathbb{R}^{d_x \times N},\]
    \[
        \boldsymbol{Y}_\text{W} = \begin{bmatrix} \boldsymbol{y}_{T-N+1} & \boldsymbol{y}_{T-N+2} & \cdots & \boldsymbol{y}_T \end{bmatrix} \in \mathbb{R}^{d_y \times N}.
    \]
    The goal is to yield, at every time \(T \ge N\), matrices \(\boldsymbol{A}_\text{W} \in \mathbb{R}^{d_x \times \ell}\) and \(\boldsymbol{B}_\text{W} \in \mathbb{R}^{d_y \times \ell}\) satisfying:
    \[
        \left\|\boldsymbol{X}_\text{W} \boldsymbol{Y}_\text{W}^\top - \boldsymbol{A}_\text{W} \boldsymbol{B}_\text{W}^\top\right\|_2 \le O(1) \epsilon\, \|\boldsymbol{X}_\text{W}\|_F \, \|\boldsymbol{Y}_\text{W}\|_F.
    \]
\end{prob}

\subsubsection{Tracking Distributed Matrix Sketch}
\label{sec:tracking-distributed-matrix-sketching}

In many scenarios, such as distributed databases, wireless sensor networks, and cloud computing, there are multiple distributed input streams. Each stream is observed by one of the $m$ distributed sites. Suppose we have a distributed data stream \(A = \{(\bm{a}_i, t_i, S_i) \mid i = 1, 2, \dots\}\), where each item \(\bm{a}_i \in \mathbb{R}^{1\times d}\) is a \(d\)-dimensional vector with timestamp \(t_i\) arriving at site \(S_i\in \{1, 2, \dots, m\}\). We adopt the definition of \textit{tracking distributed matrix sketch} given by Ghashami et al.~\cite{ghashami2014continuous}:

\begin{prob}[Tracking Distributed Matrix Sketch]
    \label{prob:dsfd}
    At any current time \(T\), define the matrix \(\bm{A}_T\) as the collection of all rows \(\bm{a}_i\) across all \(m\) sites with timestamps in the interval \((0,T]\). The goal is to maintain, at a central coordinator, a compact sketch matrix \(\bm{B}_T\) such that it approximates \(\bm{A}_T\), with the guarantee of the covariance error bound:
    \begin{equation*}
        \textup{\textbf{cov-error}}(\bm{A}_T, \bm{B}_T) = \lVert \bm{A}_T^\top\bm{A}_T - \bm{B}_T^\top\bm{B}_T\rVert_2 \le O(1) \varepsilon \lVert\bm{A}_T\rVert_F^2.
    \end{equation*}
\end{prob}

In distributed settings, attention is also often focused on recent elements, known as the problem of \textit{tracking distributed matrix sketch over sliding window}. The definition given by Zhang et al.~\cite{zhang2017tracking} is:

\begin{prob}[Tracking Distributed Matrix Sketch over Sliding Window]
    \label{prob:dswfd}
    Given a window size \(N\), define \(\bm{A}_\text{W}\) to consist of all vectors \(\bm{a}_i\) across all \(m\) sites with timestamps in the interval \((T - N,T]\). Maintain a compact sketch matrix \(\bm{B}_\text{W}\) at a central coordinator such that it approximates \(\bm{A}_\text{W}\) with bounded covariance error:
    \begin{equation*}
        \textup{\textbf{cov-error}}(\bm{A}_\text{W}, \bm{B}_\text{W}) = \lVert \bm{A}_\text{W}^\top\bm{A}_\text{W} - \bm{B}_\text{W}^\top\bm{B}_\text{W}\rVert_2 \le O(1) \varepsilon \lVert\bm{A}_\text{W}\rVert_F^2.
    \end{equation*}
\end{prob}

\begin{remark}
    In Table~\ref{tab:comparison}, the baseline algorithm we list for solving Problem~\ref{prob:dswfd} is DA2, proposed by Zhang et al.~\cite{zhang2017tracking}. Although the update time complexity of DA2 is not explicitly given in~\cite{zhang2017tracking}, since each update requires performing an SVD on a matrix of size $d \times O(1/\varepsilon)$, the update time complexity of DA2 should be dominated by this operation, with a time cost of $O(d/\varepsilon^2)$. When $\varepsilon < O(1/d)$, the time complexity degrades to $O(d^3)$. An analysis in the experimental section of the original paper also confirms this: ``This is because DA2 needs to compute matrix factorizations periodically, which leads to running time quadratic or even cubic in $d$.'' (last paragraph of Sec~IV.B in~\cite{zhang2017tracking}).
\end{remark}

\subsection{RandNLA Operators}

\subsubsection{Power Iteration}
\textbf{Power Iteration} (also known as the \textbf{Power Method}) is a classic algorithm for computing an approximate largest singular value~\cite{trefethen2022numerical,kuczynski1992estimating}. Given a matrix $\bm{A}$, the algorithm produces an estimate $\hat{\sigma}_1^2$ of the squared largest singular value of $\bm{A}$, along with an estimated corresponding top singular vector $\hat{\bm{v}}_1$. The pseudocode is shown in Algorithm~\ref{alg:power}.

\begin{algorithm}[h]
    \caption{Power Iteration Algorithm}
    \label{alg:power}
    \KwIn{$\bm{A}\in \mathbb{R}^{d\times \ell}$: the matrix for which the approximate largest singular value (squared) and corresponding singular vector are to be computed; $k$: the number of iterations.}
    \KwOut{$\hat{\sigma}_1^2$: the squared largest singular value of $\bm{A}$; $\bm{v}_1$: the corresponding singular vector.}
    \tcc{$\mathcal{N}(\bm{0}, \bm{I}_\ell)$ is the standard multivariate normal distribution}
    $\bm{x}^{(0)} \gets \mathcal{N}(\bm{0}, \bm{I}_\ell)$, $\bm{x}^{(0)} \gets \frac{\bm{x}^{(0)}}{\left\| \bm{x}^{(0)} \right\|_2}$\;

    \For{$i=1$ \KwTo $k$}{
        $\bm{x}^{(i)} \gets \bm{A}^\top \bm{A} \bm{x}^{(i-1)}$, $\bm{x}^{(i)} \gets \frac{\bm{x}^{(i)}}{\left\| \bm{x}^{(i)} \right\|_2}$\;
    }
    $\hat{\sigma_1}^2 \gets \left(\bm{x}^{(k)}\right)^\top \bm{A}^\top \bm{A} \bm{x}^{(k)}, \bm{v}_1 \gets \bm{x}^{(k)}$\;
    \Return $\hat{\sigma_1}^2, \bm{v}_1$\;
\end{algorithm}

A textbook theorem about the convergence rate of Power Iteration states the following: if the initial vector $\bm{x}^{(0)}$ in line 1 of Algorithm~\ref{alg:power} is not orthogonal to the top singular vector, then after $k$ iterations the gap between the estimated squared singular value $\hat{\sigma}_1^2$ and the true squared top singular value $\sigma_1^2$ satisfies $\lvert\hat{\sigma}_1^2 - \sigma_1^2\rvert = O\left(\sigma_2^2/\sigma_1^2\right)^{2k}$~\cite{trefethen2022numerical}, where $\sigma_2^2 / \sigma_1^2$ (\textit{spectral gap}) is the ratio of the second-largest squared singular value of $\bm{A}$ to the top one.

However, instead of the \textit{spectral-gap-dependent} error bound, the \textit{spectral-gap-free} error bound (which does not depend on $\sigma_2^2$) of Power Iteration is more convenient for proving the theorems related to the \sidsfd\ framework proposed in this paper. A \textit{spectral-gap-free} error bound was proved by Kuczy\'nski and Wo\'zniakowski as Theorem 4.1(a) in~\cite{kuczynski1992estimating}, which provides a tight bound but in a rather lengthy form. Therefore, we derive a looser yet more concise version that suffices for our analysis and state it as follows:

\begin{corollary}[Probabilistic error bound and time complexity of Power Iteration, a simplified corollary of Theorem~4.1(a) in~\cite{kuczynski1992estimating}]
    \label{thm:power}
    For Algorithm~\ref{alg:power}, if we set $k = \lceil\log_{2} d\rceil+1$, then the probability that the estimated top squared singular value $\hat{\sigma}_1^2$ exceeds half of the true value $\sigma_1^2 / 2$ is bounded by:
    \begin{equation*}
        \Pr\left[\hat{\sigma}_1^2 \ge \sigma_1^2/2\right] \ge 1- {\frac{2}{\pi\sqrt{e}}\cdot\frac{1}{\sqrt{d \log_2 d}}}.
    \end{equation*}
    Power Iteration requires only matrix-vector multiplications. The total computational time cost of Algorithm~\ref{alg:power} is $O(d \ell k)$. If we set $k = \lceil\log_2 d\rceil+1$, then the time cost becomes $O(d \ell \log d)$.
\end{corollary}

The original statement of Theorem 4.1(a) in~\cite{kuczynski1992estimating}, together with the detailed derivation of Corollary~\ref{thm:power}, can be found in Appendix~\ref{appendix:proof-power}.

\subsubsection{Simultaneous Iteration}

\textbf{Simultaneous Iteration}, also known as \textbf{Subspace Iteration} or \textbf{Orthogonal Iteration}, is an efficient algorithm for approximating the top-$k$ singular values and vectors of a matrix (in contrast to Power Iteration, which approximates only the top-1)~\cite{musco2015randomized}. The procedure for this method is shown in Algorithm~\ref{alg:simul}.%

\begin{algorithm}[h]
    \caption{Simultaneous Iteration (\texttt{simul\_iter})}
    \label{alg:simul}
    \KwIn{$\bm{A}\in \mathbb{R}^{d\times \ell}$: the matrix for which the top-$k$ singular values and their corresponding singular vectors are to be obtained; $k$: the number of singular components to compute; $\varepsilon_{\text{SI}}$: an error coefficient within the range $(0,1)$.}
    \KwOut{$\bm{Z}\in \mathbb{R}^{d\times k}$: the matrix of the approximate top-$k$ singular vectors; $\hat{\bm{\Sigma}}\in \mathbb{R}^{k\times k}$: the diagonal matrix of the approximate top-$k$ singular values.}

    $q \gets \Theta\left(\frac{\log d}{\varepsilon_{\text{SI}}}\right), \bm{\Pi} \sim \mathcal{N}(0, 1)^{\ell \times k}$\;
    $\bm{K} \gets \left(\bm{A} \bm{A}^\top \right)^q\bm{A} \bm{\Pi}$ \tcp*[f]{$\bm{K}\in\mathbb{R}^{d \times k}$, takes $\Theta\left(qd\ell k\right)$ time}\;
    $\bm{Q}, \bm{R} \gets \texttt{QR}(\bm{K})$ \tcp*[f]{$\bm{Q}\in \mathbb{R}^{d\times k}$, takes $\Theta\left(dk^2\right)$ time}\;
    $\bm{M}\gets \bm{Q}^\top\bm{A}\bm{A}^\top \bm{Q}$\tcp*[f]{$\bm{M}\in \mathbb{R}^{k\times k}$, takes $O(d\ell k)$ time }  \;
    $[\bm{U}, \hat{\bm{\Sigma}}^2, \bm{U}^\top]\gets \texttt{SVD}(\bm{M})$ \tcp*[f]{takes $O(k^3)$ time}\;
    \Return $\bm{Z} \gets \bm{Q}\bm{U}$, and $\hat{\bm{\Sigma}}$
\end{algorithm}

After running Algorithm~\ref{alg:simul} on matrix $\bm{A}$, we obtain the approximate top-$k$ singular vectors $\bm{Z} = \left[\bm{z}_1, \bm{z}_2, \dots, \bm{z}_k\right]$ and the approximate top-$k$ singular values $\hat{\bm{\Sigma}} = \operatorname{diag}(\hat{\sigma}_1, \hat{\sigma}_2, \dots, \hat{\sigma}_k)$ of $\bm{A}$. In more intuitive terms, if the exact singular value decomposition (SVD) of $\bm{A}$ is given by $\bm{A} = \bm{U} \bm{\Sigma} \bm{V}^\top$, then Algorithm~\ref{alg:simul} returns $\bm{Z} \approx \bm{U}_k$ and $\hat{\bm{\Sigma}} \approx \bm{\Sigma}_k$, where $\bm{\Sigma}_k$ is the diagonal matrix containing the top-$k$ singular values of $\bm{A}$, and $\bm{U}_k$ contains the corresponding singular vectors.

The probabilistic error guarantees for the approximations, produced by Algorithm~\ref{alg:simul} provided by Musco and Musco~\cite{musco2015randomized}, are sufficient for our later analysis and are stated as follows:

\begin{theorem}[Probabilistic error bounds and time complexity of Simultaneous Iteration, the main lemma proven in~\cite{musco2015randomized}]
    \label{thm:simul}
    Given a matrix $\bm{A} \in \mathbb{R}^{d \times \ell}$, a target rank $k$, and an error parameter $\varepsilon_{\text{SI}}$, after executing Algorithm~\ref{alg:simul}, we obtain matrices $\bm{Z}\in\mathbb{R}^{d\times k}$ and $\hat{\bm{\Sigma}}\in\mathbb{R}^{k\times k}$. With high probability (at least $99/100$), the following guarantees hold:

    \begin{enumerate}
        \item $\|\bm{A} - \bm{Z}\bm{Z}^\top \bm{A}\|_F \le \left(1 + \varepsilon_{\text{SI}}\right) \|\bm{A} - \bm{A}_k\|_F$,

              where $\bm{A}_k$ is the best rank-$k$ approximation of $\bm{A}$. If the singular value decomposition of $\bm{A}$ is $\bm{U} \bm{\Sigma} \bm{V}^\top$, then $\bm{A}_k = \bm{U}_k \bm{\Sigma}_k \bm{V}_k^\top$. \footnote{Specifically, for $k = 0$, we define $\bm{A}_0 = \bm{0}$.}

        \item $\|\bm{A} - \bm{Z} \bm{Z}^\top \bm{A}\|_2 \le \left(1 + \varepsilon_{\text{SI}}\right) \|\bm{A} - \bm{A}_k\|_2$.

        \item $\forall i,\ \left|\bm{u}_i^\top \bm{A} \bm{A}^\top \bm{u}_i - \bm{z}_i^\top \bm{A} \bm{A}^\top \bm{z}_i\right| \le \varepsilon_{\text{SI}} \cdot \sigma_{k+1}^2$, where $\bm{u}_i$ is the $i$-th left singular vector of $\bm{A}$ and $\bm{z}_i$ is the $i$-th column of $\bm{Z}$.
    \end{enumerate}

    The comments in Algorithm~\ref{alg:simul} analyze the time cost of each line. The time complexity of Algorithm~\ref{alg:simul} is dominated by line 2, which takes $\Theta(q d \ell k)$ time, where $q = \Theta\left(\frac{\log d}{\varepsilon_{\text{SI}}}\right)$. If $\varepsilon_{\text{SI}}$ is a constant, then the total time complexity of Algorithm~\ref{alg:simul} becomes $\Theta(d \ell k \log d)$.

\end{theorem}

\subsection{Frequent Directions}
\label{sec:frequent-directions}

\textbf{Frequent Directions} (FD)~\cite{liberty2013simple,ghashami2016frequent} is a deterministic matrix sketch algorithm in the row-update model. It solves Problem~\ref{prob:full} with optimal space and update time complexity.
At initialization, FD creates an all-zero matrix $\bm{B}$ with the size of $2\ell\times d$, where $\ell$ is an input parameter satisfying $\ell\le d/2$.
To process each arriving row vector $\bm{a}_i$, FD first checks whether \(\bm{B}\) contains any zero-valued rows. If so, it inserts \(\bm{a}_i\) into one of them. Otherwise, all $2\ell$ rows are filled, and it performs a singular value decomposition \([\bm{U}, \bm{\Sigma}, \bm{V}^\top] = \texttt{SVD}(\bm{B})\), rescales the ``directions” in \(\bm{B}\) using the \(\ell\)-th largest singular value \(\sigma_\ell\), and ``forgets'' the least significant direction in the column space of \(\bm{V}^\top\). The updated \(\bm{\Sigma}^\prime\) and sketch matrix \(\bm{B}^\prime\) are computed as:
\begin{equation*}
    \bm{\Sigma}^\prime = \text{diag}\left(\sqrt{\sigma_1^2 - \sigma_\ell^2}, \dots, \sqrt{\sigma_{\ell-1}^2 - \sigma_\ell^2}, 0, \dots, 0 \right),
\end{equation*}
and \(\bm{B}^\prime = \bm{\Sigma}^\prime \bm{V}^\top\), where \(\sigma_1^2 \ge \sigma_2^2 \ge \dots \ge \sigma_{2\ell}^2\). The updated \(\bm{B}^\prime\) will have at least \(\ell + 1\) zero rows, allowing the process to continue updating the sketch matrix \(\bm{B}\) as new row vectors arrive. The amortized computational cost per row is \(O(d\ell)\).

Frequent Directions has also been adapted to the matrix sketching variants defined in Section~\ref{sec:prob-defn}—for example, PFD~\cite{shi2021time} for Problem~\ref{prob:attp} (ATTP Sketch) and DS-FD~\cite{yin2024optimal} for Problem~\ref{prob:sw-fd} (Matrix Sketching over Sliding Window). However, a common drawback of these methods is that they fail to achieve the optimal update time complexity of the original Frequent Directions algorithm, especially under tight approximation error constraints $\varepsilon$. In this paper, we develop a new framework with theoretical guarantees, termed \sidsfd, to overcome this limitation for Problems~\ref{prob:sw-fd},~\ref{prob:attp},~\ref{prob:sw-amm},~\ref{prob:dsfd}, and~\ref{prob:dswfd}. %

\section{Our Framework: \sidsfd}

In this section, we present a framework, termed \sidsfd, for accelerating the update time complexity of the baseline algorithms solving the tracking matrix sketching problems introduced in Section~\ref{sec:prob-defn}. We first select Problem~\ref{prob:sw-fd} as a representative case and describe how our \sidsfd\ framework is applied. We also provide rigorous end-to-end theoretical proofs for the error bound, time, and space complexity. In later sections, we show how the \sidsfd\ framework can be applied to optimize other problems.

\subsection{\sidsfd Framework Description}

\subsubsection{High-level idea}

Taking Problem~\ref{prob:sw-fd} as an example, \sidsfd shares a common high-level idea and skeleton with Fast-DS-FD: we implicitly maintain a matrix sketch $\bm{B}_\text{W}$ to approximate the matrix $\bm{A}_\text{W}$ consisting of the most recent $N$ elements in the stream. As time progresses, matrix $\bm{A}_\text{W}$ continuously changes as new vectors arrive and the old vectors slide out of the window. When the difference between the matrix $\bm{A}_\text{W}$ and the sketch matrix $\bm{B}_\text{W}$ exceeds the error threshold $\theta$, this discrepancy is computed, and a \textbf{snapshot} containing information about the difference and the current timestamp is recorded. As the sliding window advances, snapshots with old timestamps expire. The remaining recorded snapshots can be summed to form the sketch matrix $\bm{B}_\text{W}$ that approximates matrix $\bm{A}_\text{W}$ within the current sliding window.%

In Fast-DS-FD~\cite{yin2024optimal}, the difference between $\bm{A}_\text{W}$ and $\bm{B}_\text{W}$ is computed by singular value decomposition (SVD). It calculates the singular values and vectors of the covariance residual matrix $\bm{C}^\top\bm{C} = \bm{A}_\text{W}^\top \bm{A}_\text{W} - \bm{B}_\text{W}^\top \bm{B}_\text{W}$, and components where the singular values are greater than $\theta$ are recorded as snapshots. If we set $\theta=O(1)\cdot\varepsilon\left\|\bm{A}_\text{W}\right\|_F^2$, then we will have \[\left\|\bm{A}_\text{W}^\top \bm{A}_\text{W} - \bm{B}_\text{W}^\top \bm{B}_\text{W}\right\|_2 =\left\|\bm{C}^\top\bm{C}\right\|_2\le \theta = O(1)\cdot\varepsilon\left\|\bm{A}_\text{W}\right\|_F^2,\] which satisfies the covariance error bound.  Since $\bm{A}_\text{W}$ (consequently $\bm{C}$) changes with each incoming row vector, there may be new singular values greater than $\theta$. Therefore, SVD must be performed on $\bm{C}^\top\bm{C}$ for each update. When $1/\varepsilon > d$, the residual matrix $\bm{C}$ reaches a size of $d \times d$, and each SVD update will cost $O(d^3)$ time. For high-dimensional vector streams, this complexity is unacceptable.

As mentioned in the introduction, \sidsfd uses RandNLA operators, such as Power Iteration and Simultaneous Iteration, to overcome the time bottleneck caused by SVD. This is because we only need singular values greater than the error threshold $\theta$, while the time-consuming full SVD computes all singular values, resulting in computational waste. However, since RandNLA methods are approximate algorithms, directly summing the approximate snapshots using the method of Fast-DS-FD would introduce significant \textbf{cumulative restoring norm error} in the $\bm{B}_\text{W}$, in the form of
\begin{equation}
    \label{eq:norm-error}
\left\|\sum_{i=T-N+1}^T \left(\bm{Z}_i \bm{Z}_i^\top {\bm{C}_i^\prime}^\top \bm{C}_i^\prime + {\bm{C}_i^\prime}^\top \bm{C}_i^\prime \bm{Z}_i \bm{Z}_i^\top - 2\bm{Z}_i \bm{Z}_i^\top {\bm{C}_i^\prime}^\top \bm{C}_i^\prime \bm{Z}_i \bm{Z}_i^\top\right) \right\|_2.
\end{equation}

\begin{wrapfigure}{r}{0.4\textwidth}
\centering
\includegraphics[width=\linewidth]{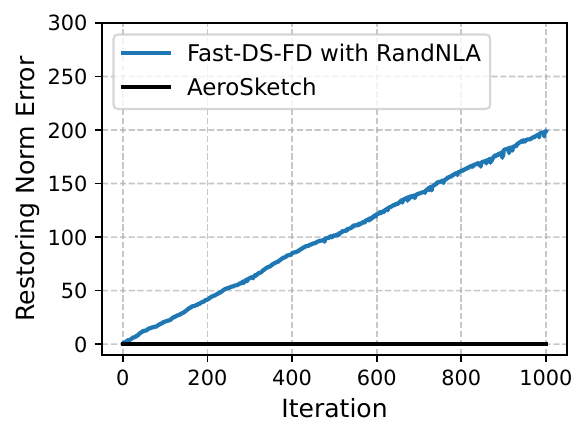}
\caption{Directly replacing the SVD in Fast-DS-FD with RandNLA would introduce cumulative restoring norm errors as shown in Eq.~\eqref{eq:norm-error}, whereas our \sidsfd does not.}
\label{fig:norm-error}
\end{wrapfigure}

If $\bm{Z}_i$'s are formed by exact singular vectors of ${\bm{C}_i^\prime}^\top \bm{C}_i^\prime$ as in Fast-DS-FD, then Eq.~\eqref{eq:norm-error} is implicitly zero. However, if the $\bm{Z}_i$'s are approximate singular vectors, the restoring norm error becomes explicitly nonzero and can accumulate as the update proceeds, as shown by the blue line in Figure~\ref{fig:norm-error} (See Appendix~\ref{app:norm-error} for a detailed derivation). It is also difficult to derive an upper bound for this error. Fortunately, we found that by redesigning the saved snapshot ($\bm{Z}_i$ and $\bm{Z}_i^\top {\bm{C}^\prime_i}^\top {\bm{C}^\prime_i}$ at line 15 of Algorithm~\ref{alg:sidsfd-update}) and the restoring method of $\bm{B}_\text{W}$ involving the snapshots (Algorithm~\ref{alg:pidsfd-query}), this cumulative error can be eliminated. Furthermore, other execution procedures and parameter selections must be meticulously configured to bound the error, space/communication costs, and time complexity. Through analysis and experiments, we find that the restoring norm error becomes zero (as shown by the black line in Figure~\ref{fig:norm-error}) and the algorithm's error now stems solely from the FD reduce operation and the residual matrix $\bm{C}_{T-N}$ at time $T-N$, both of which are more tractable for error analysis. Additionally, we discovered that the auxiliary sketch and queue in Fast-DS-FD can be removed, saving their space and time overhead.

\begin{algorithm}[h]
    \caption{ \sidsfd : \textsc{Initialize}($d, \varepsilon, \theta$)}
    \label{alg:sidsfd-init}
    \KwIn{
        $d$: dimension of input vectors of \textsc{Update};
        $\varepsilon$: upper error bound;
        $\theta$: dump threshold.  %
    }
    $\ell\gets\lceil\frac{2}{\varepsilon}\rceil$,\quad $\bm{C}_0\leftarrow \bm{0}_{2\ell\times d}$\;
    $\mathcal{S}\leftarrow \mathsf{queue}.\textsc{Initialize()}$ \tcp*[f]{Queue of snapshots}\;
\end{algorithm}

\subsubsection{Data Structures}

Algorithm~\ref{alg:sidsfd-init} shows the data structures used by \sidsfd. At initialization, \sidsfd requires the input vector dimension $d$, the upper bound on the relative covariance error $\varepsilon$, and a dump threshold $\theta$. It then initializes an empty matrix $\bm{C}_0$ of size $2\ell \times d$ and a queue of snapshots $\mathcal{S}$.

\begin{algorithm}[h]
    \caption{\sidsfd : \textsc{Update}($\bm{a}_i$)}
    \label{alg:sidsfd-update}
    \KwIn{$\bm{a}_i$: the row vector arriving at timestamp $i$.}

    Insert $\bm{a}_i$ into an empty row of $\bm{C}_{i-1}$ to get $\tilde{\bm{C}}_i$\;

    \eIf(\tcp*[f]{$\tilde{\bm{C}}_i$ has no zero rows}){$\text{rows}(\tilde{\bm{C}}_i) == 2\ell$}{
        $[\bm{U}_i, \bm{\Sigma}_i, \bm{V}_i^\top] \gets \mathtt{SVD}(\tilde{\bm{C}}_i)$  \tcp*[f]{$\bm{\Sigma}_i=\text{diag}\left(\sigma_1, \sigma_2 \dots,\sigma_{\min\left(2\ell, d\right)}\right)$}\;
        $\bm{C}_i^\prime \gets \sqrt{\max\left(\bm{\Sigma}_i^2 - \sigma_\ell^2\bm{I}, \bm{0}\right)} \, \bm{V}_i^\top$ \tcp*[f]{element-wise max}\;
    }{
        ${\bm{C}}_i^\prime \gets \tilde{\bm{C}}_i$\;
    }

    $\hat{\sigma}_1^2, \hat{\bm{v}}_1 \gets \mathtt{power\_iteration}({\bm{C}^\prime_i}^\top , k=\lceil\log_2 d \rceil+1)$ \;%

    \eIf{$\hat{\sigma}_1^2 \geq \frac{\theta}{2}$}{
    \For{$j=1$ \KwTo $\lceil \log_2 \min\left(\ell, d\right) \rceil$}{
    $k \gets \min(2^j, \ell, d) $\;
    $[\bm{Z}_i^\prime, \hat{\bm{\Sigma}}] \gets \mathtt{simul\_iter}({\bm{C}^\prime_i}^\top , k, \varepsilon_{\text{SI}}=0.4)$ \tcp*[f]{$\hat{\bm{\Sigma}}=\text{diag}\left(\hat{\sigma}_1, \hat{\sigma}_2 \dots,\hat{\sigma}_{k}\right)$}\;%
    \If{$\hat{\bm{\sigma}}^2_k<\theta$}{
    $\xi \gets \max_{\xi} \{\hat{\bm{\sigma}}_{\xi}^2\ge\theta\}$\;
    $\bm{Z}_i \gets {\bm{Z}_i^{\prime}}[:,1:\xi\ ]$\;
    $\mathcal{S}.\textsc{append}((\bm{Z}_i,\bm{Z}_i^\top {\bm{C}^\prime_i}^\top {\bm{C}^\prime_i}), s=\mathcal{S}[-1].t+1, t=i)$ \tcp*[f]{$\mathcal{S}[-1]$ is the tail element in queue $\mathcal{S}$}\;
    $\bm{C}_{i} \gets {\bm{C}^\prime_i} - {\bm{C}^\prime_i}\bm{Z}_i\bm{Z}_i^\top$  \;
    break\;
    }
    }
    $\bm{C}_i\gets \bm{C}_i^\prime$\;
    }{
    $\bm{C}_i\gets \bm{C}_i^\prime$\;
    }
\end{algorithm}

\subsubsection{Update Algorithm}

Algorithm~\ref{alg:sidsfd-update} shows how \sidsfd updates the sketch matrix $\bm{C}_i$ when a new row vector $\bm{a}_i$ arrives at timestamp $i$. First, we insert $\bm{a}_i$ into an empty row of $\bm{C}_{i-1}$ to form $\tilde{\bm{C}}_i$ (line 1). If the number of rows in $\tilde{\bm{C}}_i$ reaches $2\ell$, we apply lines 3-4, the FD reduction described in Section~\ref{sec:frequent-directions} to form $\bm{C}_i^\prime$. Otherwise, we simply set $\bm{C}_i^\prime$ to $\tilde{\bm{C}}_i$ (line 6). Next, we perform Power Iteration on matrix $\bm{C}_i^\prime$ to estimate the squared top singular value $\hat{\sigma}_1^2$ (line 7). If it exceeds half the dump threshold, i.e., $\hat{\sigma}_1^2 > \theta / 2$ (line 8), we apply Simultaneous Iteration to $\bm{C}_i^\prime$ to compute the approximate squared top-$k$ singular values greater than $\theta$ along with the corresponding singular vectors. Thus, instead of recomputing the exact singular vectors, we quickly approximate them in a few iterations.

However, there is a catch: before applying Simultaneous Iteration, we do not know how many squared singular values are greater than $\theta$, and thus we cannot predefine $k$, which is a required parameter of Simultaneous Iteration. To resolve this, we use the classic \textit{doubling} technique. We start by computing the top-2 approximate singular values and vectors of $\bm{C}_i^\prime$ by setting $k = 2$ in Simultaneous Iteration, and then continue doubling $k$ (i.e., $k = 4, 8, 16, \dots, 2^j$, as in lines~9-11). Eventually, we find the value $k = 2^j$ such that the $k$-th top estimated squared singular value is smaller than $\theta$ (line~12). This implies the existence of some $\xi$ satisfying $2^{j-1} \le \xi < 2^j$ such that the $\xi$-th largest estimated squared singular value is at least $\theta$, while the $(\xi+1)$-th and smaller estimated squared singular values fall below $\theta$ (line 13).

Once $\xi$ is determined, we extract the approximate top-$\xi$ singular vectors as a matrix denoted by $\bm{Z}_i$ (line~14). We then save a snapshot containing two matrices $\bm{Z}_i$ and $\bm{Z}_i^\top {\bm{C}_i^\prime}^\top \bm{C}_i^\prime$ along with the current timestamp $i$, and push it into the queue $\mathcal{S}$ (line~15). Next, we remove the subspace spanned by $\bm{Z}_i$ from $\bm{C}_i^\prime$ to obtain $\bm{C}_i$ (line~16). Finally, the for-loop and the update procedure terminate (line~17).

\begin{algorithm}[h]
    \caption{ \sidsfd : \textsc{Query}(lb, ub)}
    \label{alg:pidsfd-query}
    \KwIn{$lb, ub$: start and end timestamp of the query. For Problem~\ref{prob:sw-fd}, $lb=T-N$, $ub=T$.}%
    $\bm{B} \gets \bm{C}_{ub}^\top \bm{C}_{ub}$\;
    \ForAll{$((\bm{Z}_j, \bm{Z}_j^\top {\bm{C}_j^\prime}^\top \bm{C}^\prime_j), s, t=j)$ of $\mathcal{S}$}{
    \If{$lb< t=j\le ub$}{
    $\bm{B} \gets \bm{B} + \boldsymbol{Z}_j{\boldsymbol{Z}_j}^{\top}{\boldsymbol{C}^\prime_j}^{\top}\boldsymbol{C}^\prime_j+{\boldsymbol{C}^\prime_j}^{\top}\boldsymbol{C}^\prime_j\boldsymbol{Z}_j{\bm{Z}_j}^{\top}-\boldsymbol{Z}_j {\bm{Z}_j}^{\top}{\boldsymbol{C}^\prime_j}^{\top}\boldsymbol{C}^\prime_j\boldsymbol{Z}_j\bm{Z}_j^{\top}$ \;
    }
    }
    $[\bm{V}, \bm{\Lambda}, \bm{V}^\top] \gets \mathtt{eigen\_decomposition}(\bm{B})$ \tcp*[f]{$\bm{\Lambda}=\text{diag}\left(\lambda_1, \lambda_2 \dots,\lambda_{d}\right), \lambda_1\ge \lambda_2 \ge \dots \ge \lambda_d$} \;
    $\hat{\bm{B}}\gets\sqrt{\max\left(\bm{\Lambda} - \bm{I}\lambda_\ell, \bm{0}\right)} \, \bm{V}^\top$\;
    \Return $\hat{\bm{B}}[1:\ell, :]$
\end{algorithm}

\subsubsection{Query Algorithm}

Algorithm~\ref{alg:pidsfd-query} shows how to construct the sketch matrix $\hat{\bm{B}}$ for the current sliding window $(T - N,T]$. Since we store snapshots of the form $(\bm{Z}_i, \bm{Z}_i^\top {\bm{C}_i^\prime}^\top \bm{C}_i^\prime)$ for certain timestamps $i$ in the snapshot queue $\mathcal{S}$, we can compute $\bm{Z}_i \bm{Z}_i^\top {\bm{C}_i^\prime}^\top \bm{C}_i^\prime + {\bm{C}_i^\prime}^\top \bm{C}_i^\prime \bm{Z}_i \bm{Z}_i^\top - \bm{Z}_i \bm{Z}_i^\top {\bm{C}_i^\prime}^\top \bm{C}_i^\prime \bm{Z}_i \bm{Z}_i^\top$ by matrix multiplications and additions. We then sum it over all $i$ within the sliding window and add them to the residual matrix $\bm{C}_T^\top \bm{C}_T$ to obtain matrix $\bm{B}$, as shown in lines 1-4.  Finally, we perform an eigendecomposition of $\bm{B}$ (line~5), reduce it by the $\ell$-th largest eigenvalue (line~6), and return the top-$\ell$ rows of the resulting matrix as the sketch matrix $\hat{\bm{B}}\in\mathbb{R}^{\ell\times d}$ (line~7).

\subsubsection{General Unnormalized Model}

\begin{algorithm}[h]
    \caption{ \textsc{ML}-\sidsfd}
    \label{alg:mlsidsfd}
    \SetKwProg{Fn}{Function}{:}{}

    \Fn{\textup{\textsc{Initialize}($d, N, R, \varepsilon$)}}{
        \KwIn{$d$: dimension of the vectors in the stream; $N$: sliding window size; $R$: upper bound of squared norms; $\varepsilon$: desired relative covariance error bound.}
        $L\gets \lceil \log_2 R\rceil; M\gets []$\;
        \For{$i=0$ \KwTo $L-1$}{
            $M.\texttt{append}(\sidsfd.\textsc{Initialize}(d, \varepsilon, \theta=2^i\varepsilon N))$\;
        }
    }

    \Fn{\textup{\textsc{Update}($\bm{a}_i$)}}{
        \KwIn{$\boldsymbol{a}_i$: the row vector arriving at timestamp $i$.}

        \For{$j = 0$ \KwTo $L-1$} {
            \While{$len(M[j].\mathcal{S})> \frac{8}{\varepsilon}$ or $M[j].\mathcal{S}[0].t \leq i-N$} {
                $M[j].\mathcal{S}$.\textsc{Popleft}() \;
            }
            \If{$\left\|\boldsymbol{a}_i\right\|_2^2\geq \theta$} {
                $M[j].\mathcal{S}$.\textsc{append}($\bm{a}_i$, $s=M[j].\mathcal{S}[-1].t$, $t = i$)
            }
            \Else {
                $M[j].\textsc{Update}(\boldsymbol{a}_i)$ \;
            }
        }
    }

    \Fn{\textup{\textsc{Query}()}}{
        $j\gets\min_{j} \{1\le M[j].\mathcal{S}[0].s \le T-N+1\}$\;
        \Return $\mathsf{FD}$($M[j]$.\textsc{Query}($lb=T-N, ub=T$), $[\bm{a}_i \in M[j].\mathcal{S}]$) \;
    }
\end{algorithm}

For the actual implementation of \sidsfd\ for Problem~\ref{prob:sw-fd}, we adopt the \textit{parallel multi-level} technique, as described in Algorithm~\ref{alg:mlsidsfd}. Briefly, we use this technique for two main reasons: (1) As new rows arrive and the sliding window advances, old row vectors (those arriving before time $T - N$) become outdated. Therefore, we first check for and remove expired snapshots from the queue $\mathcal{S}$ (line 7).  (2) We have to set the parameter $\theta = O(1)\cdot\varepsilon\left\|\bm{A}_{T-N,T}\right\|_F^2$ at initialization, but $\left\|\bm{A}_{T-N,T}\right\|_F^2$ can vary dramatically between very small and very large values as the sliding window progresses. To handle this variability, we instantiate multiple parallel \sidsfd\ levels (lines 3-4). This multi-level sketching strategy is widely used in streaming algorithms under the sliding window model~\cite{lee2006simpler,yin2024optimal,xian2025optimal,yao2025optimal}. This technique guarantees that there will always be at least one level $j$ for which the error threshold satisfies $\theta\le \varepsilon \left\|\bm{A}_{T-N,T}\right\|_F^2$. Specifically:

\begin{fact}
    \label{prop:ml}
    There always exists one level $j=\lfloor \log_2 \frac{\|\bm{A}_{T-N,T}\|_F^2}{N}\rfloor$ such that%
    \[
        \theta = 2^j \varepsilon N \le \varepsilon \|\bm{A}_{T-N,T}\|_F^2 \le 2\theta,
    \]
    and we can find the required $j$ using line~14 of Algorithm~\ref{alg:mlsidsfd}, in the same way as Algorithm~7 does in~\cite{yin2024optimal}.
\end{fact}

\subsection{Analysis}

We provide the error bound, asymptotic space and time complexities of Algorithm~\ref{alg:mlsidsfd} in Theorem~\ref{thm:sw}.
\begin{theorem}
    \label{thm:sw}
    ML-\sidsfd (Algorithm~\ref{alg:mlsidsfd}) solves the Tracking Approximate Matrix Sketch over Sliding Window (Problem~\ref{prob:sw-fd}) with space complexity $O\left(\frac{d}{\varepsilon}\log R\right)$ and amortized time complexity $O\left(\frac{d}{\varepsilon}\log d \log R\right)$ per update, with success probability (the probability that the output sketch satisfies the covariance error bound) at least
    \begin{equation}
        \label{eq:prob}
        \frac{99}{100}\left(1- \frac{2}{\pi\sqrt{e}}\cdot\frac{1}{\sqrt{d \log_2 d}}\right).
    \end{equation}

\end{theorem}

\begin{figure}[h]
    \centering
    \includegraphics[width=0.4\linewidth]{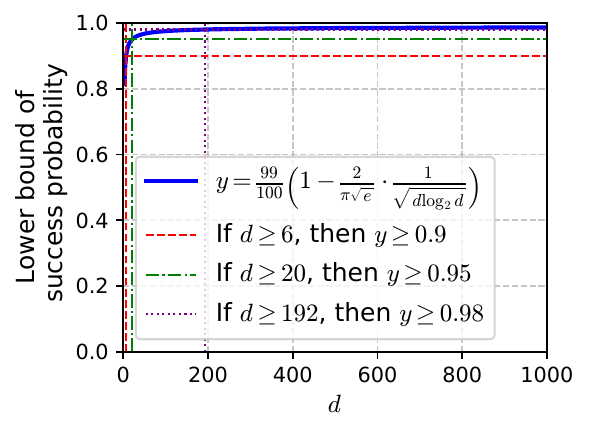}
    \caption{Success probability lower bound with respect to $d$.}
    \label{fig:prob-lb}
\end{figure}

When $d$ is not very small, the success probability lower bound Eq.~\eqref{eq:prob} is close to $\frac{99}{100}$, which is a high probability.\footnote{Note that this is a ``lower bound'' of the success probability, and the actual success probability could be much higher. In fact, in our real experiments, there were no instances where the output sketch failed to satisfy the covariance error bound.} We plot the function graph of the success probability lower bound Eq.~\eqref{eq:prob} with respect to the vector dimension \(d\) in Figure~\ref{fig:prob-lb}. As the dimension $d$ increases, the success probability lower bound of the algorithm quickly approaches \( \frac{99}{100} \). When \( d \geq  6\), the success probability is at least 0.9, and when \( d \geq 20 \) and \( d \geq 192 \), the success probability is at least 0.95 and 0.98, respectively. This indicates that when \( d \) is not very small, the success probability lower bound is close to \( \frac{99}{100} \).

\subsubsection{Reducing the Failure Probability from Eq.~\eqref{eq:prob} to Arbitrarily Small $\delta$}

Furthermore, with a simple modification to Algorithm~\ref{alg:sidsfd-update}, we can boost the success probability of the algorithm from at most $99/100$ in Eq.~\eqref{eq:prob} to $1-\delta$, where $\delta$ can be made arbitrarily close to~0. At the same time, the time complexity only increases by a multiplicative factor of $\log^2(1/\delta)$ compared with the result in Theorem~\ref{thm:sw}, i.e., it increases from $O\left(\frac{d}{\varepsilon}\log d \log R\right)$ to $O\left(\frac{d}{\varepsilon}\log d \log R\log^2(1/\delta)\right)$, while the space complexity remains the same as in Theorem~\ref{thm:sw}.

More specifically, this probability-amplification technique replaces line~11 of Algorithm~\ref{alg:sidsfd-update} with Algorithm~\ref{alg:amplify}. The key idea is to replace the original single Simultaneous Iteration with $r=\log(1/\delta)$ independent Simultaneous Iterations, and then select the iteration whose result minimizes the matrix spectral norm of the $\bm{C}_i^\prime-\bm{C}_i^\prime\bm{Z}\bm{Z}^\top$. Formally, select $Z_{r^*}$ where
\[
    r^* = \operatorname*{argmin}_{h \in [r]}
    \left\| \bm{C}_i^\prime - \bm{C}_i^\prime \bm{Z}_h \bm{Z}_h^\top \right\|_2.
\]

However, there is a pitfall here: computing the matrix spectral norm $\left\| \bm{C}_i^\prime - \bm{C}_i^\prime \bm{Z}_h \bm{Z}_h^\top \right\|_2$ exactly in each round incurs significant time overhead. Therefore, instead of computing the exact spectral norm, we repeatedly run $s=2 \log_3 \frac{2}{\delta}$ independent rounds of Power Iteration on matrix $\bm{C}_i^\prime - \bm{C}_i^\prime \bm{Z}_h \bm{Z}_h^\top$ to obtain $s$ estimates (lines~6-7), take the maximum among them as the estimate of $\left\| \bm{C}_i^\prime - \bm{C}_i^\prime \bm{Z}_h \bm{Z}_h^\top \right\|_2$ (line~8), and finally choose the round with the smallest estimated value among all $r$ rounds as the final result (line~9). It can be shown that, after this modification, we obtain the following corollary:
\begin{corollary}
    \label{coro:amplify}
    If Algorithm~\ref{alg:amplify} is applied, then ML-\sidsfd\ (Algorithm~\ref{alg:mlsidsfd}) solves the Tracking Approximate Matrix Sketch over Sliding Window (Problem~\ref{prob:sw-fd}) with space complexity $O\!\left(\frac{d}{\varepsilon}\log R\right)$ and amortized time complexity $O\left(\frac{d}{\varepsilon}\log d \log R \log^2 \frac{1}{\delta}\right)$ per update, with success probability at least $1-\delta$, where $\delta$ is a tunable parameter that can be chosen arbitrarily in the interval $(0, 1/100)$.
\end{corollary}

The detailed proof can be found in Appendix~\ref{app:amplified-estimation}.

\begin{algorithm}[h]
    \caption{Probability Amplification Procedure (Replacing Line~11 in Algorithm~\ref{alg:sidsfd-update})}
    \label{alg:amplify}
    $c \gets +\infty, \quad r\gets \log_{100} \frac{2}{\delta}, \quad s\gets 2 \log_3 \frac{2}{\delta} $\;
    \For{$h=1$ \KwTo $r$}{
        $[\bm{Z}_{h}^\prime, \hat{\bm{\Sigma}}_{h}] \gets \mathtt{simul\_iter}({\bm{C}^\prime_i}^\top , k, \varepsilon_{\text{SI}}=0.2)$ \tcp*[f]{$\hat{\bm{\Sigma}}_{h}=\text{diag}\left(\hat{\sigma}_{h1}, \hat{\sigma}_{h2} \dots,\hat{\sigma}_{hk}\right)$}\;%
        $\xi_h \gets \max_{\xi} \{\hat{\sigma}_{h\xi}^2 \ge \theta\}, \quad \bm{Z}_h\gets \bm{Z}_h^\prime[:,1:\xi]$\;
        $c_h \gets 0$\;
        \For{$g=1$ \KwTo $s$}{
            $c_{hg}^2 \gets \texttt{power\_iteration}(\bm{C}_i^\prime-\bm{C}_i^\prime\bm{Z}_h\bm{Z}_h^\top)$ \;
            $c_h \gets \max(c_h, c_{hg})$
        }
        \lIf{$c_h<c$}{$c \gets c_h,  \quad \bm{Z}_i^\prime \gets \bm{Z}_{h}^\prime$}
    }
\end{algorithm}

\subsubsection{A proof sketch of Theorem~\ref{thm:sw}}

For the complete proof of Theorem~\ref{thm:sw}, we refer the reader to Appendix~\ref{sec:proof-sw}. Here, we provide a concise proof sketch that contains the key technical ideas and intuition.

\textbf{(Error bound)} If $\left\|\boldsymbol{a}_i\right\|_2^2\geq \theta$, then at line~10 in Algorithm~\ref{alg:mlsidsfd}, $\bm{a}_i$ is saved accurately, and thus no error is introduced. Otherwise, if line 12 is triggered, then Algorithm~\ref{alg:sidsfd-update} is executed. Since line 15 of Algorithm~\ref{alg:sidsfd-update} stores matrix $\bm{Z}_i$ and $\bm{Z}_i^\top {\bm{C}^\prime_i}^\top {\bm{C}^\prime_i}$, we can accurately restore the matrix $\bm{C}^\prime_i$ at line 4/6 from the matrix $\bm{C}_i$ at line 16 through the following computation:
\begin{equation*}
    \label{eq:restore}
    \begin{split}
                & \bm{C}_i^\top \bm{C}_i  + \underline{\bm{Z}_i}\cdot \underline{\bm{Z}_i^\top {\bm{C}_i^\prime}^\top \bm{C}_i^\prime} + (\underline{\bm{Z}_i^\top {\bm{C}_i^\prime}^\top \bm{C}_i^\prime})^\top  \underline{\bm{Z}_i}^\top - \underline{\bm{Z}_i}\cdot \underline{\bm{Z}_i^\top {\bm{C}_i^\prime}^\top \bm{C}_i^\prime}\cdot \underline{\bm{Z}_i}\cdot\underline{ \bm{Z}_i}^\top \\
        {{}={}} & (\bm{C}_i^\prime - \bm{C}_i^\prime \bm{Z}_i \bm{Z}_i^\top )^\top (\bm{C}_i^\prime - \bm{C}_i^\prime \bm{Z}_i \bm{Z}_i^\top ) +\bm{Z}_i \bm{Z}_i^\top {\bm{C}_i^\prime}^\top \bm{C}_i^\prime + {\bm{C}_i^\prime}^\top \bm{C}_i^\prime \bm{Z}_i \bm{Z}_i^\top - \bm{Z}_i \bm{Z}_i^\top {\bm{C}_i^\prime}^\top \bm{C}_i^\prime \bm{Z}_i \bm{Z}_i^\top                                                                                                                                                           \\
        {{}={}} & {\bm{C}_i^\prime}^\top \bm{C}_i^\prime\cancel{-\bm{Z}_i \bm{Z}_i^\top {\bm{C}_i^\prime}^\top \bm{C}_i^\prime} \cancel{- {\bm{C}_i^\prime}^\top \bm{C}_i^\prime \bm{Z}_i \bm{Z}_i^\top}  \cancel{+\bm{Z}_i \bm{Z}_i^\top {\bm{C}_i^\prime}^\top \bm{C}_i^\prime \bm{Z}_i \bm{Z}_i^\top}                                                                                          \\
                & \cancel{+\bm{Z}_i \bm{Z}_i^\top {\bm{C}_i^\prime}^\top \bm{C}_i^\prime } \cancel{+{\bm{C}_i^\prime}^\top \bm{C}_i^\prime \bm{Z}_i \bm{Z}_i^\top}  \cancel{-\bm{Z}_i \bm{Z}_i^\top {\bm{C}_i^\prime}^\top \bm{C}_i^\prime \bm{Z}_i \bm{Z}_i^\top } = {\bm{C}_i^\prime}^\top \bm{C}_i^\prime .
    \end{split}
\end{equation*}
This can be regarded as the ``inverse operation'' of Algorithm~\ref{alg:sidsfd-update} and used in the query procedure of Algorithm~\ref{alg:pidsfd-query}. Therefore, lines 7-20 do not introduce the restoring norm error. The overall error mainly comes from the FD reduction operation in line~4. Additionally, by performing $N$ steps of the ``inverse operation'' from the current time $T$, we can recover the matrix $\bm{C}^\prime_{T-N}$ and thus $\bm{C}_{T-N}$ at time $T-N$. We can prove that the covariance error can be decomposed into two parts:
\begin{equation*}
    \left\|\bm{A}_{\text{W}}^\top \bm{A}_{\text{W}}-\bm{B}^\top \bm{B}\right\|_2   \leq \underbrace{\left\|\sum_{i=T-N+1}^{T} \bm{\Delta}_i\right\|_2 }_{\text{Term 1}} + \underbrace{\left\|\bm{C}_{T-N}^\top \bm{C}_{T-N}\right\|_2}_{\text{Term 2}}. \\
\end{equation*}
Among them, $\bm{\Delta}_i=\tilde{\bm{C}_{i}}^\top \tilde{\bm{C}_{i}} - {\bm{C}^\prime_i}^\top {\bm{C}^\prime_i}$, the gap error before and after line~4 in Algorithm~\ref{alg:sidsfd-update}. Therefore, Term~1 is the error that originates from the FD reduction, and Term~2 is the error from $\bm{C}_{T-N}$. The derivation of the upper bound for Term~1 is similar to the idea of FD~\cite{ghashami2016frequent}, and we can conclude that:
\begin{equation*}
    \text{Term 1} \le O(1)\varepsilon \left\|\bm{A}_{T-N,T}\right\|_F^2  + 2\cdot \text{Term 2}.
\end{equation*}

For the upper bound of Term 2 to hold, both the Power Iteration (line 7) and the Simultaneous Iteration (line 11) at time $T-N$ need to have converged successfully. According to Corollary~\ref{thm:power} and Theorem~\ref{thm:simul}, the probability is given by Eq.~\eqref{eq:prob}. Under the condition of successful convergence and based on Fact~\ref{prop:ml}, we can always find one level $j$ of Algorithm~\ref{alg:mlsidsfd} that satisfies:
\begin{equation*}
    \text{Term 2} \le O(1)\theta\le O(1)\varepsilon \|\bm{A}_{T-N,T}\|_F^2.
\end{equation*}
Finally, we have the probabilistic upper bound of covariance error:
\begin{equation*}
    \left\|\bm{A}_{T-N,T}^\top \bm{A}_{T-N,T}-\bm{B}^\top \bm{B}\right\|_2 \le \text{Term 1}+\text{Term 2} \le O(1) \varepsilon \|\bm{A}_{T-N,T}\|_F^2.
\end{equation*}

\textbf{(Time complexity)} In Algorithm~\ref{alg:mlsidsfd}, the update function executes the update subroutines of Algorithm~\ref{alg:sidsfd-update} at most $L=\lceil \log_2 R\rceil$ times. The time cost of Algorithm~\ref{alg:sidsfd-update} consists of 3 parts: the FD reduction in lines 2-6, the Power Iteration in line 7, and the doubling Simultaneous Iteration in lines 8-17. The SVD for FD reduction is executed once every $\ell$ updates, and the amortized time cost of lines 2-6 is $O(d\ell^2)/\ell = O(d\ell)$. According to the time complexity result in Corollary~\ref{thm:power}, the Power Iteration in line 7 takes $O(d\ell k) = O(\frac{d}{\varepsilon} \log d)$ time.

The analysis of the time cost for lines 9 to 17 is more subtle. We need to analyze each value of $k$ when the for-loop terminates during the $i$-th update, which we denote as $k_i$. $k_i$ is an integer, where intuitively each 1 in $k_i$ corresponds to a component greater than $\theta$ in $\left\|\bm{A}_{T-N,T} \right\|_F^2$. Therefore, for level $j$ given by Fact 1, we have set $\theta=\varepsilon\left\|\bm{A}_{T-N, T}\right\|_F^2$, and a constraint on the sum of $k_i$ within the current sliding window holds:
\begin{equation}
    \label{eq:sum-of-ki}
    \begin{split}
        \sum_{i=T-N+1}^{T} k_i = O\left(\frac{\left\|\bm{A}_{T-N, T}\right\|_F^2}{\theta}\right) = O\left(\frac{1}{\varepsilon}\right).
    \end{split}
\end{equation}
According to Theorem~\ref{thm:simul}, the time cost of the Simultaneous Iteration corresponding to each  $k_i$  is \( O(d\ell k_i \log d) \). Therefore, the total time cost for all  $i$  is:
\begin{equation*}
    \sum_{i=T-N+1}^{T} O(d\ell k_i \log d) = O(d\ell \log d) \sum_{i=T-N+1}^{T} k_i = O\left(\frac{d\ell}{\varepsilon}  \log d\right).
\end{equation*}

We assume $N > \ell$ (otherwise, if $N < \ell$, we can directly store all $\ell$ rows exactly, so the assumption $N \ge \ell$ is without loss of generality). Amortizing over $N$ updates, the per-update time cost of Algorithm~\ref{alg:sidsfd-update} of level $j$ becomes:
\begin{equation*}
    \frac{1}{N} O\left(\frac{d\ell}{\varepsilon} \log d\right) = O\left(\frac{d}{\varepsilon} \log d\right).
\end{equation*}

We show that the update time complexity of level $j=\lfloor \log_2 \frac{\|\bm{A}_{T-N,T}\|_F^2}{N}\rfloor$ in Algorithm~\ref{alg:mlsidsfd} given by Fact~\ref{prop:ml} is the maximum over all $L$ levels. The intuition is that for lower levels than $j$, the $\theta$ decreases exponentially, so the time-efficient line~10 can be triggered more frequently. For levels higher than $j$, the $\theta$ increases exponentially, so the sum of $k_i$ in Eq.~\eqref{eq:sum-of-ki} decreases exponentially. Therefore, the update time cost is dominated by the level $j$, and the total time cost of all $L=\lceil \log_2 R \rceil$ levels is $O\left(\frac{d}{\varepsilon}\log d \log R\right)$.

\textbf{(Space complexity)} The space cost of each level in Algorithm~\ref{alg:mlsidsfd} consists of the sketch matrix $\bm{C}$ and the snapshots in $\mathcal{S}$. The sketch matrix $\bm{C}$ requires $O(d\ell)$ space, as its maximum size is $d \times 2\ell$. The queue $\mathcal{S}$ also uses $O(d/\varepsilon)$ space, since it stores matrices $\bm{Z}_i \in \mathbb{R}^{d \times k_i}$ and $\bm{Z}_i^\top {\bm{C}_i^\prime}^\top \bm{C}_i^\prime \in \mathbb{R}^{k_i \times d}$. Summing over all $N$ updates and using the bound from Eq.~\eqref{eq:sum-of-ki}, the space cost of the queue $\mathcal{S}$ is:
\begin{equation*}
    O\left(d \sum_{i=T-N+1}^{T} k_i\right) = O\left(\frac{d}{\varepsilon}\right).
\end{equation*}

Thus, the space cost of each level is $O\left(d\ell + \frac{d}{\varepsilon}\right) = O\left(\frac{d}{\varepsilon}\right)$, and the total space cost of all $L=\lceil \log_2 R \rceil$ levels is $O\left(\frac{d}{\varepsilon} \log R\right)$.

\section{Applying \sidsfd to Other Problems}
\label{sec:applications}

Now we show how \sidsfd\ serves as a general framework that can be plugged into various algorithms to replace their time-consuming matrix factorizations for a broad class of matrix sketching problems.

\subsection{Persistent Matrix Sketch (Problem~\ref{prob:attp})}

Algorithm~\ref{alg:attp} describes how \sidsfd\ can be used to solve the ATTP matrix sketch problem. The high-level idea is similar to FD-ATTP (Algorithm~1 in ~\cite{shi2021time}), but we replace the SVD calculation with \sidsfd to avoid frequent and time-consuming matrix factorizations. %

\begin{algorithm}
    \caption{ATTP Sketch using \sidsfd Framework}
    \label{alg:attp}
    \SetKwProg{Fn}{Function}{:}{}

    \Fn{\textup{\textsc{Initialize}($d, \varepsilon$)}}{
        $\left\|\bm{A}_0\right\|_F^2 \leftarrow 0$\;
        \sidsfd.\textsc{Initialize}($d, \varepsilon, \theta = 0$)
    }

    \Fn{\textup{\textsc{Update}($\bm{a}_i$)}}{
        $\left\|\bm{A}_t\right\|_F^2 \gets \left\|\bm{A}_{t-1}\right\|_F^2 + \left\|\bm{a}_i\right\|_2^2$\;
        \sidsfd.$\theta\gets \varepsilon \left\|\bm{A}_t\right\|_F^2$\;
        \sidsfd.\textsc{Update}($\bm{a}_i$)\;
    }

    \Fn{\textup{\textsc{Query}($t$)}}{
        \Return \sidsfd.\textsc{Query}(0, $t$)\;
    }
\end{algorithm}

\begin{theorem}
    \label{thm:attp}
    Algorithm~\ref{alg:attp} solves the ATTP Matrix Sketch (Problem~\ref{prob:attp}) with space complexity $O\!\left(\frac{d}{\varepsilon}\log \left\|\bm{A}\right\|_F^2\right)$ and amortized time complexity $O\!\left(\frac{d}{\varepsilon}\log d \right)$ per update, with success probability at least that given in Eq.~\eqref{eq:prob}. Alternatively, using the probability amplification of Algorithm~\ref{alg:amplify}, it achieves amortized time complexity $O\!\left(\frac{d}{\varepsilon}\log d \log^2 \frac{1}{\delta}\right)$ per update while guarantee success probability at least $1-\delta$, where $\delta$ is a tunable parameter that can be chosen arbitrarily in the interval $(0, 1/100)$.
\end{theorem}

The detailed proof can be found in Appendix~\ref{app:attp}.

\subsection{Tracking Approximate Matrix Sketch over Sliding Window (Problem~\ref{prob:sw-amm})}

\begin{algorithm}
    \caption{\sidscod Framework for Problem~\ref{prob:sw-amm}}
    \label{alg:sw-amm}
    \SetKwProg{Fn}{Function}{:}{}

    \Fn{\textup{\textsc{Initialize}($d_x, d_y, \varepsilon, N, \theta$)}}{
        \KwIn{$d_x, d_y$: dimension of the vectors in stream; $N$: sliding window size;  $\varepsilon$: relative covariance error bound;  %
            $\theta$: dump threshold. %
        }
        \(\ell\gets\lceil\frac{2}{\varepsilon}\rceil,\quad \hat{\boldsymbol{A}} \leftarrow \boldsymbol{0}^{d_x\times l},\quad \hat{\boldsymbol{B}} \leftarrow \boldsymbol{0}^{d_y\times l} \) \;
        \(\mathcal{S} \leftarrow\) an empty queue \;
    }

    \Fn{\textup{\textsc{Update}($\bm{x}_i, \bm{y}_i$)}}{
    \KwIn{$\boldsymbol{x}_i$: the column vector of $\bm{X}$ arriving at time $i$; $\boldsymbol{y}_i$: the column vector of $\bm{Y}$ arriving at time $i$;}

    \While(\tcp*[f]{snapshot expired}){$\mathcal{S}[0].t + N \leq i$}{
        $\mathcal{S}.\textsc{popleft}()$\\
    }
    $\tilde{\bm{A}} \leftarrow \begin{bmatrix} \bm{A}_{i-1}^\top & \bm{a}_i \end{bmatrix}^\top, \tilde{\bm{B}} \leftarrow \begin{bmatrix} \bm{B}_{i-1}^\top & \bm{b}_i \end{bmatrix}^\top$\\

    \eIf(\tcp*[f]{$\tilde{\bm{A}}$ has no zero rows}){$\text{rows}(\tilde{\bm{A}}) \geq 2\ell$}{
        $\bm{A}^\prime_i,\bm{B}^\prime_i\leftarrow \textsf{COD}(\tilde{\bm{A}}, \tilde{\bm{B}}, \ell)$    \tcp*[f]{Co-Occuring Directions Sketch, refer to~\cite{mroueh2017co}}
    }{
        $\bm{A}^\prime_i,\bm{B}^\prime_i\leftarrow \tilde{\bm{A}},\tilde{\bm{B}}$
    }

    $\hat{\sigma}_1^2, \hat{\bm{v}}_1 \gets \mathtt{power\_iteration}(\bm{A}^\prime_i{\bm{B}^\prime_i}^\top, k=O(\log_2 \ell))$

    \eIf(\tcp*[f]{largest singular value}){$\hat{\sigma}_1 \geq \frac{\theta}{2}$}{
        \For{$j=1$ \KwTo $\lceil \log_2 \ell \rceil$}{
            $k \gets \min(2^j, \ell) $\\
            $[\bm{Z}^\prime, \hat{\bm{\Sigma}}] \gets \mathtt{simul\_iter}(\bm{A}^\prime_i {\bm{B}^\prime_i}^\top, k, \varepsilon_{\text{SI}})$ \\ %
            $[\bm{H}^\prime, \hat{\bm{\Sigma}}] \gets \mathtt{simul\_iter}({\bm{B}^\prime_i \bm{A}^\prime_i }^\top, k, \varepsilon_{\text{SI}})$ \\ %
            \If{$\hat{\bm{\Sigma}}_k<\theta$}{
                $k^\prime \gets \arg \max_{k^\prime} \{\hat{\bm{\Sigma}}_{k^\prime}\ge\theta\}$\\
                $\bm{Z}_i \gets \bm{Z}^\prime{[:,1:k^\prime]}, \quad \bm{H}_i\gets \bm{H}^\prime{[:,1:k^\prime]}$\\
                $\texttt{snapshot}\gets(\bm{Z}_i,\bm{Z}_i^\top \bm{A}^\prime_i {\bm{B}^\prime_i}^\top, \bm{A}^\prime_i {\bm{B}^\prime_i}^\top \bm{H}_i^\top, \bm{H}_i)$\\
                $\mathcal{S}.\textsc{append}\left(\texttt{snapshot}, s=\mathcal{S}[-1].t+1, t=i\right)$\\
                $\bm{A}_{i} \gets \left(\bm{I}-\bm{Z}\bm{Z}^\top\right)\bm{A}_i^\prime, \ \bm{B}_{i} \gets \left(\bm{I}-\bm{H}\bm{H}^\top\right)\bm{B}_i^\prime$\\
                break\;
            }
        }
    }{
        $\bm{A}_i \gets \bm{A}_i^\prime, \  \bm{B}_i \gets \bm{B}_i^\prime$
    }
    }

    \Fn{\textup{\textsc{Query}()}}{
        $\bm{A}\bm{B}^\top \gets \bm{A}_T \bm{B}_T^\top$\\
        \ForAll{$(\bm{Z}_i,\bm{Z}_i^\top \bm{A}^\prime_i {\bm{B}^\prime_i}^\top, \bm{A}^\prime_i {\bm{B}^\prime_i}^\top \bm{H}_i^\top, \bm{H}_i)$ in $\mathcal{S}$}{
            $\bm{A}\bm{B}^\top \gets \bm{A}\bm{B}^\top + \bm{Z}_i\bm{Z}_i^\top \bm{A}^\prime_i {\bm{B}^\prime_i}^\top+\bm{A}^\prime_i {\bm{B}^\prime_i}^\top \bm{H}_i^\top \bm{H}_i - \bm{Z}_i\bm{Z}_i^\top \bm{A}^\prime_i {\bm{B}^\prime_i}^\top\bm{H}_i^\top \bm{H}_i $
        }
        $\left[\bm{U}, \bm{\Sigma},\bm{V}^\top\right]\gets \text{SVD}(\bm{A}\bm{B}^\top)$\;
        ${\bm{A}}^*\gets \bm{U}\sqrt{\max\left(\bm{\Sigma}-\sigma_\ell \bm{I},\bm{0}\right)},\quad{\bm{B}}^*\gets\bm{V}\sqrt{\max\left(\bm{\Sigma}-\sigma_\ell \bm{I},\bm{0}\right)}$\;
            \Return ${\bm{A}}^*[:, 1:\ell],\quad{\bm{B}}^*[:,1:\ell]$
    }
\end{algorithm}

Algorithm~\ref{alg:sw-amm} shows the sketching algorithm for tracking an approximate matrix sketch over sliding windows. The high-level idea and algorithm skeleton are similar to hDS-COD (Algorithm~3 in~\cite{yao2025optimal}) and SO-COD (Algorithm~2 in~\cite{xian2025optimal}), but we replace the DS-COD subroutines with \sidscod to avoid frequent and time-consuming matrix factorizations.

\begin{theorem}
    \label{thm:sw-amm}
    Algorithm~\ref{alg:sw-amm} solves the Tracking Approximate Matrix Multiplication over Sliding Window (Problem~\ref{prob:sw-amm}) with space complexity $O\left(\frac{d_x+d_y}{\varepsilon}\log R\right)$ and amortized time complexity $O\left(\frac{d_x+d_y}{\varepsilon}\log d \log R \right)$ per update, with success probability at least that given in Eq.~\eqref{eq:prob}. Alternatively, using the similar like probability amplification of Algorithm~\ref{alg:amplify}, it achieves amortized time complexity $O\!\left(\frac{d_x+d_y}{\varepsilon}\log d \log^2 \frac{1}{\delta}\right)$ per update while guarantee success probability at least $1-\delta$, where $\delta$ is a tunable parameter that can be chosen arbitrarily in the interval $(0, 1/100)$.
\end{theorem}

The detailed proof can be found in Appendix~\ref{app:sw-amm}.

\subsection{Distributed Matrix Sketch (Problem~\ref{prob:dsfd})}

Algorithm~\ref{alg:dsfd} describes how \sidsfd helps to solve the Tracking Distributed Matrix Sketch problem. The high-level idea and algorithm skeleton are similar to P2 (Algorithms~5.3 and~5.4 in ~\cite{ghashami2014continuous}), although we replace the SVD with \sidsfd to avoid frequent and time-consuming matrix factorizations.

\begin{algorithm}
    \caption{\pfive : \sidsfd Framework for Problem~\ref{prob:dsfd}}
    \label{alg:dsfd}
    \SetKwProg{Proc}{Procedure}{:}{}

    \Proc(\tcp*[f]{at Site $j$}){\textup{\textsc{Update}($\bm{a}_i$)}}{
        $F_j \gets F_j + \left\|\bm{a}_i\right\|_2^2$\\
        \If{$F_j \ge \frac{\varepsilon}{m}\hat{F}$}{
            Send $F_j$ to coordinator; set $F_j \gets 0$\\
        }
        \sidsfd.\textsc{Update}($\bm{a}_i$)\\
        \If{\sidsfd.$\mathcal{S}$ is not empty}{
            Send $\bm{Z}, \bm{Z}^\top \bm{C}^\top \bm{C}$ of $\mathcal{S}[0]$ to coordinator\\
            Delete $\mathcal{S}[0]$ then \sidsfd.$\mathcal{S}$ is empty\\
            \Return $\bm{Z}, \bm{Z}^\top \bm{C}^\top \bm{C}$
        }
        \uCase{received $\hat{F}$}{
            \sidsfd.$\theta\gets \varepsilon \hat{F}$
        }
    }

    \Proc(\tcp*[f]{at Coordinator}){\textup{\textsc{Receive\_Update}}}{
        \uCase{received $F_j$}{
            $\hat{F} \mathrel{+}= F_j$ and $\#\texttt{msg}\mathrel{+}=1$ \\
            \If{$\#\texttt{msg} \ge m$}{
                $\#\texttt{msg}\gets 0$, then Broadcast $\hat{F}$ to all sites
            }
        }
        \uCase{received $\bm{Z}$ and $\bm{Z}^\top\bm{C}^\top\bm{C}$}{
            $\bm{B} \mathrel{+}= \bm{Z}\bm{Z}^\top\bm{C}^\top\bm{C}+\bm{C}^\top \bm{C}\bm{Z}\bm{Z}^\top - \bm{Z}\bm{Z}^\top\bm{C}^\top\bm{C}\bm{Z}\bm{Z}^\top$\\
        }
    }

    \Proc(\tcp*[f]{at Coordinator}){\textup{\textsc{Query}}}{
        $[\bm{V}, \bm{\Lambda}, \bm{V}^\top] \gets \mathtt{eigen\_decomposition}(\bm{B})$ \; %
    $\hat{\bm{B}}\gets\sqrt{\max\left(\bm{\Lambda} - \bm{I}\lambda_\ell, \bm{0}\right)} \, \bm{V}^\top$\;
    \Return $\hat{\bm{B}}[1:\ell, :]$
    }

\end{algorithm}

\begin{theorem}
    \label{thm:dist}
    \pfive (Algorithm~\ref{alg:dsfd}) solves the Distributed Matrix Sketch problem  (Problem~\ref{prob:dsfd}) with communication complexity $O\left(\frac{md}{\varepsilon}\log \left\|\bm{A}\right\|_F^2\right)$ and amortized time complexity $O\left(\frac{d}{\varepsilon}\log d \right)$ per update, with success probability at least that given in Eq.~\eqref{eq:prob}. Alternatively, using the probability amplification of Algorithm~\ref{alg:amplify}, it achieves amortized time complexity $O\!\left(\frac{d}{\varepsilon}\log d \log^2 \frac{1}{\delta}\right)$ per update while guarantee success probability at least $1-\delta$, where $\delta$ is a tunable parameter that can be chosen in the interval $(0, 1/100)$.
\end{theorem}

The detailed proof can be found in Appendix~\ref{app:dist}.

\subsubsection{Tracking Distributed Matrix Sketch over Sliding Window (Problem~\ref{prob:dswfd})}

Algorithm~\ref{alg:dswfd} describes how \pfive (Algorithm~\ref{alg:dsfd}) can be used to solve the Tracking Distributed Matrix Sketch over Sliding Window problem. The high-level idea and algorithm skeleton are similar to DA2 (Algorithm~5 in~\cite{zhang2017tracking}). The IWMT protocol used by DA2 is a ``black-box'' one-way communication version of P2. Here we replace IWMT with \pfive (Algorithm~\ref{alg:dsfd}) to achieve the same communication complexity while avoiding frequent and time-consuming matrix factorizations.

\begin{algorithm}
    \caption{\sidsfd Framework for Problem~\ref{prob:dswfd}}
    \label{alg:dswfd}
    \SetKwProg{Proc}{Procedure}{:}{}

    \Proc(\tcp*[f]{at Site $j$}){\textup{\textsc{Update}($\bm{a}_i$)}}{
        $\bm{Z}, \bm{Z}^\top\bm{C}^\top\bm{C} \gets$\pfive.\textsc{Update}($\bm{a}_i$)\\
        Send $(\bm{Z}, \bm{Z}^\top \bm{C}^\top \bm{C}, \texttt{flag}=\texttt{update})$\\
        $mEH_a.\textsc{Update}(\bm{a}_i)$;\\
        \If{$i \mod N\equiv 0$}{
            $Q \gets IWMT_c(mEH_a)$
        }
        \For{$\bm{e}_t\in Q$ and $t<i-N$}{
            Send $(\bm{e}_t, \texttt{flag}=\texttt{expire})$, then remove $\bm{e}_t$ from $Q$
        }
    }

    \Proc(\tcp*[f]{at Coordinator}){\textup{\textsc{Receive\_Update}}}{
        \Switch{\texttt{flag}}{
            \uCase{\texttt{update}}{
                $\bm{B} \mathrel{+}= \bm{Z}\bm{Z}^\top\bm{C}^\top\bm{C}+\bm{C}^\top \bm{C}\bm{Z}\bm{Z}^\top - \bm{Z}\bm{Z}^\top\bm{C}^\top\bm{C}\bm{Z}\bm{Z}^\top$\\
            }
            \uCase{\texttt{expire}}{
                $\bm{B} \mathrel{-}= \bm{e}_t^\top\bm{e}_t$
            }
        }
    }

    \Proc(\tcp*[f]{at Coordinator}){\textup{\textsc{Query}}}{
         $[\bm{V}, \bm{\Lambda}, \bm{V}^\top] \gets \mathtt{eigen\_decomposition}(\bm{B})$ \; %
    $\hat{\bm{B}}\gets\sqrt{\max\left(\bm{\Lambda} - \bm{I}\lambda_\ell, \bm{0}\right)} \, \bm{V}^\top$\;
    \Return $\hat{\bm{B}}[1:\ell, :]$
    }
\end{algorithm}

\begin{theorem}
    \label{thm:dist-sw}
    Algorithm~\ref{alg:dswfd} solves the Tracking Distributed Matrix Sketch over Sliding Window (Problem~\ref{prob:dswfd}) with communication complexity $O\left(\frac{md}{\varepsilon}\log \left\|\bm{A}\right\|_F^2\right)$ and amortized time complexity $O\left(\frac{d}{\varepsilon}\log d \right)$ per update, with success probability at least that given in Eq.~\eqref{eq:prob}. Alternatively, using the probability amplification of Algorithm~\ref{alg:amplify}, it achieves amortized time complexity $O\!\left(\frac{d}{\varepsilon}\log d \log^2 \frac{1}{\delta}\right)$ per update to guarantee success probability at least $1-\delta$, where $\delta$ is a tunable parameter that can be chosen arbitrarily in the interval $(0, 1/100)$.
\end{theorem}

The detailed proof can be found in Appendix~\ref{app:dist-sw}.

\section{Experiments}

\subsection{Experimental Setup}

\subsubsection{Problem Scenarios, Algorithms and Parameter Settings}

We evaluate our proposed framework \sidsfd against the space-/communication-optimal baseline algorithms listed in Table~\ref{tab:comparison} for various tracking matrix sketching problems:

\begin{itemize}[leftmargin=*]
    \item \textbf{Matrix Sketch over \underline{Sliding Window} (Problem~\ref{prob:sw-fd})}: Fast-DS-FD~\cite{yin2024optimal}.
    \item \textbf{\underline{ATTP} Matrix Sketch (Problem~\ref{prob:attp})}: PFD~\cite{shi2021time}.
    \item \textbf{Tracking Approximate Matrix Multiplication over Sliding Window (\underline{AMM}, Problem~\ref{prob:sw-amm})}: hDS-COD/SO-COD~\cite{xian2025optimal,yao2025optimal}.
    \item \textbf{Tracking \underline{Distributed} Matrix Sketch (Problem~\ref{prob:dsfd})}: P2~\cite{ghashami2014continuous}.
    \item \textbf{Distributed Matrix Sketch over Sliding Window (\underline{Distributed SW}, Problem~\ref{prob:dswfd})}: DA2~\cite{zhang2017tracking}.
\end{itemize}

The \textbf{\underline{bold-underlined}} terms are the abbreviations used in the experimental result figures. For all algorithms, common parameters are the dimension of the vector stream, denoted as $d$, and the upper bound on the relative covariance error, $\varepsilon$. Our experiments primarily focus on comparing various performance metrics of \sidsfd against baseline algorithms under varying values of $\varepsilon$. The specific metrics we focus on will be elaborated in Section~\ref{sec:metrics}.

\subsubsection{Metrics}
\label{sec:metrics}

\begin{itemize}[leftmargin=*]
    \item \textbf{Maximum sketch size} assesses the memory cost of a matrix sketch algorithm. For each algorithm, we record the peak total memory usage of all data structures (in number of rows or KB) during runtime over the entire stream.

    \item \textbf{Empirical relative covariance error} assesses the quality of the approximate sketch matrix produced by algorithms. It is defined as $\lVert \bm{A}^\top \bm{A} - \bm{B}^\top \bm{B}\rVert_2 / \lVert \bm{A} \rVert_F^2$, where $\bm{A}$ is the ground-truth matrix and $\bm{B}$ is the sketch matrix generated by the algorithm. In the case of approximate matrix multiplication, the relative error is defined as $\lVert \bm{X} \bm{Y}^\top - \bm{A} \bm{B}^\top \rVert_2 / (\lVert \bm{X} \rVert_F \lVert \bm{Y} \rVert_F)$, where $\bm{X}$ and $\bm{Y}$ are the input stream matrices, while $\bm{A}$ and $\bm{B}$ are the resulting sketches. It is important to note that the parameter $\varepsilon$ is an upper bound on the relative covariance error (i.e., the worst-case approximation guarantee). In practice, the actual relative covariance error observed during experiments may be significantly lower than $\varepsilon$. Even under the same $\varepsilon$ setting, different algorithms may exhibit different empirically measured covariance errors.

    \item \textbf{Amortized update time} refers to the average time to update for each input vector from the streams during the runtime.

    \item \textbf{Speedup ratio} measures how much faster \sidsfd is compared with the baseline. It is calculated as the ratio of the amortized update time of the baseline algorithm to the amortized update time of the \sidsfd algorithm.

    \item \textbf{Communication cost}: In distributed settings, this metric captures the total size of messages exchanged between sites and the coordinator during the execution of the sketch algorithm.

    \item \textbf{Space/communication cost ratio} compares the space/communication footprint of \sidsfd algorithm with the baseline algorithm, calculated as the ratio of space/communication cost of our \sidsfd algorithms to that of the baseline algorithms.
\end{itemize}

\subsubsection{Hardware, Environments, and Datasets}

All algorithms are implemented in Python 3.12.0. To ensure stable timing measurements, experiments are conducted on a single idle core of an Intel® Xeon® CPU E7-4809 v4 clocked at 2.10 GHz. Each algorithm is run on each dataset 20 times, and the average of the update time is recorded. For probabilistic algorithms such as Power Iteration and Simultaneous Iteration, we use a fixed random seed to ensure the reproducibility of our results. The experiments conducted in distributed settings are implemented with the help of Ray\footnote{https://www.ray.io/} 2.46.0~\cite{moritz2018ray}. We query and record the sketch and metrics every 20 steps. The datasets we use are:

\begin{enumerate}[leftmargin=*]
    \item \textbf{Uniform Random}: The entries of matrices are drawn uniformly at random from the interval $(0, 1]$. We use this dataset for AMM. We set $d_x=300, d_y=500, N=5000, T=10000$.
    \item \textbf{Random Noisy}~\cite{vershynin2011spectral}: The matrices are generated by the formula $\bm{A} = \bm{SDU} + \bm{N}/\zeta$. Here, $\bm{S}$ is an $n \times d$ matrix of signal coefficients, with each entry drawn from a standard normal distribution. $\bm{D}$ is a diagonal matrix with $\bm{D}_{i, i} = 1-(i-1)/d$. $\bm{U}$ represents the signal row space, satisfying $\bm{UU}^\top = \bm{I}_d$, and $\zeta$ is a positive real parameter for adjusting noise intensity to control the signal-to-noise ratio. The matrix $\bm{N}$ adds Gaussian noise, with $\bm{N}_{i,j}$ drawn from the standard Gaussian distribution $\mathcal{N}(0, 1)$. We use the dataset for ATTP, Distributed (SW). For ATTP, we set $d=500, T=10000$. For Distributed, we set $m=4, d=500, T=10000$. For Distributed (SW), we set $m=4, d=500, N=5000, T=10000$.
    \item \textbf{Real-world Dataset}: GloVe\footnote{\url{https://nlp.stanford.edu/projects/glove/}} provides a dataset of pre-trained word vectors~\cite{pennington2014glove} with dimension of $d=300$. We use the dataset for the Sliding Window. We set $N=5000, T=10000$.
\end{enumerate}

\subsection{Experimental Results}

Figures~\ref{fig:eps-time}-\ref{fig:d-time} show the comparison of different metrics between the baseline algorithms and the \sidsfd algorithms across the five problem scenarios. Among them, the {\color{cyan} cyan} lines represent the space/communication-optimal baseline algorithms, while the {\color{magenta} magenta} lines represent the \sidsfd-optimized algorithms proposed in this work. The {\color{blue} blue} lines represent the speedup ratio, while the {\color{green} green} lines represent the space/communication cost ratio.

\begin{figure*}[htbp]
    \centering
    \begin{subfigure}[t]{0.19\textwidth}
        \captionsetup{justification=centering, skip=0em}
        \includegraphics[width=\textwidth]{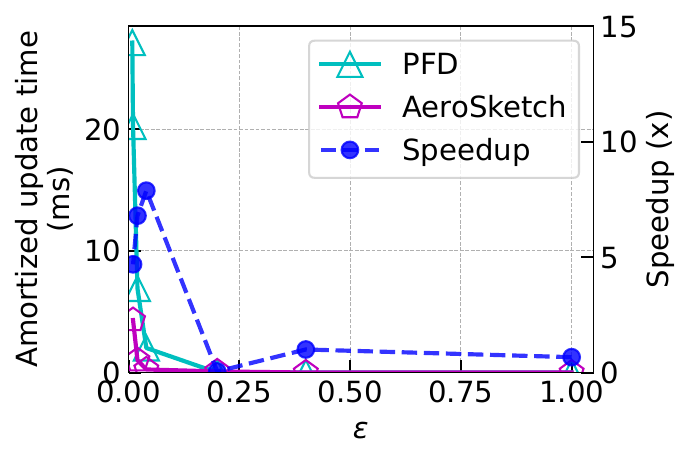}
        \caption{ATTP}\label{fig:eps-time-attp}
    \end{subfigure}
    \begin{subfigure}[t]{0.19\textwidth}
        \captionsetup{justification=centering, skip=0em}
        \includegraphics[width=\textwidth]{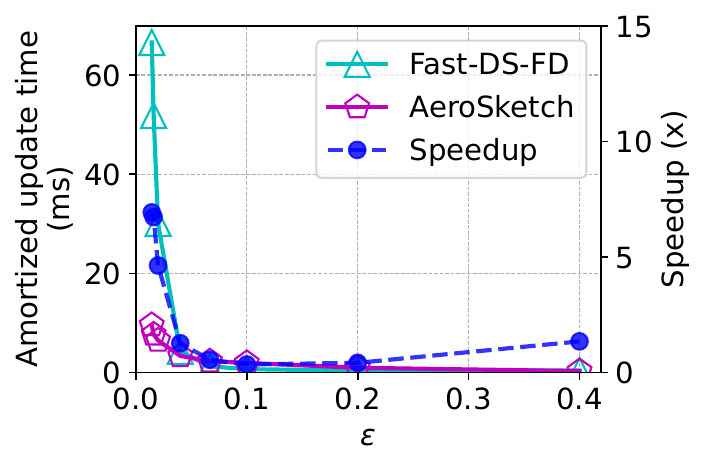}
        \caption{Sliding Window}\label{fig:eps-time-sw}
    \end{subfigure}
    \begin{subfigure}[t]{0.19\textwidth}
        \captionsetup{justification=centering, skip=0em}
        \includegraphics[width=\textwidth]{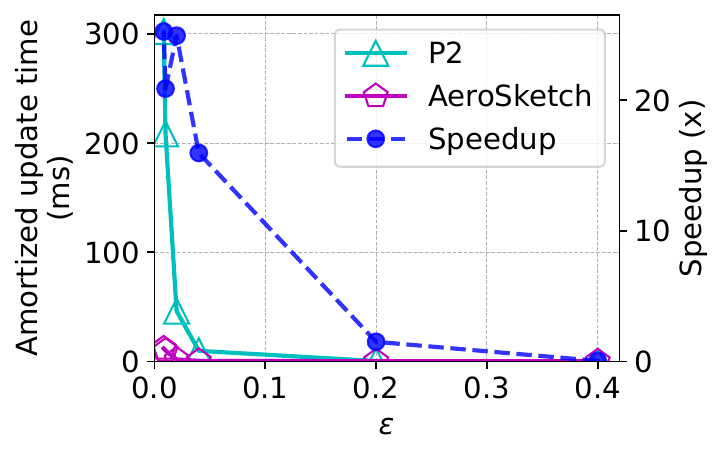}
        \caption{Distributed}\label{fig:eps-time-dsfd}
    \end{subfigure}
    \begin{subfigure}[t]{0.19\textwidth}
        \captionsetup{justification=centering, skip=0em}
        \includegraphics[width=\textwidth]{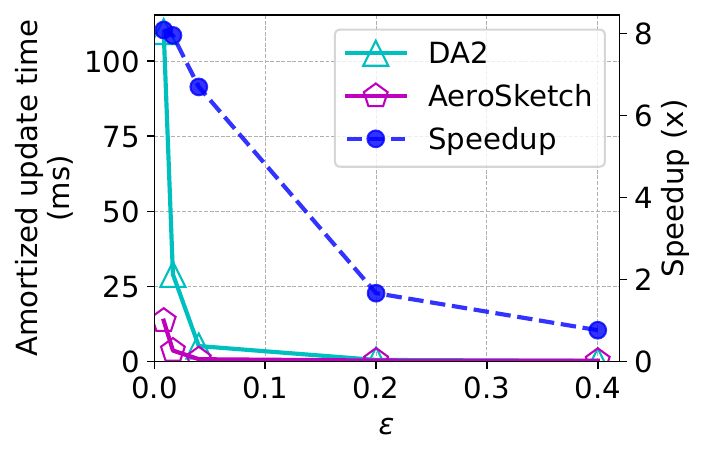}
        \caption{Distributed SW}\label{fig:eps-time-dswfd}
    \end{subfigure}
    \begin{subfigure}[t]{0.19\textwidth}
        \captionsetup{justification=centering, skip=0em}
        \includegraphics[width=\textwidth]{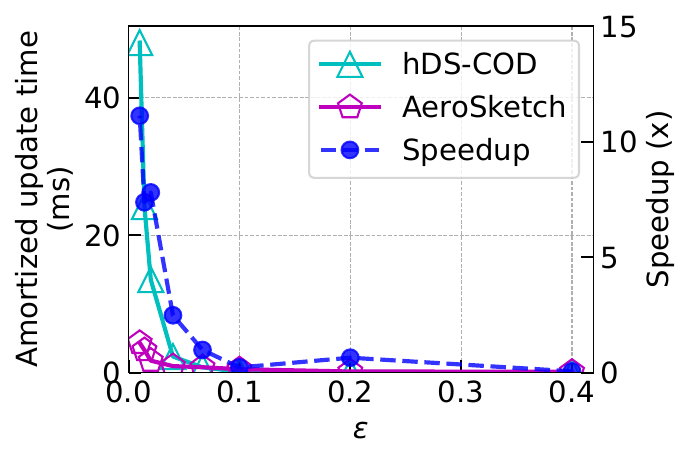}
        \caption{AMM}\label{fig:eps-time-amm}
    \end{subfigure}
\vspace{-1em}
    \caption{Amortized update time vs. parameter $\varepsilon$.}\label{fig:eps-time}
\end{figure*}

\begin{figure*}[htbp]
\vspace{-1em}
    \centering
    \begin{subfigure}[t]{0.19\textwidth}
        \captionsetup{justification=centering,skip=0em}
        \includegraphics[width=\textwidth]{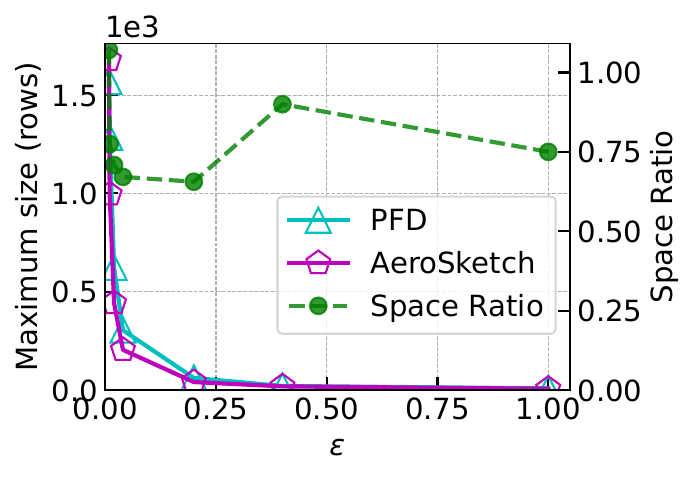}
        \caption{ATTP}\label{fig:eps-size-attp}
    \end{subfigure}
    \begin{subfigure}[t]{0.19\textwidth}
        \captionsetup{justification=centering,skip=0em}
        \includegraphics[width=\textwidth]{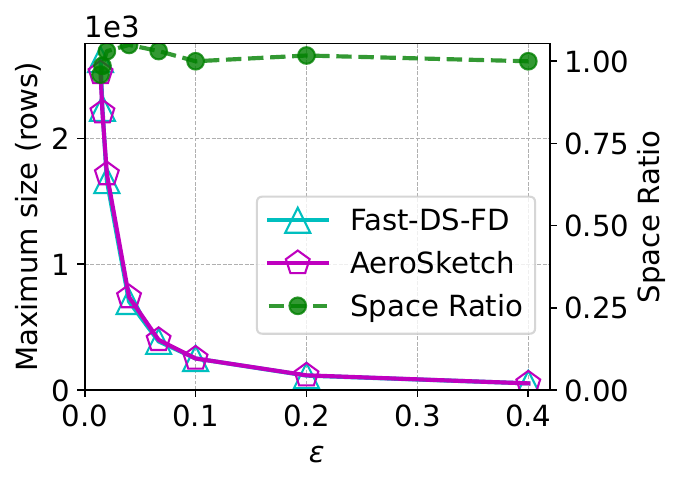}
        \caption{Sliding Window}\label{fig:eps-size-sw}
    \end{subfigure}
    \begin{subfigure}[t]{0.19\textwidth}
        \captionsetup{justification=centering,skip=0em}
        \includegraphics[width=\textwidth]{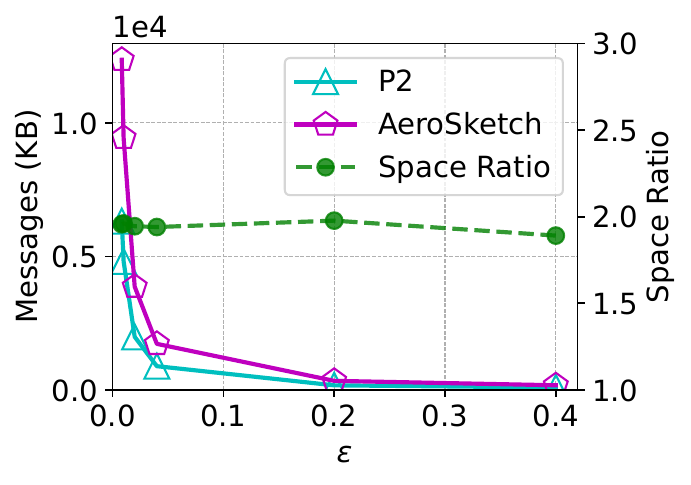}
        \caption{Distributed}\label{fig:eps-size-dsfd}
    \end{subfigure}
    \begin{subfigure}[t]{0.19\textwidth}
        \captionsetup{justification=centering,skip=0em}
        \includegraphics[width=\textwidth]{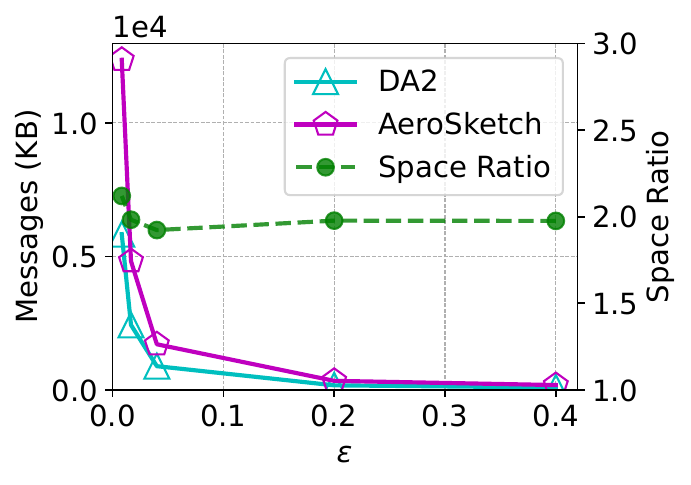}
        \caption{Distributed SW}\label{fig:eps-size-dswfd}
    \end{subfigure}
    \begin{subfigure}[t]{0.19\textwidth}
        \captionsetup{justification=centering,skip=0em}
        \includegraphics[width=\textwidth]{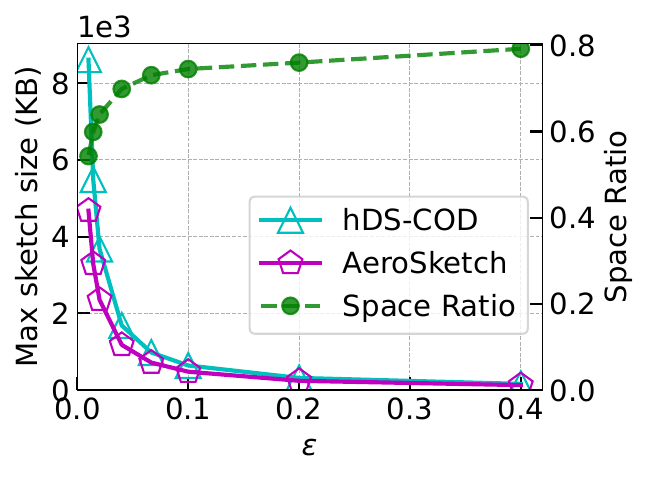}
        \caption{AMM}\label{fig:eps-size-amm}
    \end{subfigure}
\vspace{-1em}
    \caption{Space/communication cost vs. parameter $\varepsilon$.}\label{fig:eps-size}
\end{figure*}

\subsubsection{$\varepsilon$ vs. Time, Space, and Communication Cost}
For both the baseline algorithm and our \sidsfd algorithm, we set a series of parameters $\varepsilon$ (corresponding to the x-axis in Figures~\ref{fig:eps-time} and~\ref{fig:eps-size}) for each of the five problem scenarios, and recorded the amortized update time and maximum space cost (or communication cost under the distributed scenario) of both algorithms under each $\varepsilon$ setting, corresponding to the y-axis in Figures~\ref{fig:eps-time} and~\ref{fig:eps-size}.

Figure~\ref{fig:eps-time} shows the relationship between the parameter $\varepsilon$ and the amortized update time. When $\varepsilon$ is set to a larger value, our \sidsfd algorithms show insignificant acceleration compared to the baseline algorithms. However, as $\varepsilon$ gradually converges to 0, the speedup ratio of our \sidsfd algorithms relative to the baseline algorithms increases significantly, which aligns with the conclusions in Table~\ref{tab:comparison}. Taking Figure~\ref{fig:eps-time-sw} as an example, the time complexity corresponding to Fast-DS-FD is $O\left(\left(\frac{d}{\varepsilon}+\frac{1}{\varepsilon^3}\right)\log R\right)$, while the time complexity of our proposed \sidsfd algorithms in Theorem~\ref{thm:sw} of this paper is $O\left(\frac{d}{\varepsilon}\log d \log R\right)$. When $\varepsilon>1/\sqrt{d}$ (thus $d/\varepsilon>1/\varepsilon^3$), the $1/\varepsilon^3$ term in the baseline algorithm is negligible compared to $d/\varepsilon$. However, as $\varepsilon$ starts from $\varepsilon=O\left(1/d\right)$ and gradually converges to 0, the $1/\varepsilon^3$ term in the baseline algorithm's time complexity begins to dominate, causing the overall time complexity to degrade to a cubic order of $O(d^3)$. In contrast, our \sidsfd achieves a near-quadratic order of $O(d^2 \log d)$. Similar conclusions can be drawn for other scenarios in Table~\ref{tab:comparison} using analogous analysis, indicating that our \sidsfd can prevent performance degradation as $\varepsilon$ converges to 0, ensuring better time performance even under stricter error bounds.

Figure~\ref{fig:eps-size} shows the relationship between the space/communication cost of the algorithms and parameter \(\varepsilon\). We observe that across all values of \(\varepsilon\) in each problem scenario, the space/communication cost ratio between our \sidsfd algorithm and the baseline algorithm stabilizes at a constant between 0.5 and 2. This aligns with the conclusion in Table~\ref{tab:comparison} that both the \sidsfd algorithm and the baseline algorithm achieve the same order and optimal asymptotic space/communication complexity. Theoretically, under the same average error, the space or communication cost of \sidsfd is approximately twice that of the previous state-of-the-art algorithm, as shown in Figures~\ref{fig:eps-size-dsfd} and~\ref{fig:eps-size-dswfd}, as the communication cost ratio is close to 2. The reason, as noted earlier, is that the baseline algorithm computes snapshots exactly, while our approximate algorithm \sidsfd must store an additional subspace-generating matrix \(\bm{Z}_i\) (line 15 in Algorithm~\ref{alg:sidsfd-update}) to recover the original matrix from the residual matrix \(\bm{C}_i\), thus introducing a constant factor of 2. However, we observe that certain space optimization techniques can be employed to reduce the constant factors. For example, in PFD, recording partial checkpoints and full checkpoints can be optimized to record a single type of snapshot; the auxiliary queue in Fast-DS-FD and hDS-COD/SO-COD has been proven to be removable. These techniques can make the actual space/communication cost of \sidsfd essentially equal to or even slightly lower than that of the baseline algorithm under the same relative error upper bound \(\varepsilon\), as reflected in Figures~\ref{fig:eps-size-attp},~\ref{fig:eps-size-sw}, and~\ref{fig:eps-size-amm}, where the space cost ratio is equal to or slightly below 1. Overall, both theoretical analysis and experimental results verify that \sidsfd achieves the same order-optimal asymptotic space complexity as the baseline.

\begin{figure*}[htbp]
    \centering
    \begin{subfigure}[t]{0.19\textwidth}
        \captionsetup{justification=centering,skip=0em}
        \includegraphics[width=\textwidth]{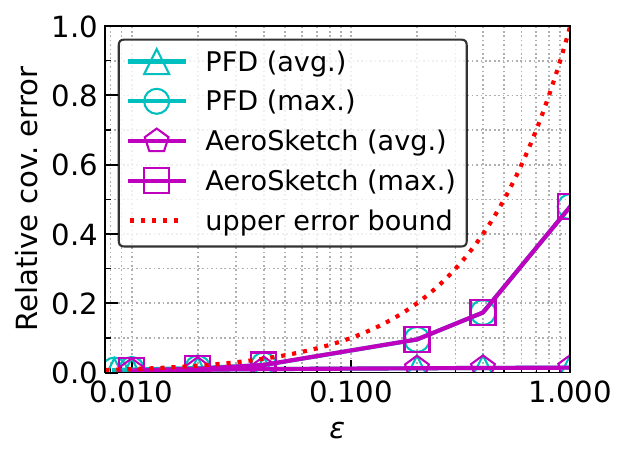}
        \caption{ATTP}\label{fig:eps-err-attp}
    \end{subfigure}
    \begin{subfigure}[t]{0.19\textwidth}
        \captionsetup{justification=centering,skip=0em}
        \includegraphics[width=\textwidth]{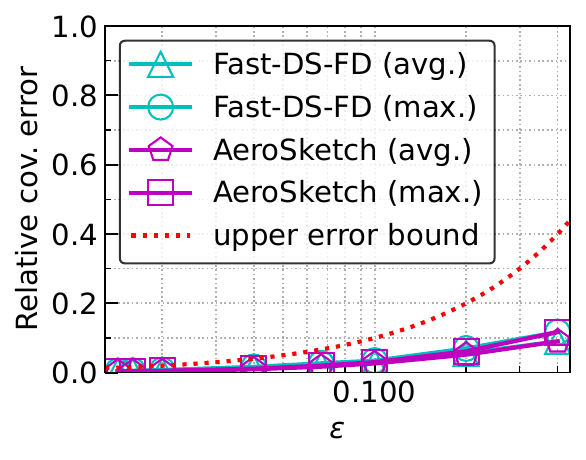}
        \caption{Sliding Window}\label{fig:eps-err-sw}
    \end{subfigure}
    \begin{subfigure}[t]{0.19\textwidth}
        \captionsetup{justification=centering,skip=0em}
        \includegraphics[width=\textwidth]{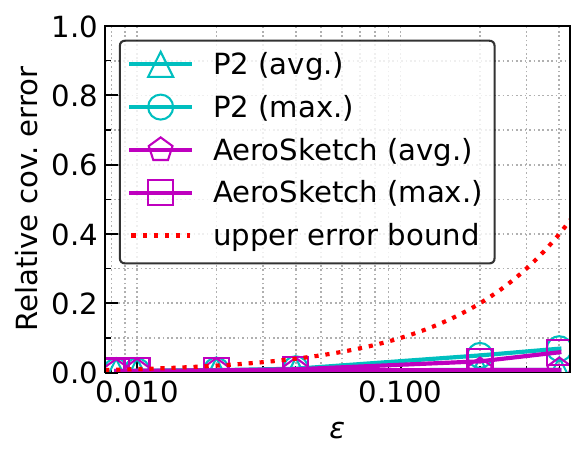}
        \caption{Distributed}\label{fig:eps-err-dsfd}
    \end{subfigure}
    \begin{subfigure}[t]{0.19\textwidth}
        \captionsetup{justification=centering,skip=0em}
        \includegraphics[width=\textwidth]{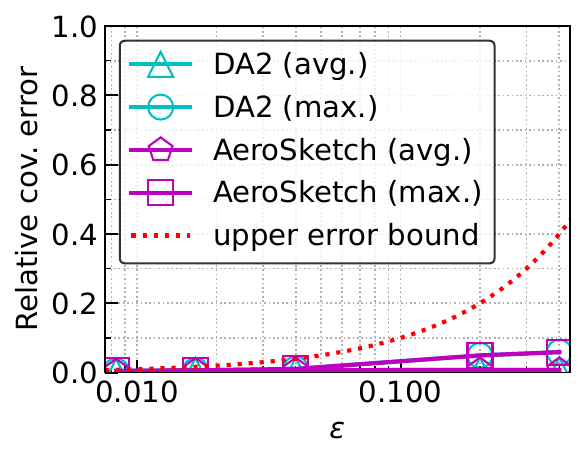}
        \caption{Distributed SW}\label{fig:eps-err-dswfd}
    \end{subfigure}
    \begin{subfigure}[t]{0.19\textwidth}
        \captionsetup{justification=centering,skip=0em}
        \includegraphics[width=\textwidth]{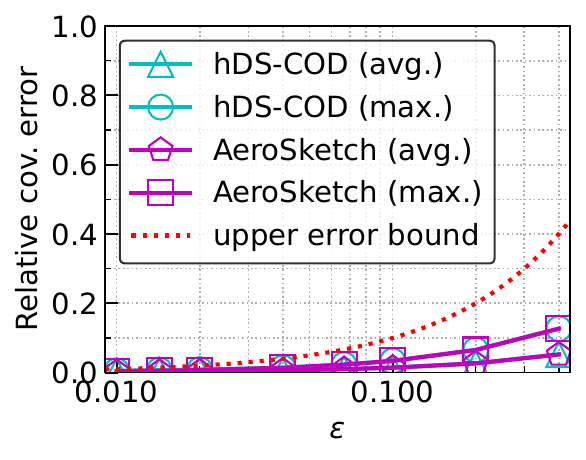}
        \caption{AMM}\label{fig:eps-err-amm}
    \end{subfigure}
    \caption{Empirical relative covariance errors vs. parameter $\varepsilon$.}\label{fig:eps-err}
\end{figure*}

\begin{remark}
    It is important to note that the parameter $\varepsilon$ used in the algorithm is an upper bound on the empirical relative covariance error (i.e., $\varepsilon$ is the worst-case approximation guarantee). As shown in Figure~\ref{fig:eps-err}, the maximum empirical relative covariance error of all algorithms is lower than the upper bound parameter $\varepsilon$, which validates the correctness of our algorithm implementation. Moreover, even under the same $\varepsilon$ setting, different algorithms may exhibit different empirically measured covariance errors. Given this, beyond the analysis of $\varepsilon$'s effect on performance metrics of time and space/communication cost in the previous subsection, we now turn to a comparison of the time and space/communication cost of different algorithms under the same \textbf{empirical errors} in the next subsection.
\end{remark}

\begin{figure*}[htbp]
    \centering
    \begin{subfigure}[t]{0.19\textwidth}
        \captionsetup{justification=centering,skip=0em}
        \includegraphics[width=\textwidth]{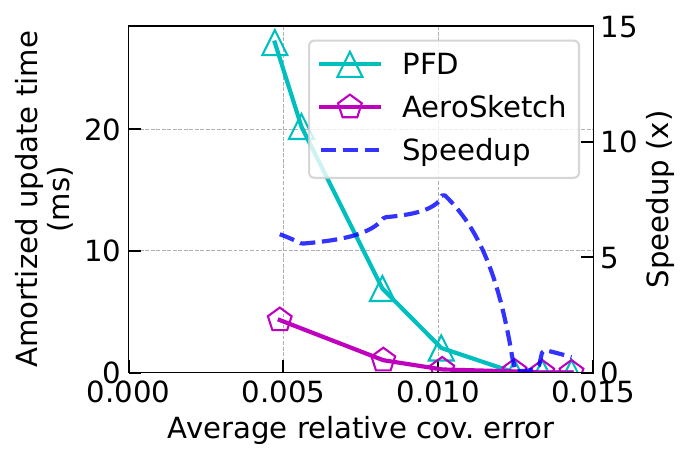}
        \caption{ATTP}\label{fig:time-avg-err-attp}
    \end{subfigure}
    \begin{subfigure}[t]{0.19\textwidth}
        \captionsetup{justification=centering,skip=0em}
        \includegraphics[width=\textwidth]{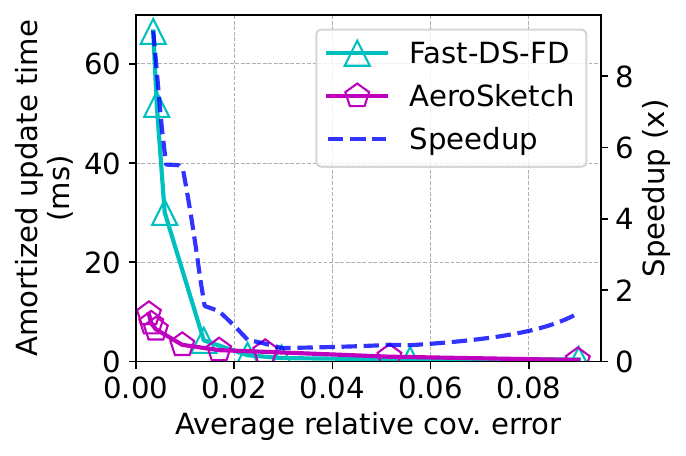}
        \caption{Sliding Window}\label{fig:time-avg-err-sw}
    \end{subfigure}
    \begin{subfigure}[t]{0.19\textwidth}
        \captionsetup{justification=centering,skip=0em}
        \includegraphics[width=\textwidth]{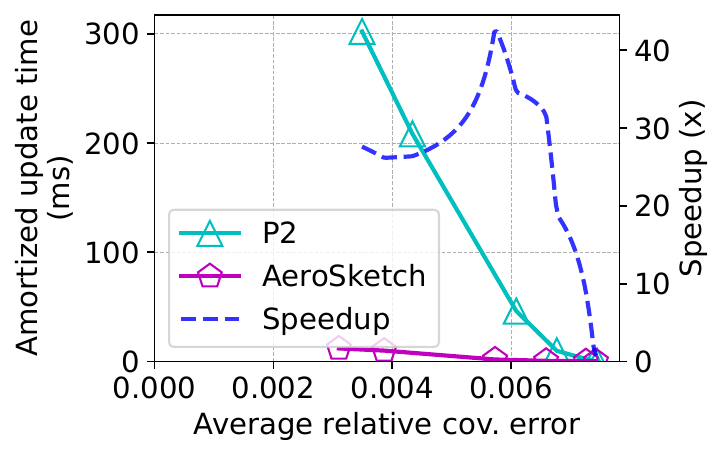}
        \caption{Distributed}\label{fig:time-avg-err-dsfd}
    \end{subfigure}
    \begin{subfigure}[t]{0.19\textwidth}
        \captionsetup{justification=centering,skip=0em}
        \includegraphics[width=\textwidth]{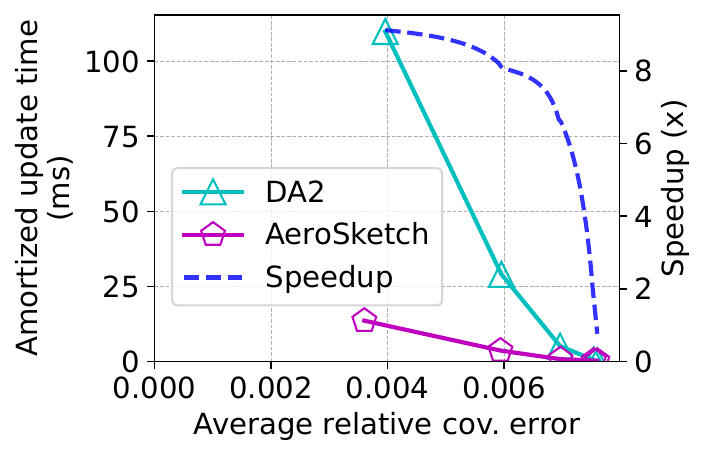}
        \caption{Distributed SW}\label{fig:time-avg-err-dswfd}
    \end{subfigure}
    \begin{subfigure}[t]{0.19\textwidth}
        \captionsetup{justification=centering,skip=0em}
        \includegraphics[width=\textwidth]{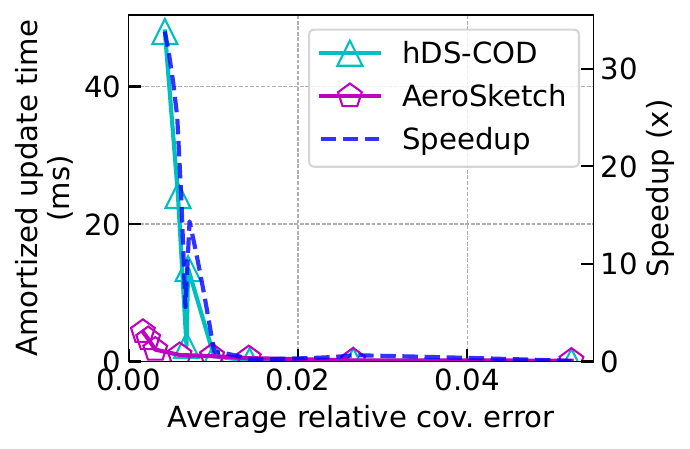}
        \caption{AMM}\label{fig:time-avg-err-amm}
    \end{subfigure}
    \caption{Amortized update time vs. Average error}\label{fig:time-avg-err}
\end{figure*}

\begin{figure*}[htbp]
\vspace{-1em}
    \centering
    \begin{subfigure}[t]{0.19\textwidth}
        \captionsetup{justification=centering,skip=0em}
        \includegraphics[width=\textwidth]{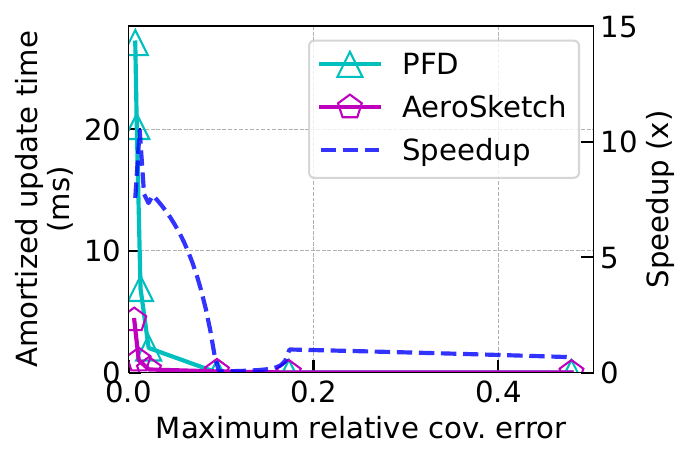}
        \caption{ATTP}\label{fig:time-max-err-attp}
    \end{subfigure}
    \begin{subfigure}[t]{0.19\textwidth}
        \captionsetup{justification=centering,skip=0em}
        \includegraphics[width=\textwidth]{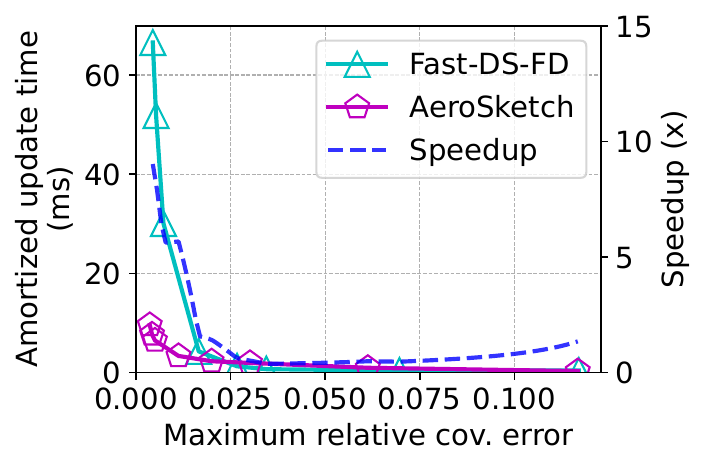}
        \caption{Sliding Window}\label{fig:time-max-err-sw}
    \end{subfigure}
    \begin{subfigure}[t]{0.19\textwidth}
        \captionsetup{justification=centering,skip=0em}
        \includegraphics[width=\textwidth]{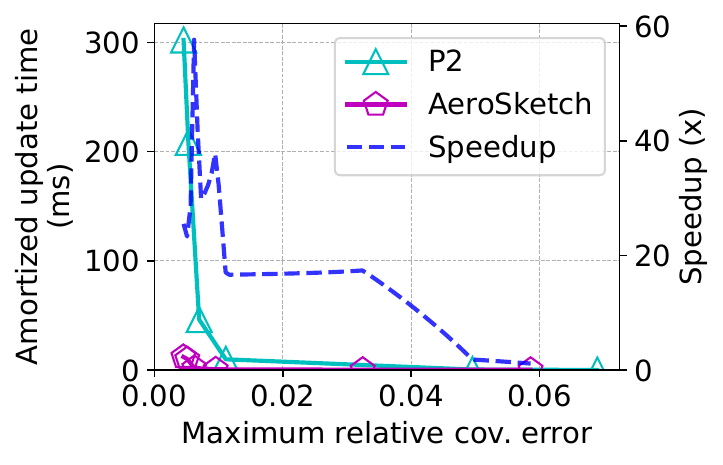}
        \caption{Distributed}\label{fig:time-max-err-dsfd}
    \end{subfigure}
    \begin{subfigure}[t]{0.19\textwidth}
        \captionsetup{justification=centering,skip=0em}
        \includegraphics[width=\textwidth]{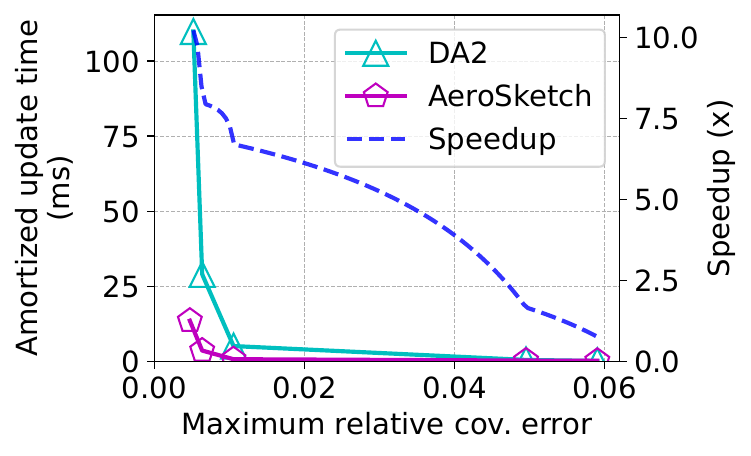}
        \caption{Distributed SW}\label{fig:time-max-err-dswfd}
    \end{subfigure}
    \begin{subfigure}[t]{0.19\textwidth}
        \captionsetup{justification=centering,skip=0em}
        \includegraphics[width=\textwidth]{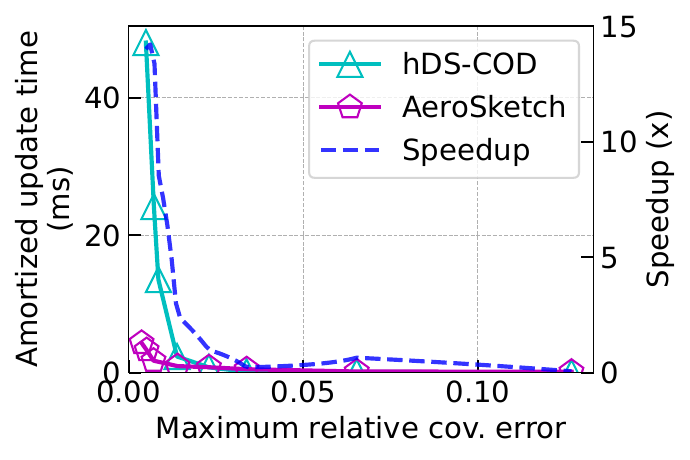}
        \caption{AMM}\label{fig:time-max-err-amm}
    \end{subfigure}
    \caption{Amortized update time vs. Maximum error}\label{fig:time-max-err}
\end{figure*}

\subsubsection{Empirical Error vs. Time vs. Space/Communication Cost}

Figures~\ref{fig:time-avg-err} and~\ref{fig:time-max-err} show the trade-off between amortized update time and avg./max. relative error across the five problem scenarios. We observe that, under the same relative error guarantee (i.e., along the same x-axis), the amortized update time of the \sidsfd algorithm is generally lower than that of the baseline algorithms. When the empirical error is not very small, the time gap between the two algorithms is limited. However, as the error bound tightens, the time advantage of the \sidsfd optimization becomes more pronounced. This aligns with our theoretical analysis: as $\varepsilon$ converges to 0, the update time complexity of baseline algorithms degrades to $O(d^3)$, while \sidsfd maintains $O(d^2 \log d)$.

\begin{figure*}[htbp]
    \centering
    \begin{subfigure}[t]{0.19\textwidth}
        \captionsetup{justification=centering,skip=0em}
        \includegraphics[width=\textwidth]{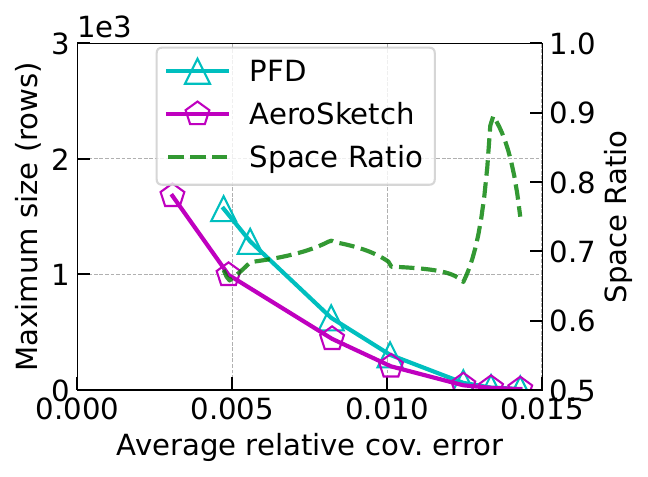}
        \caption{ATTP}\label{fig:size-avg-err-attp}
    \end{subfigure}
    \begin{subfigure}[t]{0.19\textwidth}
        \captionsetup{justification=centering,skip=0em}
        \includegraphics[width=\textwidth]{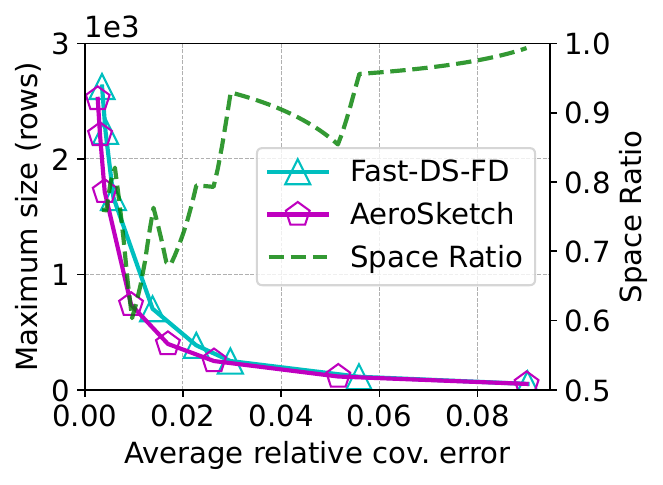}
        \caption{Sliding Window}\label{fig:size-avg-err-sw}
    \end{subfigure}
    \begin{subfigure}[t]{0.19\textwidth}
        \captionsetup{justification=centering,skip=0em}
        \includegraphics[width=\textwidth]{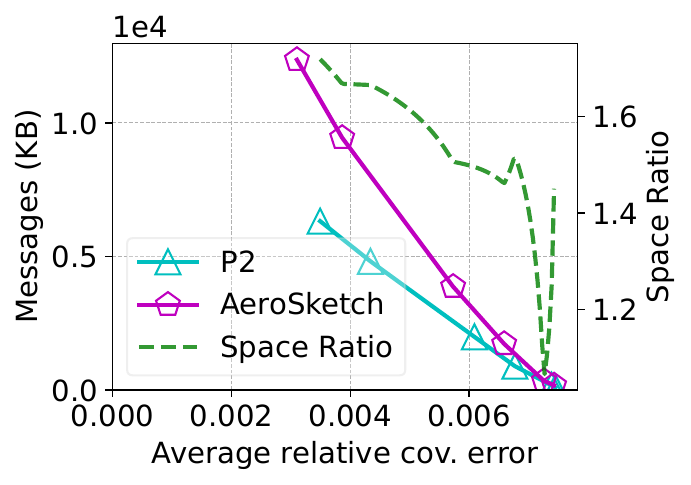}
        \caption{Distributed}\label{fig:size-avg-err-dsfd}
    \end{subfigure}
    \begin{subfigure}[t]{0.19\textwidth}
        \captionsetup{justification=centering,skip=0em}
        \includegraphics[width=\textwidth]{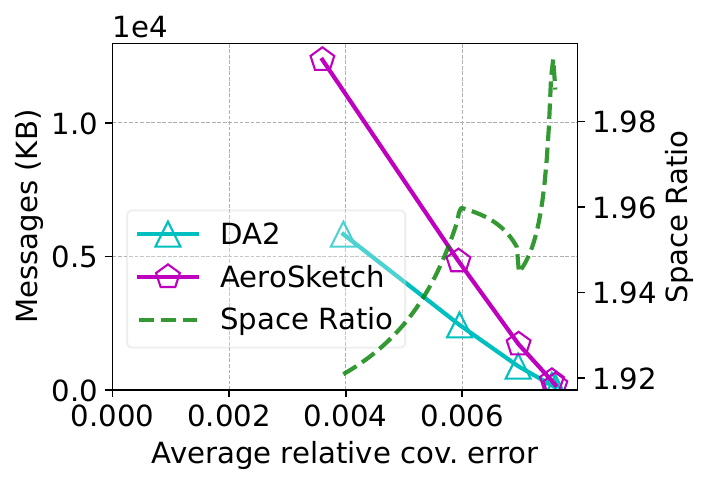}
        \caption{Distributed SW}\label{fig:size-avg-err-dswfd}
    \end{subfigure}
    \begin{subfigure}[t]{0.19\textwidth}
        \captionsetup{justification=centering,skip=0em}
        \includegraphics[width=\textwidth]{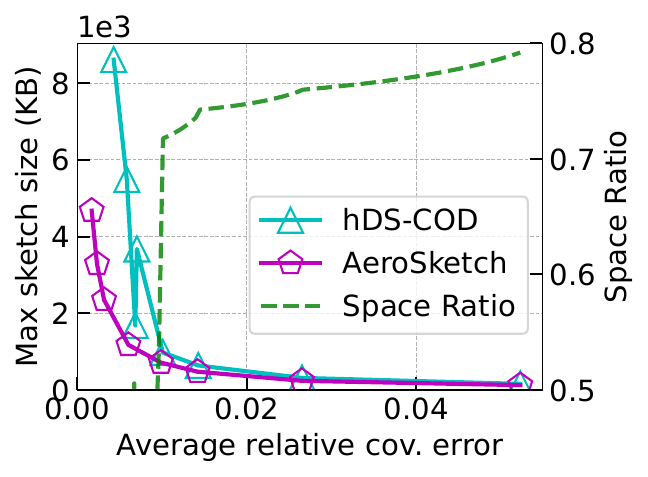}
        \caption{AMM}\label{fig:size-avg-err-amm}
    \end{subfigure}
    \caption{Space/communication cost vs. Average error}\label{fig:size-avg-err}
\end{figure*}

\begin{figure*}[htbp]
    \centering
    \begin{subfigure}[t]{0.19\textwidth}
        \captionsetup{justification=centering,skip=0em}
        \includegraphics[width=\textwidth]{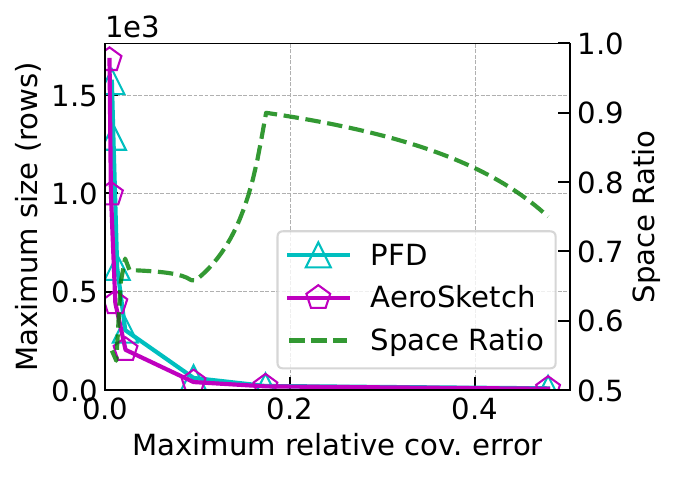}
        \caption{ATTP}\label{fig:size-max-err-attp}
    \end{subfigure}
    \begin{subfigure}[t]{0.19\textwidth}
        \captionsetup{justification=centering,skip=0em}
        \includegraphics[width=\textwidth]{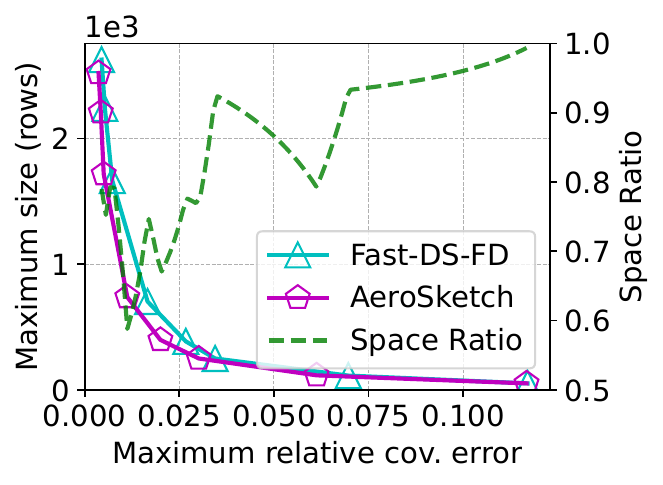}
        \caption{Sliding Window}\label{fig:size-max-err-sw}
    \end{subfigure}
    \begin{subfigure}[t]{0.19\textwidth}
        \captionsetup{justification=centering,skip=0em}
        \includegraphics[width=\textwidth]{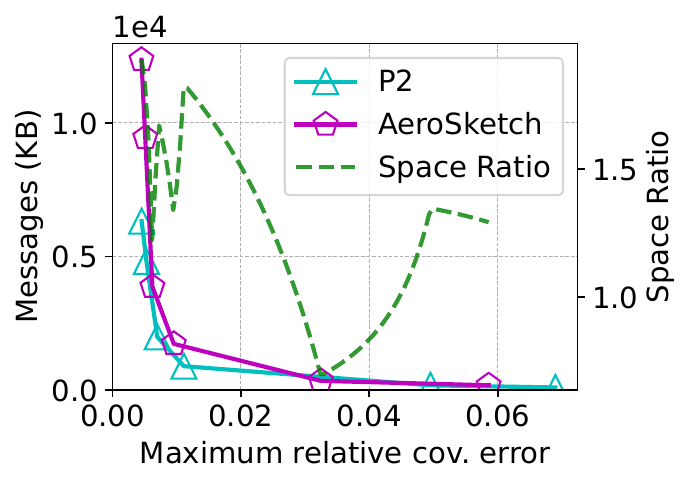}
        \caption{Distributed}\label{fig:size-max-err-dsfd}
    \end{subfigure}
    \begin{subfigure}[t]{0.19\textwidth}
        \captionsetup{justification=centering,skip=0em}
        \includegraphics[width=\textwidth]{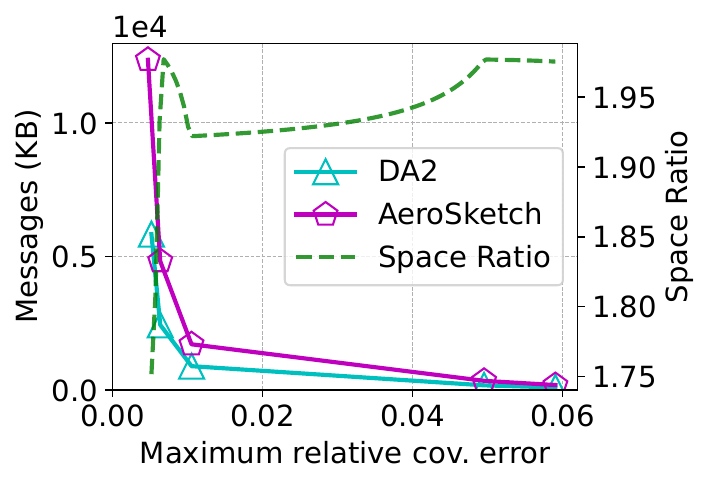}
        \caption{Distributed SW}\label{fig:size-max-err-dswfd}
    \end{subfigure}
    \begin{subfigure}[t]{0.19\textwidth}
        \captionsetup{justification=centering,skip=0em}
        \includegraphics[width=\textwidth]{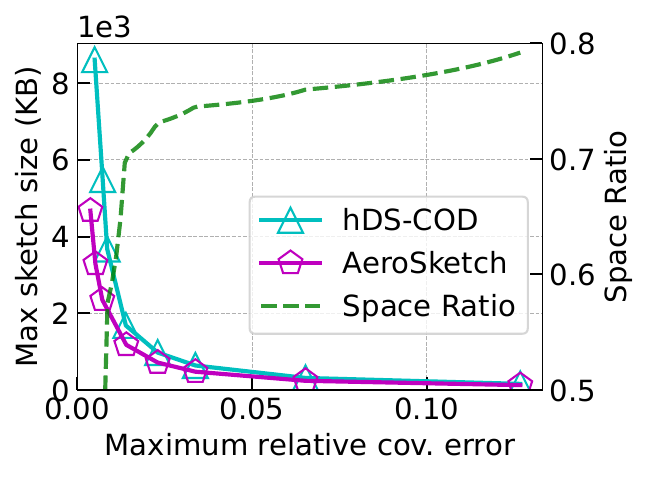}
        \caption{AMM}\label{fig:size-max-err-amm}
    \end{subfigure}
    \caption{Space/communication cost vs. Maximum error}\label{fig:size-max-err}
\end{figure*}

Figures~\ref{fig:size-avg-err} and~\ref{fig:size-max-err} show the space/communication costs vs. error. We observe that, under the same empirical error, the space/communication costs of baseline algorithms and our \sidsfd\ are comparable. In some cases, \sidsfd\ even achieves lower space costs, consistent with our earlier theoretical analysis that its space/communication complexity is of the same order and asymptotically optimal compared to baseline algorithms.

\begin{figure*}[htbp]
    \centering
    \begin{subfigure}[t]{0.19\textwidth}
        \captionsetup{justification=centering,skip=0em}
        \includegraphics[width=\textwidth]{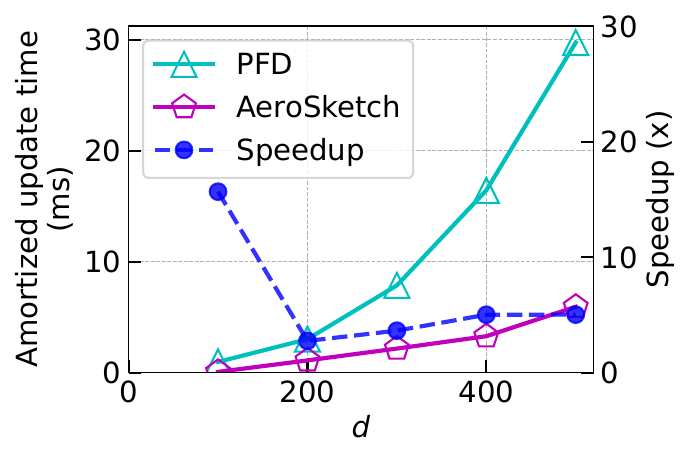}
        \caption{ATTP}\label{fig:d-time-attp}
    \end{subfigure}
    \begin{subfigure}[t]{0.19\textwidth}
        \captionsetup{justification=centering,skip=0em}
        \includegraphics[width=\textwidth]{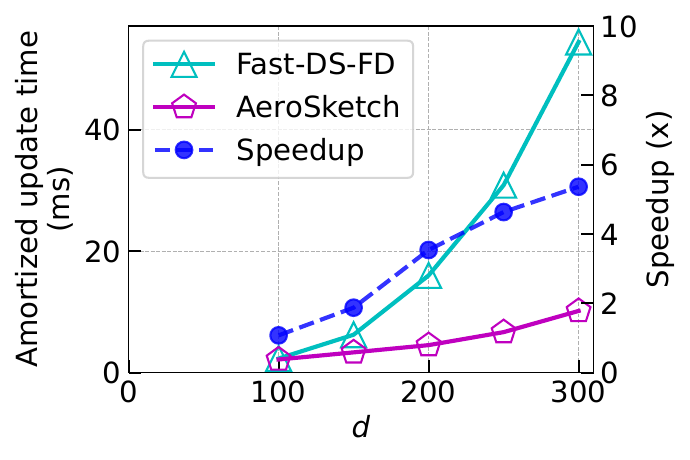}
        \caption{Sliding Window}\label{fig:d-time-sw}
    \end{subfigure}
    \begin{subfigure}[t]{0.19\textwidth}
        \captionsetup{justification=centering,skip=0em}
        \includegraphics[width=\textwidth]{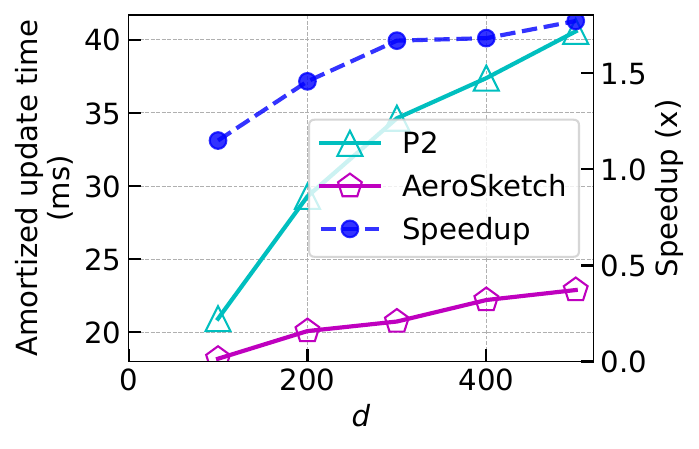}
        \caption{Distributed}\label{fig:d-time-dsfd}
    \end{subfigure}
    \begin{subfigure}[t]{0.19\textwidth}
        \captionsetup{justification=centering,skip=0em}
        \includegraphics[width=\textwidth]{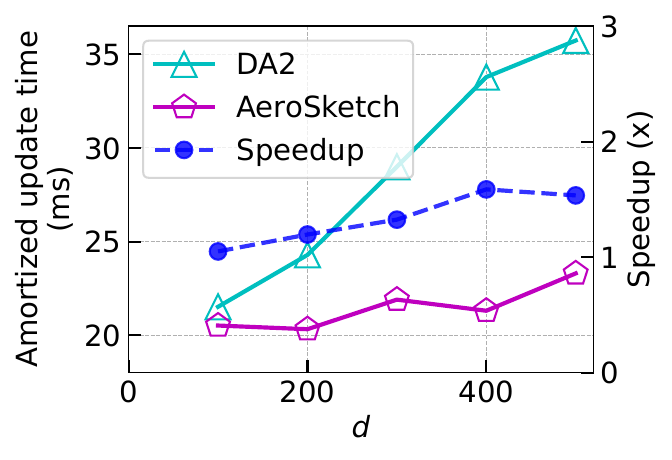}
        \caption{Distributed SW}\label{fig:d-time-dswfd}
    \end{subfigure}
    \begin{subfigure}[t]{0.19\textwidth}
        \captionsetup{justification=centering,skip=0em}
        \includegraphics[width=\textwidth]{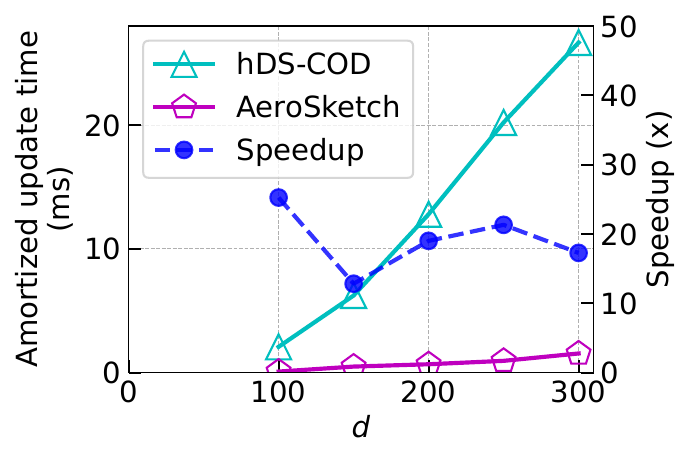}
        \caption{AMM}\label{fig:d-time-amm}
    \end{subfigure}
    \caption{$d$ vs. Amortized update time}\label{fig:d-time}
\end{figure*}

\subsubsection{Asymptotic analysis of amortized update time complexity with respect to dimension $d$}

In addition, the time efficiency of \sidsfd becomes more prominent for higher-dimensional vector streams. As we analyzed earlier, when $\varepsilon$ tightens, the update time complexities of previous baseline algorithms degrade to a cubic order of $O(d^3)$. In contrast, our \sidsfd achieves a near-quadratic order of $O(d^2 \log d)$. The relationship between the empirical amortized update time and the vector dimension \(d\) plotted in Figure~\ref{fig:d-time} clearly validates this point. For each value of \(d\), we set \(\varepsilon = 4/d = O\left(1/d\right)\). We observe that as \(d\) increases, the time costs of baseline algorithms grow significantly faster compared to \sidsfd. As dimension $d$ increases further, the time advantage of \sidsfd\ is expected to become even more substantial.

\section{Conclusion and Future Work}

In this paper, we present a new framework, \sidsfd, for optimizing the update time complexity of matrix sketch algorithms to a near-optimal level, up to a logarithmic factor. We show that \sidsfd can address a wide range of matrix sketch problems while achieving optimal asymptotic space complexity. We provide end-to-end proofs for the error bounds, time complexity, and space complexity to demonstrate the correctness and effectiveness of our approach. We implement the algorithms optimized by our framework and evaluate them on various matrix sketch problems, using large-scale synthetic and real-world streams, in comparison with baseline algorithms. The results show that our framework significantly outperforms prior methods in terms of update time, especially under tight approximation error constraints. Meanwhile, it remains competitive in space and communication cost. The experimental results corroborate our theoretical analysis.

Potential future work encompasses the following direction: 
Is there a (truly deterministic) algorithm that achieves a truly optimal amortized update time complexity $\Omega(d/\varepsilon)$, rather than the near-optimal one with a logarithmic factor—thus outperforming \sidsfd?

\begin{acks}
  \begin{sloppypar}
    This research was supported in part by National Science and Technology Major Project (2022ZD0114802), by National Natural Science Foundation of China (No. 92470128, No. U2241212). We also wish to acknowledge the support provided by the fund for building world-class universities (disciplines) of Renmin University of China, by Engineering Research Center of Next-Generation Intelligent Search and Recommendation, Ministry of Education, by Intelligent Social Governance Interdisciplinary Platform, Major Innovation \& Planning Interdisciplinary Platform for the ``Double-First Class'' Initiative, Public Policy and Decision-making Research Lab, and Public Computing Cloud, Renmin University of China.

    We thank Mingji Yang and Guanyu Cui for their valuable comments that improved the presentation of this article. We also thank Professor Sibo Wang for his many helpful suggestions.
  \end{sloppypar}
\end{acks}

\bibliographystyle{ACM-Reference-Format}
\bibliography{sample-base}

\appendix

\clearpage

\appendix

\section{Proof of Corollary~\ref{thm:power}}
\label{appendix:proof-power}

\begin{repcorollary}{thm:power}[Probabilistic error bound and time complexity of Power Iteration, a simplified corollary of Theorem 4.1(a) in~\cite{kuczynski1992estimating}]
    For Algorithm~\ref{alg:power}, if we set $k = \lceil\log_{2} d\rceil+1$, then the probability that the estimated top squared singular value $\hat{\sigma}_1^2$ exceeds one half of the true value $\sigma_1^2 / 2$ is bounded as:
    \begin{equation*}
        \Pr\left[\hat{\sigma}_1^2 \ge \sigma_1^2/2\right] \ge 1- {\frac{2}{\pi\sqrt{e}}\cdot\frac{1}{\sqrt{d \log_2 d}}},
    \end{equation*}
    where $e$ denotes Euler's number. Power iteration requires only matrix-vector multiplications. The total computational cost of Algorithm~\ref{alg:power} is $O(d \ell k)$. If we set $k = \lceil\log_2 d\rceil+1$, then the computational cost becomes $O(d \ell \log d)$.
\end{repcorollary}

Corollary~\ref{thm:power} is a simplified corollary derived from Theorem 4.1(a) in~\cite{kuczynski1992estimating}, which states that

\begin{theorem}[Theorem 4.1(a) in~\cite{kuczynski1992estimating}]
    \label{thm:4.1a}
    Suppose that $\bm{A}^\top \bm{A}\in\mathbb{R}^{d\times d}$. Let $f^{\text{prob}}(\bm{A}, k, \varepsilon)$ be the probability that after $k$ iterations, the Power Iteration Algorithm fails to approximate the squared largest singular value of $\bm{A}$ with relative error at most $\varepsilon$, i.e., $\left| \frac{\hat{\sigma}_1^2-\sigma_1^2}{\sigma_1^2}\right|>\varepsilon$. For any $k\ge 2$, we have
    \begin{equation}
        \label{eq:apx-1}
        f^{\text{prob}}(\bm{A}, k, \varepsilon) \le \min \left\{0.824, \frac{0.354}{\sqrt{\varepsilon(k-1)}}\right\} \sqrt{d}(1-\varepsilon)^{k-1 / 2}.
    \end{equation}
\end{theorem}

Here, we provide the derivation from Theorem~\ref{thm:4.1a} to Corollary~\ref{thm:power}:

\begin{proof}[Proof of Corollary~\ref{thm:power}]
    We set the number of iterations $k=\lceil\log_2 d\rceil + 1$. Plugging $k$ in \eqref{eq:apx-1}, we have
    \begin{equation*}
        \Pr\left[\left| \frac{\hat{\sigma}_1^2-\sigma_1^2}{\sigma_1^2}\right|>\varepsilon\right]\le \frac{0.354}{\sqrt{\varepsilon \cdot \log_2 d}} \sqrt{d}(1-\varepsilon)^{\log_2 d + 1 / 2}.
    \end{equation*}
    Then we set $\varepsilon = 1/2$, we have
    \begin{align*}
        & \phantom{{}={}} \Pr\left[\hat{\sigma}_1^2<\sigma_1^2/2\right]
        =\Pr\left[ \frac{\hat{\sigma}_1^2-\sigma_1^2}{\sigma_1^2}<-\frac{1}{2}\right]
        \le\Pr\left[\left| \frac{\hat{\sigma}_1^2-\sigma_1^2}{\sigma_1^2}\right|>\frac{1}{2}\right] \\
        & = \Pr\left[\left| \frac{\hat{\sigma}_1^2-\sigma_1^2}{\sigma_1^2}\right|>\varepsilon\right]
        \le \frac{0.354}{\sqrt{\varepsilon\cdot\log_2 d}} \sqrt{d}(1-\varepsilon)^{\log_2 d + 1 / 2} \\
        & = \frac{0.354}{\sqrt{(\log_2 d)/2}} \sqrt{d}\cdot\left(\frac{1}{2}\right)^{\log_2 d + 1 / 2}
        = \frac{0.354}{\sqrt{d \log_2 d}},
    \end{align*}
    where the numerical constant corresponds to the numerical evaluation of $\frac{192}{105} \cdot \frac{1}{\pi \sqrt{e}} \approx 0.354$, and it holds that $\frac{192}{105} \cdot \frac{1}{\pi \sqrt{e}} \le \frac{2}{\pi \sqrt{e}}$. Therefore, we have
    \begin{equation*}
        \Pr\left[\hat{\sigma}_1^2<\sigma_1^2/2\right]\le \frac{2}{\pi\sqrt{e}}\cdot\frac{1}{\sqrt{d \log_2 d}}.
    \end{equation*}
    Hence, we have
    \begin{equation*}
        \Pr\left[\hat{\sigma}_1^2\ge \sigma_1^2/2\right]\ge 1- \frac{2}{\pi\sqrt{e}}\cdot\frac{1}{\sqrt{d \log_2 d}},
    \end{equation*}
    which is the desired result in Corollary~\ref{thm:power}. 
\end{proof}

\section{Implicit Restoring Norm Error in Fast-DS-FD}
\label{app:norm-error}

The implicit restoring norm error in Fast-DS-FD is introduced in the dumping snapshot procedure of lines~15-21 in the update operation (Algorithm~3 in~\cite{yin2024optimal}) and the restoring of the snapshots during the query operation (Algorithm~4 in~\cite{yin2024optimal}). The for-loop across lines~16-21 of Algorithm~3 performs multiple dumping snapshot procedures $\bm{D}\gets\bm{D}-\bm{D}\bm{v}_j\bm{v}_j^\top$. We denote by $\bm{C}_i^\prime$ the value of $\bm{D}$ before the for-loop, and by $\bm{C}_i$ its value after the for-loop. The $\bm{v}_j$s are the singular vectors of ${\bm{C}_i^\prime}^\top\bm{C}_i^\prime$. After the for-loop, we get
\begin{equation*}
    \bm{C}_i^\top \bm{C}_i = (\bm{I}-\bm{V}_i\bm{V}_i^\top){\bm{C}_i^\prime}^\top\bm{C}_i^\prime(\bm{I}-\bm{V}_i\bm{V}_i^\top),
\end{equation*}
where $\bm{V}_i$ is the matrix formed by some columns of singular vectors of ${\bm{C}_i^\prime}^\top\bm{C}_i^\prime$. During the query operation, Fast-DS-FD attempts to restore the ${\bm{C}_i^\prime}^\top \bm{C}_i^\prime$ by the snapshot $\sqrt{\bm{\Sigma}_i} \, \bm{V}_i^\top$, where $\bm{\Sigma}_i$ is the diagonal matrix of the singular values corresponding to the singular vectors $\bm{V}_i$:
\begin{equation*}
    \hat{\bm{C}_i^\prime}^\top\hat{\bm{C}_i^\prime} = \bm{C}_i^\top \bm{C}_i + \bm{V}_i\bm{\Sigma}_i\bm{V}_i^\top.
\end{equation*}
Since $\bm{V}_i$ consists of (orthonormal) eigenvectors of ${\bm{C}_i^\prime}^\top\bm{C}_i^\prime$, we have
\begin{equation*}
    \bm{\Sigma}_i = \bm{V}_i^\top{\bm{C}_i^\prime}^\top \bm{C}_i^\prime\bm{V}_i.
\end{equation*}
Then we get
\begin{equation*}
    \begin{split}
        \hat{\bm{C}_i^\prime}^\top\hat{\bm{C}_i^\prime} = &\ \bm{C}_i^\top \bm{C}_i + \bm{V}_i \bm{V}_i^\top{\bm{C}_i^\prime}^\top \bm{C}_i^\prime\bm{V}_i\bm{V}_i^\top \\
        = &\ (\bm{I}-\bm{V}_i\bm{V}_i^\top){\bm{C}_i^\prime}^\top\bm{C}_i^\prime(\bm{I}-\bm{V}_i\bm{V}_i^\top)+ \bm{V}_i \bm{V}_i^\top{\bm{C}_i^\prime}^\top \bm{C}_i^\prime\bm{V}_i\bm{V}_i^\top \\
        = &\ {\bm{C}_i^\prime}^\top\bm{C}_i^\prime - \bm{V}_i \bm{V}_i^\top{\bm{C}_i^\prime}^\top \bm{C}_i^\prime - {\bm{C}_i^\prime}^\top \bm{C}_i^\prime\bm{V}_i\bm{V}_i^\top + 2\bm{V}_i \bm{V}_i^\top{\bm{C}_i^\prime}^\top \bm{C}_i^\prime\bm{V}_i\bm{V}_i^\top. \\
    \end{split}
\end{equation*}
The error between the approximate $\hat{\bm{C}_i^\prime}^\top\hat{\bm{C}_i^\prime}$ and the original ${\bm{C}_i^\prime}^\top\bm{C}_i^\prime$ is 
\begin{equation}
    \label{eq:restore-4}
        {\bm{C}_i^\prime}^\top\bm{C}_i^\prime - \hat{\bm{C}_i^\prime}^\top\hat{\bm{C}_i^\prime}= \bm{V}_i \bm{V}_i^\top{\bm{C}_i^\prime}^\top \bm{C}_i^\prime + {\bm{C}_i^\prime}^\top \bm{C}_i^\prime\bm{V}_i\bm{V}_i^\top - 2\bm{V}_i \bm{V}_i^\top{\bm{C}_i^\prime}^\top \bm{C}_i^\prime\bm{V}_i\bm{V}_i^\top. \\
\end{equation}

If $\bm{V}_i$ consists of exact singular vectors computed by a full SVD, then the error is zero. However, if $\bm{V}_i$ are substituted by approximate singular vectors, e.g.,  $\bm{Z}_i \approx \bm{V}_i$ computed via randomized numerical linear algebra (RandNLA), then Eq.~\eqref{eq:restore-4} becomes
\begin{equation*}
        {\bm{C}_i^\prime}^\top\bm{C}_i^\prime - \hat{\bm{C}_i^\prime}^\top\hat{\bm{C}_i^\prime}= \bm{Z}_i \bm{Z}_i^\top{\bm{C}_i^\prime}^\top \bm{C}_i^\prime + {\bm{C}_i^\prime}^\top \bm{C}_i^\prime\bm{Z}_i\bm{Z}_i^\top - 2\bm{Z}_i \bm{Z}_i^\top{\bm{C}_i^\prime}^\top \bm{C}_i^\prime\bm{Z}_i\bm{Z}_i^\top, \\
\end{equation*}
and whether the error is zero is not guaranteed. Moreover, the error may accumulate over updates, and the norm of the restoration error over the sliding window is
\begin{equation*}
\left\|\sum_{i\in W} \left(\bm{Z}_i \bm{Z}_i^\top {\bm{C}_i^\prime}^\top \bm{C}_i^\prime + {\bm{C}_i^\prime}^\top \bm{C}_i^\prime \bm{Z}_i \bm{Z}_i^\top - 2\bm{Z}_i \bm{Z}_i^\top {\bm{C}_i^\prime}^\top \bm{C}_i^\prime \bm{Z}_i \bm{Z}_i^\top\right) \right\|_2.
\end{equation*}

\section{Proof of Theorem~\ref{thm:sw}}
\label{sec:proof-sw}

\begin{reptheorem}{thm:sw}
    ML-\sidsfd (Algorithm~\ref{alg:mlsidsfd}) solves the Tracking Approximate Matrix Sketch over Sliding Window (Problem~\ref{prob:sw-fd}) with optimal space complexity $O\left(\frac{d}{\varepsilon}\log R\right)$ and near-optimal time complexity $O\left(\frac{d}{\varepsilon}\log d \log R\right)$, with success probability (the probability that the output sketch satisfies the covariance error bound) at least 
    \begin{equation*}
        \frac{99}{100}\left(1- \frac{2}{\pi\sqrt{e}}\cdot\frac{1}{\sqrt{d \log_2 d}}\right),
    \end{equation*}
    when $d$ is not very small, the success probability lower bound Eq.~\eqref{eq:prob} is close to $\frac{99}{100}$, which is a high probability.
\end{reptheorem}

The proof of Theorem~\ref{thm:sw} is divided into Proposition~\ref{prop:error-bound}, Proposition~\ref{prop:time-complexity}, Proposition~\ref{prop:sidsfd-space}, and Corollary~\ref{coro:time-space}, presented in the following subsections.

\subsection{Error Analysis}
\label{subsec:error}

First, we prove the error bound of \sidsfd, as stated in the following proposition.

\begin{proposition}
    \label{prop:error-bound}
    Let $\hat{\bm{B}}$ be the sketch of the current sliding window $(T - N,T]$ returned at line 6 of Algorithm~\ref{alg:pidsfd-query} (i.e., from $\textsc{Query}(lb = T - N,\ ub = T)$). This is our target sketch. We show that, with success probability at least \(\frac{99}{100}\left(1- \frac{2}{\pi\sqrt{e}}\cdot\frac{1}{\sqrt{d \log_2 d}}\right)\), the following error bound holds:
    \begin{equation}
        \label{eq:analy}
        \left\|\bm{A}_{T - N, T}^\top \bm{A}_{T - N, T} - \hat{\bm{B}}^\top \hat{\bm{B}}\right\|_2 \le O(1)\, \varepsilon\, \left\|\bm{A}_{T - N, T}\right\|_F^2,
    \end{equation}
    where $\bm{A}_{T - N, T}$ is the matrix formed by stacking the vectors $\bm{a}_i$ for $i = T - N+1,\ldots,T$.
\end{proposition}

\begin{proof}
    We denote the matrix $\bm{B}$ before the execution of eigen decomposition at line~5 of Algorithm~\ref{alg:pidsfd-query} over the window $(T-N, i]$ as $\bm{B}_i$. Formally,
    \begin{equation}
        \label{eq:fd-svd-2k-update-1}
        \bm{B}_i=\bm{C}_i^\top \bm{C}_i +\sum_{j=T-N+1}^{i} \left(\boldsymbol{Z}_j\boldsymbol{Z}_j^{\top}{\boldsymbol{C}_j^\prime}^{\top}\boldsymbol{C}_j^\prime+{\boldsymbol{C}_j^\prime}^{\top}\boldsymbol{C}_j^\prime\boldsymbol{Z}_j\bm{Z}_j^{\top}-\boldsymbol{Z}_j\bm{Z}_j^{\top}{\boldsymbol{C}_j^\prime}^{\top}\boldsymbol{C}_j^\prime\boldsymbol{Z}_j\bm{Z}_j^{\top}\right).\footnote{In fact, not every index \( j \) explicitly corresponds to the matrices \( \boldsymbol{Z}_j \) and \( \boldsymbol{Z}_j^{\top}{\boldsymbol{C}_j^\prime}^{\top}\boldsymbol{C}_j^\prime \) stored in the queue \( \mathcal{S} \); in such cases, we set \( \boldsymbol{Z}_j = \bm{0} \) here.}
    \end{equation}

    We introduce the per-update error as:
    \begin{equation}
        \label{eq:fd-svd-2k-update-8}
        \begin{split}
            \bm{\Delta}_{i} = \boldsymbol{a}_{i}^\top\boldsymbol{a}_{i}+\boldsymbol{B}_{i-1}-\boldsymbol{B}_{i}. \\
        \end{split}
    \end{equation}
    Here $\bm{a}_i\in \mathbb{R}^{1\times d}$ is a row vector. Then summing over the window yields
    \begin{equation}
        \label{eq:fd-svd-2k-update-12}
        \begin{split}
            \left\|\sum_{i=T-N+1}^{T} \bm{\Delta}_i\right\|_2 & \overset{Eq.~\eqref{eq:fd-svd-2k-update-8}}{=} \left\|\left(\sum_{i=T-N+1}^{T} \boldsymbol{a}_{i}^\top\boldsymbol{a}_{i}\right)-\bm{B}_{T} + \bm{B}_{T-N}\right\|_2 \\
                                                              & = \left\|\bm{A}_{T-N,T}^\top \bm{A}_{T-N,T}-\bm{B}_T+\bm{B}_{T-N}\right\|_2                                                                                                  \\
                                                              & \geq \left\|\bm{A}_{T-N, T}^\top \bm{A}_{T-N,T}-\bm{B}_T\right\|_2 - \left\|\bm{B}_{T-N}\right\|_2                                                                           \\
                                                              & \overset{Eq.~\eqref{eq:fd-svd-2k-update-1}}{=} \left\|\bm{A}_{T-N, T}^\top \bm{A}_{T-N,T}-\bm{B}_T\right\|_2 - \left\|\bm{C}_{T-N}^\top \bm{C}_{T-N}\right\|_2   .                              \\
        \end{split}
    \end{equation}
    Rearranging terms yield
    \begin{equation}
        \label{eq:fd-svd-2k-update-15}
            \left\|\bm{A}_{T-N,T}^\top \bm{A}_{T-N,T}-\bm{B}_T\right\|_2\overset{Eq.~\eqref{eq:fd-svd-2k-update-12}}{\leq}  \underbrace{\left\|\sum_{i=T-N+1}^{T} \bm{\Delta}_i\right\|_2 }_{\text{Term 1}} + \underbrace{\left\|\bm{C}_{T-N}^\top \bm{C}_{T-N}\right\|_2}_{\text{Term 2}}. 
    \end{equation}

    \paragraph{\textbf{Bounding Term 1}}

    Substituting Eq.~\eqref{eq:fd-svd-2k-update-1} into Eq.~\eqref{eq:fd-svd-2k-update-8}, we obtain:
    \begin{equation*}
            \bm{\Delta}_{i}=  \boldsymbol{a}_{i}^\top\boldsymbol{a}_{i}+\boldsymbol{C}_{i-1}^{\top}\boldsymbol{C}_{i-1}-{\boldsymbol{C}_i}^{\top}\boldsymbol{C}_{i}  - ( \boldsymbol{Z}_i \bm{Z}_i^{\top}{\boldsymbol{C}^{\prime}_i}^{\top}{\boldsymbol{C}^{\prime}_i}+{\boldsymbol{C}^{\prime}_i}^{\top}\boldsymbol{C}^{\prime}_i\boldsymbol{Z}_i \bm{Z}_i^{\top}-\boldsymbol{Z}_i \bm{Z}_i^{\top}{\boldsymbol{C}^\prime_i}^{\top}\boldsymbol{C}^\prime_i\boldsymbol{Z}_i \bm{Z}_i^{\top}). \\
    \end{equation*}

    We utilize two update relationships: $\bm{C}_i^\top \bm{C}_i = (\bm{C}_i^\prime - \bm{C}_i^\prime \bm{Z}_i \bm{Z}_i^\top )^\top (\bm{C}_i^\prime - \bm{C}_i^\prime \bm{Z}_i \bm{Z}_i^\top )$ (from line 16 of Algorithm~\ref{alg:sidsfd-update}, with $\bm{Z}_i=\bm{0}$ if line 19 is triggered), and $\tilde{\bm{C}}_i^\top \tilde{\bm{C}}_i = \bm{C}_{i-1}^\top \bm{C}_{i-1} + \bm{a}_i^\top \bm{a}_i$ (line 1). Substituting these into the preceding expression yields the following property:
    \begin{equation*}
        \begin{split}
            \bm{\Delta}_{i}  {{}={}} & \tilde{\bm{C}_{i}}^\top \tilde{\bm{C}_{i}} - (\bm{C}_{i}^\prime - \bm{C}_{i}^\prime\bm{Z}_i \bm{Z}_i^\top )^\top (\bm{C}_{i}^\prime - \bm{C}_{i}^\prime\bm{Z}_i \bm{Z}_i^\top )                                                                                                                                        \\
                               & - ( \boldsymbol{Z}_i \bm{Z}_i^{\top}{\boldsymbol{C}^{\prime}_i}^{\top}{\boldsymbol{C}^{\prime}_i}+{\boldsymbol{C}^{\prime}_i}^{\top}\boldsymbol{C}^{\prime}_i\boldsymbol{Z}_i \bm{Z}_i^{\top}-\boldsymbol{Z}_i \bm{Z}_i^{\top}{\boldsymbol{C}^\prime_i}^{\top}\boldsymbol{C}^\prime_i\boldsymbol{Z}_i \bm{Z}_i^{\top}) \\
            {{}={}}                  & \tilde{\bm{C}_{i}}^\top \tilde{\bm{C}_{i}} - {\bm{C}^\prime_i}^\top {\bm{C}^\prime_i}.                                                                                                                                                                                                                                  \\
        \end{split}
    \end{equation*}

    Lines 2-4 of Algorithm~\ref{alg:sidsfd-update} yield the following expression:
    \begin{equation}
        \label{eq:fd-svd-2k-update-10}
        \bm{\Delta}_{i}=
        \begin{cases}
            \bm{V}_i\min\left(\bm{\Sigma}_i^2, \bm{I}(\sigma_\ell^2)_i\right)\bm{V}_i^\top & \text{rows}(\tilde{\bm{C}}_i) \geq 2\ell, \\
            \bm{0}                                                                         & \text{else},                     \\
        \end{cases}
    \end{equation}
    where $\bm{\Sigma}_i^2$ is the squared singular value matrix from line 3 of the $i$-th update of Algorithm~\ref{alg:sidsfd-update}, and $(\sigma_\ell^2)_i$ denotes the $\ell$-th largest squared singular value of $\bm{\Sigma}_i^2$. (If $\text{rows}(\tilde{\bm{C}}_i) < 2\ell$, we set $(\sigma_\ell^2)_i = 0$.)

    Next, summing the matrices $\bm{\Delta}_i$ from Eq.~\eqref{eq:fd-svd-2k-update-10} over the window $(T - N,T]$ and analyzing its $\ell_2$-norm, we derive:
    \begin{equation}
        \label{eq:fd-svd-2k-update-11}
        \left\|\sum_{i=T-N+1}^{T} \bm{\Delta}_i\right\|_2  \leq \sum_{i=T-N+1}^{T} \left\|\bm{\Delta}_i\right\|_2  = \sum_{i=T-N+1}^T (\sigma_\ell^2)_i.     \\
    \end{equation}

    We now aim to bound $\sum_{i=T-N+1}^T (\sigma_\ell^2)_i$. We start from:
    \begin{equation}
        \label{eq:fd-svd-2k-update-13}
        \begin{split}
            \tr\left(\bm{B}_T\right) & =\left(\sum_{i=T-N+1}^{T} \left(\tr\left(\bm{B}_{i}\right) -\tr\left(\bm{B}_{i-1}\right) \right) \right) + \tr\left(\bm{B}_{T-N}\right)                                               \\
                                          & \overset{Eq.~\eqref{eq:fd-svd-2k-update-1}}{=} \left(\sum_{i=T-N+1}^{T} \left(\tr\left(\bm{B}_{i}\right) -\tr\left(\bm{B}_{i-1}\right) \right) \right) + \tr\left(\bm{C}_{T-N}^\top \bm{C}_{T-N}\right) \\
                                          & = \left(\sum_{i=T-N+1}^{T} \tr\left(\bm{B}_i-\bm{B}_{i-1}\right)\right) + \left\|\bm{C}_{T-N}\right\|_F^2.                                         \\
        \end{split}
    \end{equation}
    Substituting Eq.~\eqref{eq:fd-svd-2k-update-8} into Eq.~\eqref{eq:fd-svd-2k-update-13} yields
    \begin{equation*}
        \begin{split}
            \tr\left(\bm{B}_{T}\right) & = \left(\sum_{i=T-N+1}^{T} \left(\tr(\bm{a}_i^\top \bm{a}_i)-\tr(\bm{\Delta}_i)\right)\right)+\|\bm{C}_{T-N}\|_F^2 \\
                                          & \leq \|\bm{A}_{T-N, T}\|_F^2 - \ell\sum_{i=T-N+1}^T (\sigma_\ell^2)_i + \|\bm{C}_{T-N}\|_F^2,                       \\
        \end{split}
    \end{equation*}
    where the last inequality holds since $\tr\left(\bm{\Delta}_i\right) \ge \ell \cdot (\sigma_\ell^2)_i$, as inferred from Eq.~\eqref{eq:fd-svd-2k-update-10}. Rearranging terms gives:
    \begin{equation}
        \label{eq:fd-svd-2k-update-14}
        \sum_{i=T-N+1}^T (\sigma_\ell^2)_i \leq \frac{1}{\ell}\left(\|\bm{A}_{T-N,T}\|_F^2 - \tr\left(\bm{B}_{T}\right) + \|\bm{C}_{T-N}\|_F^2\right).
    \end{equation}
    Notice that from Eq.~\eqref{eq:fd-svd-2k-update-1} we have
    \begin{equation}
        \label{eq:fd-svd-2k-update-14-1}
        \begin{split}
        \tr\left(\bm{B}_{T}\right) = &\tr\left(\bm{C}_T^\top \bm{C}_T\right) +\sum_{j=T-N+1}^{i} \tr\left(\boldsymbol{Z}_j\boldsymbol{Z}_j^{\top}{\boldsymbol{C}_j^\prime}^{\top}\boldsymbol{C}_j^\prime+{\boldsymbol{C}_j^\prime}^{\top}\boldsymbol{C}_j^\prime\boldsymbol{Z}_j\bm{Z}_j^{\top}-\boldsymbol{Z}_j\bm{Z}_j^{\top}{\boldsymbol{C}_j^\prime}^{\top}\boldsymbol{C}_j^\prime\boldsymbol{Z}_j\bm{Z}_j^{\top}\right) \\
        =  &\tr\left(\bm{C}_T^\top \bm{C}_T\right) +\sum_{j=T-N+1}^{i}\left( \tr\left(\boldsymbol{Z}_j\boldsymbol{Z}_j^{\top}{\boldsymbol{C}_j^\prime}^{\top}\boldsymbol{C}_j^\prime\right)+\tr\left(\left(\bm{I}-\boldsymbol{Z}_j\bm{Z}_j^{\top}\right){\boldsymbol{C}_j^\prime}^{\top}\boldsymbol{C}_j^\prime\boldsymbol{Z}_j\bm{Z}_j^{\top}\right)\right) \\
        =  &\tr\left(\bm{C}_T^\top \bm{C}_T\right) +\sum_{j=T-N+1}^{i}\left( \tr\left(\boldsymbol{Z}_j^{\top}{\boldsymbol{C}_j^\prime}^{\top}\boldsymbol{C}_j^\prime\boldsymbol{Z}_j\right)+\tr\left(\boldsymbol{Z}_j\bm{Z}_j^{\top}\left(\bm{I}-\boldsymbol{Z}_j\bm{Z}_j^{\top}\right){\boldsymbol{C}_j^\prime}^{\top}\boldsymbol{C}_j^\prime\right)\right) \\
        =  &\left\|\bm{C}_T\right\|_F^2 +\sum_{j=T-N+1}^{i} \left\|\boldsymbol{C}_j^\prime\boldsymbol{Z}_j\right\|_F^2 \ge 0.\\
        \end{split}
    \end{equation}
    Therefore, Eq.~\eqref{eq:fd-svd-2k-update-14} can be simplified as:
    \begin{equation}
        \label{eq:fd-svd-2k-update-14-2}
        \sum_{i=T-N+1}^T (\sigma_\ell^2)_i \leq \frac{1}{\ell}\left(\|\bm{A}_{T-N,T}\|_F^2 + \|\bm{C}_{T-N}\|_F^2\right).
    \end{equation}
    
    Combining Eq.~\eqref{eq:fd-svd-2k-update-11} and Eq.~\eqref{eq:fd-svd-2k-update-14-2}, we obtain:
    \begin{equation}
        \label{eq:fd-svd-2k-update-23}
        \left\|\sum_{i=T-N+1}^{T} \bm{\Delta}_i\right\|_2 \le \frac{1}{\ell}\left(\left\|\bm{A}_{T-N,T}\right\|_F^2  + \left\|\bm{C}_{T-N}\right\|_F^2\right).
    \end{equation}

    \paragraph{\textbf{Bounding Term 2.}}
    We state the following lemma:
    \begin{lemma}
        \label{lem:power-iter}
        We have
        \begin{equation*}
            \begin{split}
                \left\|\bm{C}_{T-N}^\top \bm{C}_{T-N}\right\|_2 \leq 4\theta \leq 4 \varepsilon \left\|\bm{A}_{T-N,T}\right\|_F^2,
            \end{split}
        \end{equation*}
        with success probability at least \(\frac{99}{100}\left(1- \frac{2}{\pi\sqrt{e}}\cdot\frac{1}{\sqrt{d \log_2 d}}\right)\). 
    \end{lemma}

    The proof of Lemma~\ref{lem:power-iter} is deferred to Appendix~\ref{sec:proof-lem-power-iter}. Leveraging Eq.~\eqref{eq:fd-svd-2k-update-23} in conjunction with Lemma~\ref{lem:power-iter}, we extend the analysis initiated at Eq.~\eqref{eq:fd-svd-2k-update-15}:
    \begin{equation}
        \label{eq:fd-svd-2k-update-24}
        \begin{split}
        \text{Term 1} + \text{Term 2} =                    & \left\|\sum_{i=T-N+1}^{T} \bm{\Delta}_i\right\|_2 + \left\|\bm{C}_{T-N}^\top \bm{C}_{T-N}\right\|_2         \\
        \overset{Eq.~\eqref{eq:fd-svd-2k-update-23}}{\leq} & \frac{1}{\ell}\left\|\bm{A}_{T-N,T}\right\|_F^2 + \frac{1}{\ell}\left\|\bm{C}_{T-N}\right\|_F^2 + \left\|\bm{C}_{T-N}^\top \bm{C}_{T-N}\right\|_2     \\
        \leq                                               & \frac{1}{\ell}\left\|\bm{A}_{T-N,T}\right\|_F^2 + 3\left\|\bm{C}_{T-N}^\top \bm{C}_{T-N}\right\|_2                                                    \\
        \overset{Lemma~\ref{lem:power-iter}}{\leq}         & \frac{1}{\ell}\left\|\bm{A}_{T-N,T}\right\|_F^2 + 12\varepsilon \left\|\bm{A}_{T-N,T}\right\|_F^2,                                                      \\
        \end{split}
    \end{equation}
    where the second-to-last inequality holds because $\left\|\bm{C}_{T-N}\right\|_F^2 = \sum_{i=1}^{2\ell} \sigma_i^2 \le 2\ell \sigma_1^2 = 2\ell \left\|\bm{C}_{T-N}^\top \bm{C}_{T-N}\right\|_2$ since $\bm{C}_{T-N}$ has at most $2\ell$ rows.

    In the initialization of \sidsfd\ (Algorithm~\ref{alg:sidsfd-init}), the parameter is set to $\ell=\lceil \frac{2}{\varepsilon} \rceil$. Plugging it into Eq.~\eqref{eq:fd-svd-2k-update-24}, we obtain:
    \begin{equation*}
        \text{Term 1} + \text{Term 2} \le \frac{25}{2}\varepsilon \left\|\bm{A}_{T-N,T}\right\|_F^2 = O(1)\varepsilon \left\|\bm{A}_{T-N,T}\right\|_F^2.
    \end{equation*}

    Returning to Eq.~\eqref{eq:fd-svd-2k-update-15}, we have now established the error bound of \sidsfd\ over the sliding window, which holds with probability at least \(\frac{99}{100}\left(1- \frac{2}{\pi\sqrt{e}}\cdot\frac{1}{\sqrt{d \log_2 d}}\right)\):
    \begin{equation}\label{eq:foo-1}\left\|\bm{A}_{T-N, T}^\top \bm{A}_{T-N, T} -\bm{B}_T\right\|_2\le O(1) \varepsilon \left\|\bm{A}_{T-N, T}\right\|_F^2.\end{equation}

    Notice that $\bm{B}_T$ represents the matrix $\bm{B}$ in line~5 of Algorithm~\ref{alg:pidsfd-query}, thus we have:
    \begin{equation}
        \label{eq:fd-svd-2k-update-26}
            \bm{B}_T = \bm{V}\bm{\Lambda}\bm{V}^\top. 
   \end{equation}
   According to line~6 of Algorithm~\ref{alg:pidsfd-query}, we have:
   \begin{equation}
        \label{eq:fd-svd-2k-update-27}
        \hat{\bm{B}}^\top \hat{\bm{B}} = \bm{V}\max\left(\bm{\Lambda} - \bm{I}\lambda_\ell, \bm{0}\right)\bm{V}^\top.
   \end{equation}
   Combining Eq.~\eqref{eq:fd-svd-2k-update-26} and Eq.~\eqref{eq:fd-svd-2k-update-27}, we derive:
    \begin{equation*}
        \begin{split}
            \bm{B}_T- \hat{\bm{B}}^\top \hat{\bm{B}} = \bm{V}\min\left(\bm{\Lambda} , \bm{I}\lambda_\ell\right)\bm{V}^\top.
        \end{split}
    \end{equation*}
    Therefore, we have:
    \begin{equation}
        \label{eq:fd-svd-2k-update-28}
            \left\|\bm{B}_T - \hat{\bm{B}}^\top \hat{\bm{B}}\right\|_2 = \left\| \bm{V}\min\left(\bm{\Lambda} , \bm{I}\lambda_\ell\right)\bm{V}^\top\right\|_2=\max\left(\left|\lambda_\ell\right|, \left|\lambda_d\right|\right),
    \end{equation}
    as $\lambda_1\ge \lambda_2 \ge \dots \ge \lambda_d$ are the eigenvalues of the Hermitian matrix $\bm{B}_T$.

    We will use Weyl's inequality, stated as follows:
    \begin{theorem}[Weyl's inequality~\cite{franklin2000matrix}]
    \label{thm:wyle}
    Let $\bm{X}, \bm{Y}$ be Hermitian on inner product space $\bm{V}$ with dimension $n$, with spectrum ordered in descending order $\lambda_1 \geq \dots \geq \lambda_n$. Note that these eigenvalues can be ordered, because they are real (as eigenvalues of Hermitian matrices). Then we have the following inequality:
    \begin{equation*}
    \lambda_{i+j-1}(\bm{X}+\bm{Y}) \leq \lambda_i(\bm{X})+\lambda_j(\bm{Y}) \leq \lambda_{i+j-n}(\bm{X}+\bm{Y}).
    \end{equation*}
    \end{theorem}

    We take $\bm{X}=\bm{B}_T$ and $\bm{Y}=-\bm{A}_{T-N, T}^\top \bm{A}_{T-N, T}$ in Theorem~\ref{thm:wyle} and obtain:
    \begin{equation}
        \label{eq:fd-svd-2k-update-29}
        \begin{split}
        \lambda_{i+j-1}\left(\bm{B}_T-\bm{A}_{T-N, T}^\top \bm{A}_{T-N, T}\right) & \leq \lambda_i\left(\bm{B}_T\right)+\lambda_j\left(-\bm{A}_{T-N, T}^\top \bm{A}_{T-N, T}\right)\\ & \leq \lambda_{i+j-d}\left(\bm{B}_T-\bm{A}_{T-N, T}^\top \bm{A}_{T-N, T}\right).
        \end{split}
    \end{equation}
    Notice that from Eq.~\eqref{eq:foo-1}, we have:
    \begin{equation*}
    \forall i,\quad -O(1)\frac{1}{\ell}\left\|\bm{A}_{T-N, T}\right\|_F^2\le \lambda_{i}\left(\bm{A}_{T-N, T}^\top \bm{A}_{T-N, T}-\bm{B}_T\right)\le O(1)\frac{1}{\ell}\left\|\bm{A}_{T-N, T}\right\|_F^2.
    \end{equation*}
    And notice that 
    \begin{equation*}
    \forall i, \quad -\frac{1}{d-i+1}\left\|\bm{A}_{T-N, T}\right\|_F^2\le\lambda_i\left(-\bm{A}_{T-N, T}^\top \bm{A}_{T-N, T}\right)\le 0.
    \end{equation*}
    Setting $i=d$ and $j=1$, we have from Eq.~\eqref{eq:fd-svd-2k-update-29} that
    \begin{equation*}
        -O(1)\frac{1}{\ell}\left\|\bm{A}_{T-N, T}\right\|_F^2\le\lambda_{d}\left(\bm{B}_T\right)\le \left(O(1)\frac{1}{\ell}+\frac{1}{d}\right)\left\|\bm{A}_{T-N, T}\right\|_F^2\le O(1)\frac{1}{\ell}\left\|\bm{A}_{T-N, T}\right\|_F^2,
    \end{equation*}
    where the last inequality holds since $\ell \le d$, otherwise, there will be no $\ell$-th largest eigenvalue.
    Setting $i=\ell$ and $j=d-\ell+1$, we have from Eq.~\eqref{eq:fd-svd-2k-update-29} that
    \begin{equation*}
        -O(1)\frac{1}{\ell}\left\|\bm{A}_{T-N, T}\right\|_F^2\le\lambda_{\ell}\left(\bm{B}_T\right)\le O(1)\frac{1}{\ell}\left\|\bm{A}_{T-N, T}\right\|_F^2.
    \end{equation*}
    Plugging the above two inequalities into Eq.~\eqref{eq:fd-svd-2k-update-28}, we obtain:
    \begin{equation}
        \label{eq:fd-svd-2k-update-32}
            \left\|\bm{B}_T - \hat{\bm{B}}^\top \hat{\bm{B}}\right\|_2 \le O(1)\frac{1}{\ell}\left\|\bm{A}_{T-N, T}\right\|_F^2 = O(1) \varepsilon \left\|\bm{A}_{T-N, T}\right\|_F^2.
    \end{equation}
    Finally, combining Eq.~\eqref{eq:foo-1} and Eq.~\eqref{eq:fd-svd-2k-update-32}, we conclude the proof:
    \begin{equation*}
        \left\|\bm{A}_{T-N, T}^\top \bm{A}_{T-N, T} -\hat{\bm{B}}^\top \hat{\bm{B}}\right\|_2\le \left\|\bm{A}_{T-N, T}^\top \bm{A}_{T-N, T} -\bm{B}_T\right\|_2+\left\|\bm{B}_T - \hat{\bm{B}}^\top \hat{\bm{B}}\right\|_2\le O(1)\varepsilon \left\|\bm{A}_{T-N, T}\right\|_F^2.\qedhere
    \end{equation*}
\end{proof}

\subsection{Time Complexity}

We now analyze the time complexity of each update step in \sidsfd\ (Algorithm~\ref{alg:sidsfd-update}).

\begin{proposition}
    \label{prop:time-complexity}
    The amortized time cost of each update procedure of \sidsfd (Algorithm~\ref{alg:sidsfd-update}) is $O\left(\frac{d}{\varepsilon}\log d\right)$.
\end{proposition}

\begin{proof}
    If the number of rows in $\tilde{\bm{C}}_i$ reaches $2\ell$, we perform SVD at line 3, which takes $O(d\ell^2)$ time. After this step, the number of rows in $\bm{C}_i^\prime$ is reduced to $\ell$. Thus, SVD is executed once every $\ell$ updates, and the amortized time cost of lines 2-6 is $O(d\ell^2)/\ell = O(d\ell)$.

    According to the time complexity result in Corollary~\ref{thm:power}, the power iteration at line 7 takes $O(d\ell k) = O(d\ell \log d)$ time, where $k=\lceil\log_2 d \rceil+ 1$ as defined in Algorithm~\ref{alg:sidsfd-update}.

    The analysis of the time cost for the for-loop (lines~9-17) of Simultaneous Iterations is more subtle. During this loop, we compute approximate singular values for increasing $k=2^j$ ($j$ is the loop index) by executing Simultaneous Iteration (line~11) on ${\bm{C}_i^\prime}^\top \bm{C}_i^\prime$ -- specifically, the top-$k$ approximate singular values for $k=2,4,8,\dots,2^j,\dots$. Suppose that in the $j$-th iteration, the if-statement (line~12) finds that the (approximate) $k = 2^j$-th singular value of ${\bm{C}_i^\prime}^\top \bm{C}_i^\prime$ is smaller than $\theta$. In that case, we execute lines 13-17 and terminate the for-loop.

    The following property holds for the approximate singular values obtained from Simultaneous Iteration on $\bm{C}_i^\prime$ as we do not trigger the if-statement in line~12 for $k/2$ but trigger it for $k$:
    \begin{equation*}
        \hat{\sigma}_{k/2}^2(\bm{C}_i^\prime) = \hat{\sigma}_{2^{j-1}}^2(\bm{C}_i^\prime)\geq \hat{\sigma}_{\xi}^2(\bm{C}_i^\prime) \geq \theta .
    \end{equation*}
    We further get the following bound:
    \begin{equation}
        \label{eq:fd-svd-2k-bound}
        \begin{split}
               \left\|\bm{C}_i^\prime\right\|_F^2-\left\|\bm{C}_i^\prime-\bm{C}_i^\prime\bm{Z}_i\bm{Z}_i^\top\right\|_F^2 = & \tr \left({\bm{C}_i^\prime}^\top\bm{C}_i^\prime\right) -\tr\left(\left(\bm{I}-\bm{Z}_i\bm{Z}_i^\top\right){\bm{C}_i^\prime}^\top\bm{C}_i^\prime\left(\bm{I}-\bm{Z}_i\bm{Z}_i^\top\right)\right)                                                                                                                              \\
            = & \tr \left({\bm{C}_i^\prime}^\top\bm{C}_i^\prime\right) -\tr\left({\bm{C}_i^\prime}^\top\bm{C}_i^\prime\left(\bm{I}-\bm{Z}_i\bm{Z}_i^\top\right)\left(\bm{I}-\bm{Z}_i\bm{Z}_i^\top\right)\right)                                                                                                                              \\
            = & \tr \left({\bm{C}_i^\prime}^\top\bm{C}_i^\prime\right) -\tr\left({\bm{C}_i^\prime}^\top\bm{C}_i^\prime\left(\bm{I}-\bm{Z}_i\bm{Z}_i^\top\right)\right)                                                                                                                                                                       \\
            = & \tr\left({\bm{C}_i^\prime}^\top\bm{C}_i^\prime\bm{Z}_i\bm{Z}_i^\top\right)                                                                                                                       =     \tr\left({\bm{C}_i^\prime}^\top\bm{C}_i^\prime\bm{Z}_i\bm{Z}_i^\top\bm{Z}_i\bm{Z}_i^\top\right)                       \\
            = & \tr\left(\bm{Z}_i\bm{Z}_i^\top{\bm{C}_i^\prime}^\top\bm{C}_i^\prime\bm{Z}_i\bm{Z}_i^\top\right)                                                                                                  =     \sum_{i=1}^{\xi} \hat{\sigma}_i^2(\bm{C}_i^\prime)                                               \geq  2^{j-1} \theta. \\
        \end{split}
    \end{equation}

    Substituting $\bm{C}_i = \bm{C}_i^\prime - \bm{C}_i^\prime \bm{Z}_i \bm{Z}_i^\top$ into Eq.~\eqref{eq:fd-svd-2k-bound}, we obtain the following relationship between the two updates:
    \begin{equation}
        \label{eq:fd-svd-2k-update}
        \begin{split}
            \|\bm{C}^{\prime}_{i}\|_F^2 - \|\bm{C}_i\|_F^2 & \ge 2^{j-1}\theta,                                    \\
            \|\bm{C}^{\prime}_{i}\|_F^2                    & \le \|\bm{C}_{i-1}\|_F^2 + \bm{a}_{i}\bm{a}_{i}^\top,
        \end{split}
    \end{equation}
    where the “$<$” relation in the last inequality holds when the if-branch at lines 2-4 is executed. Eq.~\eqref{eq:fd-svd-2k-update} can then be rewritten as:
    \begin{equation}
        \label{eq:fd-svd-2k-update-2}
        \|\bm{C}_{i-1}\|_F^2 + \bm{a}_i\bm{a}_i^\top - \|\bm{C}_{i}\|_F^2 \ge 2^{j-1}\theta = \frac{k_i \theta}{2},
    \end{equation}
    where $k_i$ substitutes $2^j$ at line 12 of Algorithm~\ref{alg:sidsfd-update}, representing the upper bound on the $k$ at line~10 of the $i$-th update.

    Summing Eq.~\eqref{eq:fd-svd-2k-update-2} from $i = T - N + 1$ to $i = T$ (over the active sliding window; snapshots before this window have already expired via lines 7-8 of Algorithm~\ref{alg:mlsidsfd}), we obtain:
    \begin{equation*}
               \left(\sum_{i=T-N+1}^{T} \bm{a}_i\bm{a}_i^\top\right) - \|\bm{C}_{T}\|_F^2 + \|\bm{C}_{T-N}\|_F^2 =  \|\bm{A}_{T-N,T}\|_F^2-\|\bm{C}_{T}\|_F^2 + \|\bm{C}_{T-N}\|_F^2                                   \ge  \frac{\theta}{2} \sum_{i=T-N+1}^{T} k_i, \\
    \end{equation*}
    which can be reformulated as:
    \begin{equation}
        \label{eq:fd-svd-2k-update-4}
            \sum_{i=T-N+1}^{T} k_i \leq  \frac{2}{\theta} \left(\left\|\bm{A}_{T-N,T}\right\|_F^2 + \|\bm{C}_{T-N}\|_F^2 \right) \leq \frac{2}{\theta}\left(\left\|\bm{A}_{T-N,T}\right\|_F^2+8\ell\theta\right),
    \end{equation}
    where the last inequality holds with high probability, according to Lemma~\ref{lem:power-iter}, and that $\|\bm{C}_{T-N}\|_F^2 \le 2\ell \|\bm{C}_{T-N}^\top \bm{C}_{T-N}\|_2 \overset{w.h.p.}{\leq} 8\ell\theta$.

    By Fact~\ref{prop:ml}, plugging $\theta \le \varepsilon \|\bm{A}_{T-N,T}\|_F^2 < 2\theta$ into Eq.~\eqref{eq:fd-svd-2k-update-4}, we derive the following bound:
    \begin{equation}
        \label{eq:fd-svd-2k-update-25}
        \begin{split}
            \sum_{i=T-N+1}^{T} k_i \leq \frac{2}{\theta} \left(\frac{2\theta}{\varepsilon} + 8\ell\theta\right)=\frac{4}{\varepsilon} +16\ell = O\left(\frac{1}{\varepsilon}\right),
        \end{split}
    \end{equation}
    where the last equality holds because \sidsfd\ is initialized with $\ell = \lceil \frac{2}{\varepsilon} \rceil$ in Algorithm~\ref{alg:sidsfd-init}.

    For the $i$-th update step, the time cost of the Simultaneous Iteration procedure (line 11) is $O(d\ell (k_i + k_i/2 + k_i/4 + \cdots) \log \min\left(2\ell, d\right)) = O(d\ell k_i \log d)$. The matrix operations of lines 15-16 take $O(d\ell k_i)$ time. Therefore, the time within the for-loop is dominated by Simultaneous Iteration, yielding a total cost of $O(d\ell k_i \log d)$.

    Summing the time cost of the for-loop (lines 9-17) over all $N$ updates, we get the total time cost of $N$ updates within a sliding window:
    \begin{equation*}
        \sum_{i=T-N+1}^{T} O(d\ell k_i \log d) = O(d\ell \log d) \sum_{i=T-N+1}^{T} k_i = O\left(\frac{d\ell}{\varepsilon}  \log d\right),
    \end{equation*}
    where the last derivation holds according to Eq.~\eqref{eq:fd-svd-2k-update-25}.

    $O\left(\frac{d\ell}{\varepsilon} \log d\right)$ is the total time cost of $N$ updates within a sliding window. Next, we analyze the amortized time cost per update step. Note that we trivially assume $N > \ell$ (otherwise, if $N \le \ell$, we can directly store all $\ell$ rows exactly, so the assumption $N > \ell\Rightarrow \frac{\ell}{N}<1$ is without loss of generality). Amortizing over $N$ updates, the per-update time cost of Algorithm~\ref{alg:sidsfd-update} becomes:
    \begin{equation*}
        \frac{1}{N} O\left(\frac{d\ell}{\varepsilon} \log d\right) = O\left(\frac{d}{\varepsilon} \log d\right).
    \end{equation*}
    Therefore, the amortized time complexity of Algorithm~\ref{alg:sidsfd-update} is $O\left(\frac{d}{\varepsilon} \log d\right)$. %
\end{proof}

\subsection{Space Complexity}

\begin{proposition}
\label{prop:sidsfd-space}
The space complexity of \sidsfd\ is $O\left(\frac{d}{\varepsilon}\right)$.
\end{proposition}

\begin{proof}
The space cost of the algorithm is dominated by the sketch matrix $\bm{C}$ and the snapshots stored in the queue $\mathcal{S}$. The sketch matrix $\bm{C}$ requires $O(d\ell)$ space, as its maximum size is $d \times 2\ell$. The queue $\mathcal{S}$ uses $O\left(d\sum_{i=T-N+1}^{T}k_i\right)$ space, since it stores matrices $\bm{Z}_i \in \mathbb{R}^{d \times k_i}$ and $\bm{Z}_i^\top {\bm{C}_i^\prime}^\top \bm{C}_i^\prime \in \mathbb{R}^{k_i \times d}$ at line 15.

Summing over all $N$ updates, and using the bound from Eq.~\eqref{eq:fd-svd-2k-update-25}, the space cost of the queue $\mathcal{S}$ is:
\begin{equation*}
    O\left(d \sum_{i=T-N+1}^{T} k_i\right) = O\left(\frac{d}{\varepsilon}\right).
\end{equation*}

Thus, the total space cost of \sidsfd is $O\left(d\ell + \frac{d}{\varepsilon}\right) = O\left(\frac{d}{\varepsilon}\right)$.
\end{proof}

\begin{corollary}
\label{coro:time-space}
As Algorithm~\ref{alg:mlsidsfd} uses $\lceil\log_2 R\rceil$ levels of the \sidsfd\ framework, its overall space complexity is $O\left(\frac{d}{\varepsilon} \log R\right)$, and its amortized time complexity is 
$O\left(\frac{d}{\varepsilon} \log d \log R\right)$, according to Propositions~\ref{prop:time-complexity} and~\ref{prop:sidsfd-space}.
\end{corollary}

\subsection{Proof for Lemma~\ref{lem:power-iter}}
\label{sec:proof-lem-power-iter}

    \begin{replem}{lem:power-iter}
        We have
        \begin{equation*}
            \begin{split}
                \left\|\bm{C}_{T-N}^\top \bm{C}_{T-N}\right\|_2 \leq 4\theta \leq 4 \varepsilon \left\|\bm{A}_{T-N,T}\right\|_F^2,
            \end{split}
        \end{equation*}
        with success probability at least \(\frac{99}{100}\left(1- \frac{2}{\pi\sqrt{e}}\cdot\frac{1}{\sqrt{d \log_2 d}}\right)\). 
    \end{replem}
\begin{proof}
    To prove the lower bound of $\Pr\left[ \left\|\bm{C}_{T-N}^\top \bm{C}_{T-N}\right\|_2 \le 4\theta\right]$, we consider the upper bound of $\Pr\left[ \left\|\bm{C}_{T-N}^\top \bm{C}_{T-N}\right\|_2 \geq 4\theta\right]$. We only consider the scenario that the squared largest singular value of $\bm{C}^\prime_i$ exceeds $4\theta$, but after the update the squared largest singular value of $\bm{C}_i$ is still larger than $4\theta$, i.e., $$\Pr\left[\left\|\bm{C}_{T-N}^\top \bm{C}_{T-N}\right\|_2 \geq 4\theta \mid \|{(\bm{C}^\prime_{T-N}})^\top \bm{C}^\prime_{T-N}\|_2 \geq 4\theta\right].$$

    This event can be decomposed into either (1) the power iteration at line 9 of Algorithm~\ref{alg:sidsfd-update} fails to converge (thus the if-statement is not triggered), more specifically, the probability that the largest singular value of $({\bm{C}_{T-N}^\prime})^\top \bm{C}_{T-N}^\prime$ is larger than $4\theta$ but the estimated largest singular value $\hat{\sigma}_1^2$ returned by power iteration method is smaller than $\theta/2$ (i.e., $\hat{\sigma}_1^2\le\theta/2$). The probability of this event is bounded by the following bound according to Corollary~\ref{thm:power}:
    \begin{equation}
        \label{eq:fd-svd-2k-update-17}
        \begin{split}
        &\Pr\left[ \left\|\bm{C}_{T-N}^\top \bm{C}_{T-N}\right\|_2 \geq 4\theta\mid \hat{\sigma}_1^2 \leq \theta/2\  \cap\  \left\|{\bm{C}^\prime_{T-N}}^\top \bm{C}^\prime_{T-N}\right\|_2 \geq 4\theta\right] =  1            ,          \\
        &\Pr\left[ \hat{\sigma}_1^2 \leq \theta/2 \mid \left\|({\bm{C}^\prime_{T-N}})^\top \bm{C}^\prime_{T-N}\right\|_2 \geq 4\theta\right] \leq                                                        {\frac{2}{\pi\sqrt{e}}\cdot\frac{1}{\sqrt{d \log_2 d}}},
        \end{split}
    \end{equation}
    or (2) the power iteration succeeds in converging ($\hat{\sigma}_1^2>\theta/2$) but the simultaneous iteration at line 11 of Algorithm~\ref{alg:sidsfd-update} fails to converge, more specifically, the returned $\bm{Z}^\prime$ of simultaneous iteration method is not a ``good'' approximation of the larger singular vectors. That is to analyze the probability:
    \begin{equation}
        \label{eq:fd-svd-2k-update-18}
        \begin{split}
                & \Pr\left[ \left\|\bm{C}_{T-N}^\top \bm{C}_{T-N}\right\|_2 \geq 4\theta\mid \hat{\sigma}_1^2 > \theta/2 \ \cap \ \left\|({\bm{C}^\prime_{T-N}})^\top \bm{C}^\prime_{T-N}\right\|_2 \geq 4\theta \right]          .                                \\
        \end{split}
    \end{equation}

    By the Spectral norm error of the simultaneous iteration (Theorem~\ref{thm:simul} (2) and (3)), we have the following bound, where $\sigma_i$ and $\hat{\sigma}_i$ denote the $i$-th largest exact and approximate singular values (given by Simultaneous Iteration algorithm) of $\bm{C}^\prime_{T-N}$, respectively:
    \begin{equation*}
        \Pr\left[\left\|\boldsymbol{C}_{T-N}^\prime (\boldsymbol{I}-\boldsymbol{Z Z}^{\top})\right\| _{2}\le(1+\varepsilon_{\text{SI}})\sigma_{\xi+1}  \cap \ \left|\sigma_{\xi+1}^2-\hat{\sigma}_{\xi+1}^2     \right|\le\varepsilon_{\text{SI}}\cdot\sigma_{\xi+1}^2\right] \ge 99/100,
    \end{equation*}
    which can be rewritten as
    \begin{equation*}
        \Pr\left[\left\|\boldsymbol{C}_{T-N}^\prime (\boldsymbol{I}-\boldsymbol{Z Z}^{\top})\right\| _{2}\le\frac{1+\varepsilon_{\text{SI}}}{\sqrt{1-\varepsilon_{\text{SI}}}}\hat{\sigma}_{\xi+1}\right] \ge 99/100.
    \end{equation*}
    Because line 13 of Algorithm~\ref{alg:sidsfd-update} guarantees $\hat{\sigma}_{\xi+1}^2< \theta$, we have
    \begin{equation}
        \label{eq:fd-svd-2k-update-19-2}
        \begin{split}
            & \phantom{{}={}} \Pr\left[ \left\|\bm{C}_{T-N}^\top \bm{C}_{T-N}\right\|_2 \leq \frac{(1+\varepsilon_{\text{SI}})^2}{1-\varepsilon_{\text{SI}}}\theta\right]                                                                                                                          \\
        & = \Pr\left[\left\| (\boldsymbol{I}-\boldsymbol{Z Z}^{\top}) \left(\boldsymbol{C}_{T-N}^\prime\right)^{\top}\boldsymbol{C}_{T-N}^\prime (\boldsymbol{I}-\boldsymbol{Z Z}^{\top})\right\| _{2}\le \frac{(1+\varepsilon_{\text{SI}})^2}{1-\varepsilon_{\text{SI}}}\theta\right]\\
        & = \Pr\left[\left\|\boldsymbol{C}_{T-N}^\prime (\boldsymbol{I}-\boldsymbol{Z Z}^{\top})\right\| _{2}\le\frac{1+\varepsilon_{\text{SI}}}{\sqrt{1-\varepsilon_{\text{SI}}}}\sqrt{\theta}\right]\\
        & \ge \Pr\left[\left\|\boldsymbol{C}_{T-N}^\prime (\boldsymbol{I}-\boldsymbol{Z Z}^{\top})\right\| _{2}\le\frac{1+\varepsilon_{\text{SI}}}{\sqrt{1-\varepsilon_{\text{SI}}}}\hat{\sigma}_{\xi+1}\right]\\
        & \ge 99/100.
        \end{split}
    \end{equation}

    Take $\varepsilon_{\text{SI}}=0.4$, then we have $4\ge \frac{(1+\varepsilon_{\text{SI}})^2}{1-\varepsilon_{\text{SI}}}$ and Eq.~\eqref{eq:fd-svd-2k-update-19-2} can be rewritten as
    \begin{equation*}
            \Pr\left[ \left\|\bm{C}_{T-N}^\top \bm{C}_{T-N}\right\|_2 \leq 4\theta\right] \ge  \Pr\left[ \left\|\bm{C}_{T-N}^\top \bm{C}_{T-N}\right\|_2 \leq \frac{(1+\varepsilon_{\text{SI}})^2}{1-\varepsilon_{\text{SI}}}\theta\right] \ge 99/100.
    \end{equation*}
    Finally the probability in Eq.~\eqref{eq:fd-svd-2k-update-18} is
    \begin{equation}
        \label{eq:fd-svd-2k-update-18-1}
                 \Pr\left[ \left\|\bm{C}_{T-N}^\top \bm{C}_{T-N}\right\|_2 \geq 4\theta\mid \hat{\sigma}_1^2 > \theta/2 \ \cap \ \|{\bm{C}^\prime_{T-N}}^\top \bm{C}^\prime_{T-N}\|_2 \geq 4\theta \right] < 1/100.
    \end{equation}

    Combining Eq.~\eqref{eq:fd-svd-2k-update-17} and Eq.~\eqref{eq:fd-svd-2k-update-18-1} and using the law of total probability, we have the following bound:
    \begin{equation*}
        \begin{split}
              & \Pr\left[ \left\|\bm{C}_{T-N}^\top \bm{C}_{T-N}\|_2 \geq 4\theta\mid \|{\bm{C}^\prime_{T-N}}^\top \bm{C}^\prime_{T-N}\right\|_2 \geq 4\theta \right]                                            \\
            {{}={}} & \Pr\left[ \left\|\bm{C}_{T-N}^\top \bm{C}_{T-N}\|_2 \geq 4\theta \ \cap \ \hat{\sigma}_1^2 \leq \sigma_1/2 \mid \|{\bm{C}^\prime_{T-N}}^\top \bm{C}^\prime_{T-N}\right\|_2 \geq 4\theta \right] \\
              & + \Pr\left[ \left\|\bm{C}_{T-N}^\top \bm{C}_{T-N}\|_2 \geq 4\theta \ \cap \ \hat{\sigma}_1^2 > \sigma_1/2 \mid \|{\bm{C}^\prime_{T-N}}^\top \bm{C}^\prime_{T-N}\right\|_2 \geq 4\theta \right]  \\
            {{}={}} & \Pr\left[ \left\|\bm{C}_{T-N}^\top \bm{C}_{T-N}\|_2 \geq 4\theta\mid \hat{\sigma}_1^2 \leq \sigma_1/2 \cap \|{\bm{C}^\prime_{T-N}}^\top \bm{C}^\prime_{T-N}\right\|_2 \geq 4\theta \right]      \\
              & \times \Pr\left[\hat{\sigma}_1^2 \leq 4\sigma_1/2\mid \left\|{\bm{C}^\prime_{T-N}}^\top \bm{C}^\prime_{T-N}\right\|_2 \geq 4\theta \right]                                                      \\
              & + \Pr\left[ \left\|\bm{C}_{T-N}^\top \bm{C}_{T-N}\|_2 \geq 4\theta\mid \hat{\sigma}_1^2 > \sigma_1/2\ \cap \ \|{\bm{C}^\prime_{T-N}}^\top \bm{C}^\prime_{T-N}\right\|_2 \geq 4\theta \right]    \\
              & \times \Pr\left[\hat{\sigma}_1^2 > \sigma_1/2\mid \left\|{\bm{C}^\prime_{T-N}}^\top \bm{C}^\prime_{T-N}\right\|_2 \geq 4\theta \right]                                                          \\
               {{}\le{}} &{\frac{2}{\pi\sqrt{e}}\cdot\frac{1}{\sqrt{d \log_2 d}}} + \frac{1}{100}\left(1- {\frac{2}{\pi\sqrt{e}}\cdot\frac{1}{\sqrt{d \log_2 d}}}\right)                                                                                          \\
               {{}={}} &1-\frac{99}{100}\left(1- {\frac{2}{\pi\sqrt{e}}\cdot\frac{1}{\sqrt{d \log_2 d}}}\right).
        \end{split}
    \end{equation*}
    Then we have
    \begin{equation*}
        \begin{split}
            \Pr\left[ \left\|\bm{C}_{T-N}^\top \bm{C}_{T-N}\right\|_2 \geq 4\theta\right] & \le  \Pr\left[ \left\|\bm{C}_{T-N}^\top \bm{C}_{T-N}\|_2 \geq  4\theta\mid \|{\bm{C}^\prime_{T-N}}^\top \bm{C}^\prime_{T-N}\right\|_2 \geq 4\theta \right] \\
            & \leq 1-\frac{99}{100}\left(1- {\frac{2}{\pi\sqrt{e}}\cdot\frac{1}{\sqrt{d \log_2 d}}}\right).                                            \\
        \end{split}
    \end{equation*}
    
    Finally, we conclude that with probability at least \begin{equation*}\frac{99}{100}\left(1- {\frac{2}{\pi\sqrt{e}}\cdot\frac{1}{\sqrt{d \log_2 d}}}\right),\end{equation*} we have
    \begin{equation*}
        \left\|\bm{C}_{T-N}^\top \bm{C}_{T-N}\right\|_2 \leq 4\theta \leq  4\varepsilon \left\|\bm{A}_{T-N,T}\right\|_F^2.\qedhere
    \end{equation*}
\end{proof}

\section{Analysis of the Amplified Estimation Procedure (Algorithm~\ref{alg:amplify})}
\label{app:amplified-estimation}

Consider the $c_{hg}$ at line 8 in Algorithm~\ref{alg:amplify}. From Corollary~\ref{thm:power}, if $d\ge 2$, we have
\begin{equation*}
    \Pr\left[\frac{1}{\sqrt{2}}\left\|\bm{C}_i^\prime-\bm{C}_i^\prime\bm{Z}_h\bm{Z}_h^\top\right\|_2\le c_{hg}\le\left\|\bm{C}_i^\prime-\bm{C}_i^\prime\bm{Z}_h\bm{Z}_h^\top\right\|_2\right] \ge 1- \frac{2}{\pi\sqrt{e}}\cdot\frac{1}{\sqrt{d \log_2 d}}\ge \frac{2}{3}.
\end{equation*}
Since we take $c_h = \max_{g \in [s]} c_{hg}$, repeating the estimator $s$ times amplifies its success probability, yielding
\begin{equation*}
    \Pr\left[\frac{1}{\sqrt{2}}\left\|\bm{C}_i^\prime-\bm{C}_i^\prime\bm{Z}_h\bm{Z}_h^\top\right\|_2\le c_{h}\le\left\|\bm{C}_i^\prime-\bm{C}_i^\prime\bm{Z}_h\bm{Z}_h^\top\right\|_2\right] \ge 1-\left(\frac{1}{3}\right)^s.
\end{equation*}

We now ensure that all $r$ outer rounds simultaneously produce accurate estimates. 
Let $M$ denote this event, i.e.,
\begin{equation*}
    \bigcap_{h=1}^{r} \left(\frac{1}{\sqrt{2}}\left\|\bm{C}_i^\prime-\bm{C}_i^\prime\bm{Z}_h\bm{Z}_h^\top\right\|_2\le c_{h}\le\left\|\bm{C}_i^\prime-\bm{C}_i^\prime\bm{Z}_h\bm{Z}_h^\top\right\|_2\right).
\end{equation*}
By a union bound, we have
\begin{equation*}
    \begin{split}
    \Pr\left[M\right] = 1-\Pr\left[\overline{M}\right] =& 1-\Pr\left[\bigcup_{h=1}^{r} \left(c_h<\frac{1}{\sqrt{2}}\left\|\bm{C}_i^\prime-\bm{C}_i^\prime\bm{Z}_h\bm{Z}_h^\top\right\|_2 \cup c_h>\left\|\bm{C}_i^\prime-\bm{C}_i^\prime\bm{Z}_h\bm{Z}_h^\top\right\|_2\right)\right]\\
    \ge & 1-\sum_{h=1}^{r}\Pr\left[c_h<\frac{1}{\sqrt{2}}\left\|\bm{C}_i^\prime-\bm{C}_i^\prime\bm{Z}_h\bm{Z}_h^\top\right\|_2\cup c_h>\left\|\bm{C}_i^\prime-\bm{C}_i^\prime\bm{Z}_h\bm{Z}_h^\top\right\|_2\right] \\
    = & 1- \sum_{h=1}^{r} \left(1-\Pr\left[\frac{1}{\sqrt{2}}\left\|\bm{C}_i^\prime-\bm{C}_i^\prime\bm{Z}_h\bm{Z}_h^\top\right\|_2\le c_{h}\le\left\|\bm{C}_i^\prime-\bm{C}_i^\prime\bm{Z}_h\bm{Z}_h^\top\right\|_2\right]\right) \\
    \ge & 1-r\cdot \left(\frac{1}{3}\right)^s.
    \end{split}
\end{equation*}

Conditioned on event $M$, the ordering between the true residual norms and their estimates is preserved up to a factor $\sqrt{2}$. 
Let $r^\prime=\arg\min_{h\in [r]} c_h$ and 
$r^*=\arg\min_{h\in [r]}\left\|\bm{C}_i^\prime-\bm{C}_i^\prime\bm{Z}_{h}\bm{Z}_{h}^\top\right\|_2$. 
Then
\begin{equation*}
    \min_{h\in [r]} \left\|\bm{C}_i^\prime-\bm{C}_i^\prime\bm{Z}_{h}\bm{Z}_{h}^\top\right\|_2 \ge c_{r^*} \ge c_{r^\prime} \ge \frac{1}{\sqrt{2}}\left\|\bm{C}_i^\prime-\bm{C}_i^\prime\bm{Z}_{r^\prime}\bm{Z}_{r^\prime}^\top\right\|_2.
\end{equation*}
We now substitute the chosen parameters $r=\log_{100}\frac{2}{\delta}$ and $s=2\log_3\frac{2}{\delta}$ (line~1 of Algorithm~\ref{alg:amplify}). 
Combining the above inequality with the failure probability of simultaneous iteration gives
\begin{equation*}
    \begin{split}
        \Pr\left[\left\|\bm{C}_i^\prime-\bm{C}_i^\prime\bm{Z}_{r^\prime}\bm{Z}_{r^\prime}^\top\right\|_2 \ge \sqrt{2} \frac{1+\varepsilon_{\text{SI}}}{\sqrt{1-\varepsilon_{\text{SI}}}}\sqrt{\theta}\right] \le&\Pr\left[\left\|\bm{C}_i^\prime-\bm{C}_i^\prime\bm{Z}_{r^\prime}\bm{Z}_{r^\prime}^\top\right\|_2 \ge \sqrt{2} \frac{1+\varepsilon_{\text{SI}}}{\sqrt{1-\varepsilon_{\text{SI}}}}\sqrt{\theta}\mid M\right]+\Pr\left[\overline{M}\right]\\
        \le & \Pr\left[\min_{h\in[r]}\left\|\bm{C}_i^\prime-\bm{C}_i^\prime\bm{Z}_{r}\bm{Z}_{r}^\top\right\|_2 \ge \frac{1+\varepsilon_{\text{SI}}}{\sqrt{1-\varepsilon_{\text{SI}}}}\sqrt{\theta}\right] + r\cdot \left(\frac{1}{3}\right)^s \\
        \le & \left(\frac{1}{100}\right)^r + r\cdot \left(\frac{1}{3}\right)^s\le \delta.
    \end{split}
\end{equation*}
Finally, setting $\varepsilon_{\text{SI}}=0.2$ (line~3 in Algorithm~\ref{alg:amplify}) yields
\begin{equation}
    \label{eq:amplify}
    \Pr\left[\left\|\bm{C}_i^\top\bm{C}_i\right\|_2 \le 4\theta\right]=\Pr\left[\left\|\bm{C}_i^\prime-\bm{C}_i^\prime\bm{Z}_{r^\prime}\bm{Z}_{r^\prime}^\top\right\|_2^2 \le 4\theta\right] \ge 1-\delta.
\end{equation}
Subsequently, by replacing Lemma~\ref{lem:power-iter} used in the proof of Section~\ref{subsec:error} with Eq.~\eqref{eq:amplify}, while keeping the rest of the argument unchanged, we obtain the error and probability guarantees stated in the Corollary~\ref{coro:amplify}.

Algorithm~\ref{alg:amplify} does not incur any additional space overhead. 
In terms of time complexity, for each update, the total cost of the Simultaneous Iteration in the $r$ rounds (line~3) is $O\left(r\cdot d\ell k_i \log d\right)=O\left(d\ell k_i \log d \log\frac{1}{\delta}\right)$. 

For each update, the $r \times s$ rounds of Power Iteration (line~7) incur a cost of $$O\left(\log k_i \cdot r\cdot s \cdot d\ell \log d \right)=O\left(d\ell k_i \log d \log^2 \frac{1}{\delta} \right).$$
Therefore, the amortized time complexity per update is dominated by the cost of the Power Iteration and, according to Eq.~\eqref{eq:fd-svd-2k-update-25}, is equal to
$$\frac{1}{N}\sum_{i=T-N+1}^{T} O\left(d\ell k_i \log d \log^2 \frac{1}{\delta} \right) =\frac{1}{N} O\left(\frac{d\ell}{\varepsilon} \log d \log^2 \frac{1}{\delta} \right)=O\left(\frac{d}{\varepsilon} \log d \log^2 \frac{1}{\delta} \right). $$

ML-\sidsfd (Algorithm~\ref{alg:mlsidsfd}) has $L=O(\log R)$ levels, so the total amortized time complexity per update of ML-\sidsfd is
$$O\left(L \cdot \frac{d}{\varepsilon} \log d \log^2 \frac{1}{\delta} \right)=O\left(\frac{d}{\varepsilon} \log R \log d \log^2 \frac{1}{\delta} \right). $$

\section{Proof of Theorem~\ref{thm:attp}}
\label{app:attp}

\begin{reptheorem}{thm:attp}
        Algorithm~\ref{alg:attp} solves the ATTP Matrix Sketch (Problem~\ref{prob:attp}) with space complexity $O\!\left(\frac{d}{\varepsilon}\log \left\|\bm{A}\right\|_F^2\right)$ and amortized time complexity $O\!\left(\frac{d}{\varepsilon}\log d \right)$ per update, with success probability at least that given in Eq.~\eqref{eq:prob}. Alternatively, using the probability amplification of Algorithm~\ref{alg:amplify}, it achieves amortized time complexity $O\!\left(\frac{d}{\varepsilon}\log d \log^2 \frac{1}{\delta}\right)$ per update while guarantee success probability at least $1-\delta$, where $\delta$ is a tunable parameter that can be chosen arbitrarily in the interval $(0, 1/100)$.
\end{reptheorem}

\begin{proof}
    \textbf{Error bound analysis:}

    The proof of the error bound is roughly the same as that of Proposition~\ref{prop:error-bound}.

    Eq.~\eqref{eq:fd-svd-2k-update-1} in the proof of Proposition~\ref{prop:error-bound} becomes:
    \begin{equation*}
        \bm{B}_i=\sum_{j=0}^{i} \left(\boldsymbol{Z}_j\boldsymbol{Z}_j^{\top}{\boldsymbol{C}_j^\prime}^{\top}\boldsymbol{C}_j^\prime+{\boldsymbol{C}_j^\prime}^{\top}\boldsymbol{C}_j^\prime\boldsymbol{Z}_j\bm{Z}_j^{\top}-\boldsymbol{Z}_j\bm{Z}_j^{\top}{\boldsymbol{C}_j^\prime}^{\top}\boldsymbol{C}_j^\prime\boldsymbol{Z}_j\bm{Z}_j^{\top}\right).
    \end{equation*}

    We first define the single-update error as:
    \begin{equation}
        \label{eq:attp-2k-update-8}
            \bm{\Delta}_{i}  = \boldsymbol{a}_{i}^\top\boldsymbol{a}_{i}+\boldsymbol{C}_{i-1}^{\top}\boldsymbol{C}_{i-1}+\boldsymbol{B}_{i-1}-\boldsymbol{C}_{i}^{\top}\boldsymbol{C}_{i}-\boldsymbol{B}_{i}.
    \end{equation}
    Then we have
    \begin{equation}
        \label{eq:attp-2k-update-12}
        \begin{split}
            \left\|\sum_{i=1}^{T} \bm{\Delta}_i\right\|_2 & \overset{Eq.~\eqref{eq:attp-2k-update-8}}{=} \left\|\left(\sum_{i=1}^{T} \boldsymbol{a}_{i}^\top\boldsymbol{a}_{i}\right)-\bm{C}_{T}^\top \bm{C}_{T}-\bm{B}_{T}\right\|_2 \\
                                                          & = \left\|\bm{A}_{T}^\top \bm{A}_{T}-\bm{C}_T^\top \bm{C}_T-\bm{B}_T\right\|_2                                                                                                        \\
                                                          & \geq \left\|\bm{A}_{T}^\top \bm{A}_{T}-\bm{B}_T\right\|_2 - \left\|\bm{C}_{T}^\top \bm{C}_{T}\right\|_2 .                                                                             \\
        \end{split}
    \end{equation}
    Rearranging the terms yields
    \begin{equation*}
             \left\|\bm{A}_{T}^\top \bm{A}_{T}-\bm{B}_T\right\|_2 \overset{Eq.~\eqref{eq:attp-2k-update-12}}{\leq} \underbrace{\left\|\sum_{i=1}^{T} \bm{\Delta}_i\right\|_2}_{\text{Term 1}} + \underbrace{\left\|\bm{C}_{T}^\top \bm{C}_{T}\right\|_2}_{\text{Term 2}}. \\
    \end{equation*}
    Bounding Term 1 is similar to the proof of Proposition~\ref{prop:error-bound}, except that from Eq.~\eqref{eq:fd-svd-2k-update-13}:
    \begin{equation*}
        \begin{split}
            \tr\left(\bm{B}_{T}+\bm{C}_T^\top\bm{C}_T\right) & =\sum_{i=1}^{T} \left[\tr\left(\bm{B}_i+\bm{C}_i^\top \bm{C}_i\right) - \tr\left(\bm{B}_{i-1}+\bm{C}_{i-1}^\top \bm{C}_{i-1}\right) \right]  \\
                                          & = \sum_{i=1}^{T} \left[\tr(\bm{a}_i^\top \bm{a}_i)-\tr(\bm{\Delta}_i)\right] \leq \left\|\bm{A}_{T}\right\|_F^2 - \ell\sum_{i=1}^T (\sigma_\ell^2)_i.                                       \\
        \end{split}
    \end{equation*}

    Therefore, we derive:
    \begin{equation*}
        \sum_{i=1}^T (\sigma_\ell^2)_i \leq \frac{1}{\ell}\left(\left\|\bm{A}_{T}\right\|_F^2 - \tr\left(\bm{B}_T+\bm{C}_T^\top \bm{C}_T\right)\right)\le \frac{1}{\ell}\left\|\bm{A}_{T}\right\|_F^2.
    \end{equation*}

    Finally, Term 1 is bounded by
    \begin{equation*}
        \begin{split}
            \text{Term 1} \le \frac{1}{\ell}\left\|\bm{A}_T\right\|_F^2.
        \end{split}
    \end{equation*}

    Bounding Term 2 is similar to the proof of Proposition~\ref{prop:error-bound}. Considering line 6 of Algorithm~\ref{alg:attp} and Lemma~\ref{lem:power-iter} yields (w.h.p.):
    \begin{equation*}
        \text{Term 2} = \left\|\bm{C}_T^\top \bm{C}_T\right\| \leq 4\theta \le 4\varepsilon \left\|\bm{A}_T\right\|_F^2.
    \end{equation*}

    In conclusion, we have (w.h.p.)
    \begin{equation*}
        \text{Term 1} + \text{Term 2} \le O(1)\varepsilon \left\|\bm{A}_{T}\right\|_F^2.
    \end{equation*}

    After that, the remainder of the proof is similar to that of Theorem~\ref{prop:error-bound}.

    \textbf{Space Complexity Analysis:}
    The proof of the space complexity is roughly the same as the proof of Proposition~\ref{prop:time-complexity}.

    Eq.~\eqref{eq:fd-svd-2k-bound} in the proof of Proposition 3.4 becomes:
    \begin{equation*}
        \left\|\bm{C}_i^\prime\right\|_F^2-\left\|\bm{C}_i^\prime-\bm{C}_i^\prime\bm{Z}_i\bm{Z}_i^\top\right\|_F^2 \ge 2^{j-1} \theta_i                                                                                          .                                                                                                                         \\
    \end{equation*}

    Eq.~\eqref{eq:fd-svd-2k-update} becomes:
    \begin{equation*}
        \begin{split}
            \left\|\bm{C}^{\prime}_{i}\right\|_F^2 - \left\|\bm{C}_i\right\|_F^2 & \ge 2^{j-1}\theta_i,                                  \\
            \left\|\bm{C}^{\prime}_{i}\right\|_F^2                    & \le \left\|\bm{C}_{i-1}\right\|_F^2 + \bm{a}_{i}\bm{a}_{i}^\top,
        \end{split}
    \end{equation*}

    Eq.~\eqref{eq:fd-svd-2k-update-2} becomes:
    \begin{equation}
        \label{eq:attp-2k-update-2}
        \left\|\bm{C}_{i-1}\right\|_F^2 + \bm{a}_i\bm{a}_i^\top - \left\|\bm{C}_{i}\right\|_F^2 \ge 2^{j-1}\theta_i = \frac{k_i \theta_i}{2},
    \end{equation}

    Plugging $\theta_i = \varepsilon \|\bm{A}_i\|_F^2$ into Eq.~\eqref{eq:attp-2k-update-2} and summing from the beginning of the epoch to time $i$, we obtain
    \begin{equation}
        \label{eq:attp-2k-update-3}
        \left\|\bm{A}_i\right\|_F^2 - \left\|\bm{C}_i\right\|_F^2 \ge \sum_{j=1}^i \frac{k_j \varepsilon \left\|\bm{A}_j\right\|_F^2}{2}.
    \end{equation}

    Before proceeding, we need the following lemma.
    
    \begin{lemma}
        \label{lem:fd-svd-2k-update-4}
        Given Eq.~\eqref{eq:attp-2k-update-3}, we have the following inequality:
        \begin{equation}
            \label{eq:attp-2k-update-4}
            \left\|\bm{A}_i\right\|_F^2\ge \frac{\frac{k_1}{2}\varepsilon}{\prod_{j=2}^{i}\left(1-\frac{k_j}{2}\varepsilon\right)}\left\|\bm{A}_1\right\|_F^2,
        \end{equation}
        where $0<\frac{k_j}{2}\varepsilon<1$ for $j=1,\ldots,i$ and $k_1=1$.
    \end{lemma}

    The proof of Lemma~\ref{lem:fd-svd-2k-update-4} is deferred to Section~\ref{sec:proof-fd-svd-2k-update-4}. Taking the logarithm of both sides of Eq.~\eqref{eq:attp-2k-update-4} gives:
    \begin{equation*}
        \log \frac{2\left\|\bm{A}_i\right\|_F^2}{\varepsilon\left\|\bm{A}_1\right\|_F^2} \ge \sum_{j=1}^i -\log \left(1-\frac{k_j}{2}\varepsilon\right)\ge \sum_{j=1}^{i} \frac{k_j}{2}\varepsilon,
    \end{equation*}
    where the last inequality holds because $-\log (1-x) \ge x$ for $x\in (0,1)$. Then we have
    \begin{equation*}
        \sum_{j=1}^{T} k_j \le \frac{2}{\varepsilon} \log \frac{2\|\bm{A}_T\|_F^2}{\varepsilon},
    \end{equation*}
    as $\left\|\bm{A}_1\right\|_F^2\ge 1$.

    Then, as in the proof of Proposition~\ref{prop:sidsfd-space}, the total space complexity of Algorithm~\ref{alg:attp} is
    \begin{equation*}
        O\left(d \sum_{i=1}^{T} k_i\right) = O\left(\frac{d}{\varepsilon} \log \frac{\left\|\bm{A}_T\right\|_F^2}{\varepsilon}\right).
    \end{equation*}

    \textbf{Time Complexity Analysis:}

    The total time complexity of Simultaneous Iterations in Algorithm~\ref{alg:attp} is
    \begin{equation*}
        \sum_{i=1}^{T} O(d\ell k_i \log d) = O(d\ell \log d) \sum_{i=1}^{T} k_i = O\left(\frac{d\ell}{\varepsilon}  \log d \log \frac{\left\|\bm{A}_T\right\|_F^2}{\varepsilon}\right)= O\left(\frac{d\ell}{\varepsilon}  \log d \log \frac{T R}{\varepsilon}\right).
    \end{equation*}
    If $T>\ell \log \frac{TR}{\varepsilon}$, the amortized time complexity of Simultaneous Iterations in Algorithm~\ref{alg:attp} is $O\left(\frac{d}{\varepsilon} \log d \right)$. As the Power Iteration in Algorithm~\ref{alg:attp} consumes $O\left(\frac{d}{\varepsilon} \log d\right)$ time per update, the amortized time complexity of Algorithm~\ref{alg:attp} per update is $O\left(\frac{d}{\varepsilon} \log d \right)$.
\end{proof}

\subsection{Proof of Lemma~\ref{lem:fd-svd-2k-update-4}}
\label{sec:proof-fd-svd-2k-update-4}

\begin{replem}{lem:fd-svd-2k-update-4}
Given
\begin{equation*}
    \|\bm{A}_i\|_F^2 - \|\bm{C}_i\|_F^2 \ge \sum_{j=1}^i \frac{k_j \varepsilon \|\bm{A}_j\|_F^2}{2},
\end{equation*}
we have the following inequality:
\begin{equation*}
    \|\bm{A}_i\|_F^2\ge \frac{\frac{k_1}{2}\varepsilon}{\prod_{j=2}^{i}\left(1-\frac{k_j}{2}\varepsilon\right)}\|\bm{A}_1\|_F^2,
\end{equation*}
where $0<\frac{k_j}{2}\varepsilon<1$ for $j=1,\ldots,i$ and $k_1=1$.
\end{replem}

\begin{proof}
    We are given: for each \(i\),
    \begin{equation*}
        \sum_{j=1}^i \|\bm{A}_j\|_F^2 \ge \sum_{j=1}^{i} \left( \frac{k_j}{2} \varepsilon \sum_{t=1}^{j} \|\bm{A}_t\|_F^2 \right),
    \end{equation*}
    where \(\|\bm{A}_j\|_F^2 > 0\) and \(0 \le \frac{k_j}{2} \varepsilon < 1\) (i.e., \(0 \le k_j \varepsilon < 2\)). Define \(S_i = \sum_{j=1}^i \|\bm{A}_j\|_F^2\), then the inequality becomes:
    \begin{equation*}
        S_i \ge \sum_{j=1}^{i} \left( \frac{k_j}{2} \varepsilon \right) S_j.
    \end{equation*}
    We need to prove:
    \begin{equation*}
        S_i \ge \frac{\frac{k_1}{2} \varepsilon}{\prod_{j=2}^{i} \left(1 - \frac{k_j}{2} \varepsilon\right)} \|\bm{A}_1\|_F^2.
    \end{equation*}
    Since \(\|\bm{A}_1\|_F^2 = S_1\), this is equivalent to:
    \begin{equation*}
        S_i \ge \frac{\frac{k_1}{2} \varepsilon}{\prod_{j=2}^{i} \left(1 - \frac{k_j}{2} \varepsilon\right)} S_1.
    \end{equation*}

    Define \(\alpha_j = \frac{k_j}{2} \varepsilon\), so \(0 \le \alpha_j < 1\). Define \(B_m = \prod_{j=2}^{m} (1 - \alpha_j)\) for \(m \ge 2\), and \(B_1 = 1\) (empty product). We will prove by induction on \(m \ge 1\),
    \begin{equation}
        \label{eq:app-i-1}
        \sum_{j=1}^{m} \alpha_j S_j \ge \frac{\alpha_1}{B_m} S_1. 
    \end{equation}

    \subsubsection*{Base case (\(m = 1\)):}
    \begin{equation*}
        \sum_{j=1}^{1} \alpha_j S_j = \alpha_1 S_1 = \frac{\alpha_1}{B_1} S_1
    \end{equation*}
    holds.

    \subsubsection*{Inductive step:}
    Assume that for \(m = k-1\), inequality~\eqref{eq:app-i-1} holds, i.e.,
    \begin{equation*}
        \sum_{j=1}^{k-1} \alpha_j S_j \ge \frac{\alpha_1}{B_{k-1}} S_1.
    \end{equation*}
    Consider \(m = k\). From the given inequality for \(i = k\):
    \begin{equation*}
        S_k \ge \sum_{j=1}^{k} \alpha_j S_j,
    \end{equation*}
    so
    \begin{equation*}
        S_k (1 - \alpha_k) \ge \sum_{j=1}^{k-1} \alpha_j S_j \ge \frac{\alpha_1}{B_{k-1}} S_1.
    \end{equation*}
    Therefore,
    \begin{equation*}
        S_k \ge \frac{\alpha_1}{B_{k-1} (1 - \alpha_k)} S_1 = \frac{\alpha_1}{B_k} S_1,
    \end{equation*}
    where \(B_k = B_{k-1} (1 - \alpha_k)\). Hence,
    \begin{equation*}
        \alpha_k S_k \ge \alpha_k \cdot \frac{\alpha_1}{B_k} S_1.
    \end{equation*}
    Now,
    \begin{equation*}
        \sum_{j=1}^{k} \alpha_j S_j = \sum_{j=1}^{k-1} \alpha_j S_j + \alpha_k S_k \ge \frac{\alpha_1}{B_{k-1}} S_1 + \alpha_k \cdot \frac{\alpha_1}{B_k} S_1.
    \end{equation*}
    Since \(B_{k-1} = B_k / (1 - \alpha_k)\), we have
    \begin{equation*}
        \frac{\alpha_1}{B_{k-1}} S_1 = \frac{\alpha_1 (1 - \alpha_k)}{B_k} S_1.
    \end{equation*}
    Thus,
    \begin{equation*}
        \sum_{j=1}^{k} \alpha_j S_j \ge \frac{\alpha_1 (1 - \alpha_k)}{B_k} S_1 + \alpha_k \cdot \frac{\alpha_1}{B_k} S_1 = \frac{\alpha_1}{B_k} S_1 (1 - \alpha_k + \alpha_k) = \frac{\alpha_1}{B_k} S_1.
    \end{equation*}
    Therefore, inequality (1) holds for \(m = k\). By induction, (1) holds for all \(m \ge 1\).

    Now, for any \(i\), from the given inequality:
    \begin{equation*}
        S_i \ge \sum_{j=1}^{i} \alpha_j S_j \ge \frac{\alpha_1}{B_i} S_1 = \frac{\alpha_1}{\prod_{j=2}^{i} (1 - \alpha_j)} S_1.
    \end{equation*}
    Substituting \(\alpha_j = \frac{k_j}{2} \varepsilon\), we obtain:
    \begin{equation*}
        S_i \ge \frac{\frac{k_1}{2} \varepsilon}{\prod_{j=2}^{i} \left(1 - \frac{k_j}{2} \varepsilon\right)} S_1 = \frac{\frac{k_1}{2} \varepsilon}{\prod_{j=2}^{i} \left(1 - \frac{k_j}{2} \varepsilon\right)} \|\bm{A}_1\|_F^2.
    \end{equation*}
    Hence,
    \begin{equation*}
        \sum_{j=1}^i \|\bm{A}_j\|_F^2 \ge \frac{\frac{k_1}{2} \varepsilon}{\prod_{j=2}^{i} \left(1 - \frac{k_j}{2} \varepsilon\right)} \|\bm{A}_1\|_F^2,
    \end{equation*}
    as required.
\end{proof}

\section{Proof of Theorem~\ref{thm:sw-amm}}
\label{app:sw-amm}

\begin{reptheorem}{thm:sw-amm}
    Algorithm~\ref{alg:sw-amm} solves the Tracking Approximate Matrix Multiplication over Sliding Window (Problem~\ref{prob:sw-amm}) with space complexity $O\left(\frac{d_x+d_y}{\varepsilon}\log R\right)$ and amortized time complexity $O\left(\frac{d_x+d_y}{\varepsilon}\log d \log R \right)$ per update, with success probability at least that given in Eq.~\eqref{eq:prob}. Alternatively, using the similar like probability amplification of Algorithm~\ref{alg:amplify}, it achieves amortized time complexity $O\!\left(\frac{d_x+d_y}{\varepsilon}\log d \log^2 \frac{1}{\delta}\right)$ per update while guarantee success probability at least $1-\delta$, where $\delta$ is a tunable parameter that can be chosen arbitrarily in the interval $(0, 1/100)$.
\end{reptheorem}

\begin{proof}
    The proof of Theorem~\ref{thm:sw-amm} is similar to the proof in Appendix~\ref{sec:proof-sw}.

    \textbf{Error Analysis:}
    We denote the matrix $\bm{A}^\top\bm{B}$ at line 31 of Algorithm~\ref{alg:sw-amm} over the window $(T-N, i]$ as $\hat{\bm{A}}_i^\top \hat{\bm{B}}_i$. That is,
    \begin{equation}
        \label{eq:amm-2k-update-1}
        \hat{\bm{A}}_i \hat{\bm{B}}_i^\top=\bm{A}_i \bm{B}_i^\top +\sum_{j=T-N+1}^{i} \left(\boldsymbol{Z}_j\boldsymbol{Z}_j^{\top}{\boldsymbol{A}_j^\prime}{\boldsymbol{B}_j^\prime}^\top+{\boldsymbol{A}_j^\prime}{\boldsymbol{B}_j^\prime}^\top\boldsymbol{H}_j\bm{H}_j^{\top}-\boldsymbol{Z}_j\bm{Z}_j^{\top}{\boldsymbol{A}_j^\prime}{\boldsymbol{B}_j^\prime}^\top\boldsymbol{H}_j\bm{H}_j^{\top}\right).
    \end{equation}

    We first define the single-update error as:
    \begin{equation}
        \label{eq:amm-2k-update-8}
            \bm{\Delta}_{i}  = \boldsymbol{x}_{i}\boldsymbol{y}_{i}^\top+\hat{\boldsymbol{A}}_{i-1}\hat{\boldsymbol{B}}_{i-1}^\top-\hat{\boldsymbol{A}}_{i}\hat{\boldsymbol{B}}_{i}^\top.
    \end{equation}

    Then we have
    \begin{equation}
        \label{eq:amm-2k-update-12}
        \begin{split}
            \left\|\sum_{i=T-N+1}^{T} \bm{\Delta}_i\right\|_2 & \overset{Eq.~\eqref{eq:amm-2k-update-8}}{=} \left\|\left(\sum_{i=T-N+1}^{T} \boldsymbol{x}_{i}\boldsymbol{y}_{i}^\top\right)-\hat{\bm{A}}_{T} \hat{\bm{B}}_{T}^\top + \hat{\bm{A}}_{T-N} \hat{\bm{B}}_{T-N}^\top\right\|_2 \\
                                                              & = \left\|\bm{X}_{T-N,T} \bm{Y}_{T-N,T}^\top-\hat{\bm{A}}_{T} \hat{\bm{B}}_{T}^\top + \hat{\bm{A}}_{T-N} \hat{\bm{B}}_{T-N}^\top\right\|_2                                                                                         \\
                                                              & \geq \left\|\bm{X}_{T-N, T} \bm{Y}_{T-N,T}^\top-\hat{\bm{A}}_{T} \hat{\bm{B}}_{T}^\top\right\|_2  - \left\|\hat{\bm{A}}_{T-N} \hat{\bm{B}}_{T-N}^\top\right\|_2                                                                   \\
                                                              & \overset{Eq.~\eqref{eq:amm-2k-update-1}}{=} \left\|\bm{X}_{T-N, T} \bm{Y}_{T-N,T}^\top-\hat{\bm{A}}_{T} \hat{\bm{B}}_{T}^\top\right\|_2 - \left\|\bm{A}_{T-N} \bm{B}_{T-N}^\top\right\|_2 ,
        \end{split}
    \end{equation}
    which can be rewritten as
    \begin{equation*}
        \left\|\bm{X}_{T-N,T} \bm{Y}_{T-N,T}^\top-\hat{\bm{A}}_T \hat{\bm{B}}_T^\top\right\|_2\overset{Eq.~\eqref{eq:amm-2k-update-12}}{\leq} \underbrace{\left\|\sum_{i=T-N+1}^{T} \bm{\Delta}_i\right\|_2}_{\text{Term 1}} + \underbrace{\left\|\bm{A}_{T-N} \bm{B}_{T-N}^\top\right\|_2}_{\text{Term 2}}. \\
    \end{equation*}

    \paragraph{\textbf{Bounding Term 1}}

    Substituting Eq.~\eqref{eq:amm-2k-update-1} into Eq.~\eqref{eq:amm-2k-update-8}, we have:
    \begin{equation*}
        \begin{split}
            \bm{\Delta}_{i} {{}={}} & \boldsymbol{x}_{i}\boldsymbol{y}_{i}^\top+\boldsymbol{A}_{i-1}\boldsymbol{B}_{i-1}^{\top}-{\boldsymbol{A}_i}\boldsymbol{B}_{i}^{\top}                                                                                                                                                                                   \\
                             & - ( \boldsymbol{Z}_i \bm{Z}_i^{\top}{\boldsymbol{A}^{\prime}_i}^{\top}{\boldsymbol{B}^{\prime}_i}+{\boldsymbol{A}^{\prime}_i}{\boldsymbol{B}^{\prime}_i}^\top\boldsymbol{H}_i \bm{H}_i^{\top}-\boldsymbol{Z}_i \bm{Z}_i^{\top}{\boldsymbol{A}^\prime_i}{\boldsymbol{B}^\prime_i}^\top\boldsymbol{H}_i \bm{H}_i^{\top}). \\
        \end{split}
    \end{equation*}

    We take the relationship between the two updates: $\bm{A}_i \bm{B}_i^\top= (\bm{I}-\bm{Z}_i \bm{Z}_i^\top ) {\bm{A}_i^\prime} {\bm{B}_i^\prime}^\top (\bm{I}- \bm{H}_i \bm{H}_i^\top )$ (from line 23 of Algorithm~\ref{alg:sw-amm}; if line 25 is triggered, then $\bm{Z}_i$ and $\bm{H}_i$ are considered to be zero), and $\tilde{\bm{A}}_i \tilde{\bm{B}}_i^\top = \bm{A}_{i-1} \bm{B}_{i-1} ^\top+ \bm{x}_i \bm{y}_i^\top$ (line 7), and substitute these into the previous expression to derive the following property:
    \begin{equation*}
        \begin{split}
            \bm{\Delta}_{i}  {{}={}} & \tilde{\bm{A}_{i}} \tilde{\bm{B}_{i}} ^\top- (\bm{I}-\bm{Z}_i \bm{Z}_i^\top ) {\bm{A}_i^\prime} {\bm{B}_i^\prime}^\top (\bm{I}- \bm{H}_i \bm{H}_i^\top )                                                                                                                                                                \\
                               & - (  \boldsymbol{Z}_i \bm{Z}_i^{\top}{\boldsymbol{A}^{\prime}_i}^{\top}{\boldsymbol{B}^{\prime}_i}+{\boldsymbol{A}^{\prime}_i}{\boldsymbol{B}^{\prime}_i}^\top\boldsymbol{H}_i \bm{H}_i^{\top}-\boldsymbol{Z}_i \bm{Z}_i^{\top}{\boldsymbol{A}^\prime_i}{\boldsymbol{B}^\prime_i}^\top\boldsymbol{H}_i \bm{H}_i^{\top}) \\
            {{}={}}                  & \tilde{\bm{A}_{i}} \tilde{\bm{B}_{i}}^\top - {\bm{A}^\prime_i} {\bm{B}^\prime_i}^\top.                                                                                                                                                                                                                                  \\
        \end{split}
    \end{equation*}

    Consider lines~8-11 of Algorithm~\ref{alg:sw-amm}; from these steps, we have:
    \begin{equation}
        \label{eq:attp-2k-update-10}
        \bm{\Delta}_{i}=
        \begin{cases}
            \bm{U}_i\min\left(\bm{\Sigma}_i^2, \bm{I}(\sigma_\ell^2)_i\right)\bm{V}_i^\top & \text{rows}(\bm{A}_i^\prime) \geq 2\ell, \\
            \bm{0}                                                                         & \text{else},                     \\
        \end{cases}
    \end{equation}
    where $\bm{\Sigma}_i^2$ denotes the squared singular value matrix computed at line 3 of the $i$-th update, and $(\sigma_\ell^2)_i$ represents the $\ell$-th largest squared singular value in $\bm{\Sigma}_i^2$. (If $\text{rows}(\bm{A}_i^\prime) < 2\ell$, we set $(\sigma_\ell^2)_i = 0$.)

    Next, summing the matrices $\bm{\Delta}_i$ from Eq.~\eqref{eq:attp-2k-update-10} over the interval $(T - N,T]$ and analyzing their $\|\cdot\|_2$ norm, we obtain:
    \begin{equation}
        \label{eq:attp-2k-update-11}
        \left\|\sum_{i=T-N+1}^{T} \bm{\Delta}_i\right\|_2  \leq \sum_{i=T-N+1}^{T} \left\|\bm{\Delta}_i\right\|_2  = \sum_{i=T-N+1}^T (\sigma_\ell^2)_i.     \\
    \end{equation}

    We now aim to bound $\sum_{i=T-N+1}^T (\sigma_\ell^2)_i$. Let $\|\cdot\|_*$ be the \textit{1-Schatten norm}. If the singular value decomposition of $\bm{A}$ is $\bm{A}=\bm{U\Sigma V}^\top$, then $\|\bm{A}\|_* = \tr(\bm{\Sigma}) = \sum_{i=1}^d \sigma_i(\bm{A})$.
    \begin{equation}
        \label{eq:amm-2k-update-13}
        \begin{split}
                &\sum_{i=T-N+1}^{T}\left(\left\|\bm{A}_i^\prime {\bm{B}_i^\prime}^\top\right\|_* - \left\|\bm{A}_{i-1}\bm{B}_{i-1}^\top\right\|_*\right) \\=   & \sum_{i=T-N+1}^{T}\left(\left\|\tilde{\bm{A}}_i \tilde{\bm{B}}_i^\top\right\|_* - \left\|\bm{A}_{i-1}\bm{B}_{i-1}^\top\right\|_*\right) -\sum_{i=T-N+1}^{T}\left(\left\|\tilde{\bm{A}}_i \tilde{\bm{B}}_i^\top\right\|_* -\left\|\bm{A}_i^\prime {\bm{B}_i^\prime}^\top\right\|_*\right) \\
            \le & \sum_{i=T-N+1}^{T} \|\bm{x}_i\|_2\|\bm{y}_i\|_2 - \ell \sum_{i=T-N+1}^{T} (\sigma_\ell^2)_i.
        \end{split}
    \end{equation}
    Rearranging Eq.~\eqref{eq:amm-2k-update-13} yields
    \begin{equation}
        \label{eq:amm-2k-update-21}
        \begin{split}
            &\ell \sum_{i=T-N+1}^{T} (\sigma_\ell^2)_i \\ \le & \left(\sum_{i=T-N+1}^{T} \|\bm{x}_i\|_2\|\bm{y}_i\|_2\right) - \sum_{i=T-N+1}^{T} \left(\|\bm{A}_i^\prime {\bm{B}_i^\prime}^\top\|_* - \|\bm{A}_{i-1}\bm{B}_{i-1}^\top\|_*\right)                                         \\
            =                                              & \left(\sum_{i=T-N+1}^{T} \|\bm{x}_i\|_2\|\bm{y}_i\|_2\right)-\|\bm{A}_T^\prime {\bm{B}_T^\prime}^\top\|_* + \|\bm{A}_{T-N} \bm{B}_{T-N}^\top \|_* + \sum_{i=T-N+1}^{T-1} \left(\|\bm{A}_{i}\bm{B}_{i}^\top\|_*-\|\bm{A}_{i}^\prime {\bm{B}_{i}^\prime}^\top\|_*\right)   \\
            \le                                            & \left(\sum_{i=T-N+1}^{T} \|\bm{x}_i\|_2\|\bm{y}_i\|_2\right)+ \|\bm{A}_{T-N} \bm{B}_{T-N}^\top \|_*  + \sum_{i=T-N+1}^{T-1} \left(\|\bm{A}_{i}\bm{B}_{i}^\top\|_*-\|\bm{A}_{i}^\prime {\bm{B}_{i}^\prime}^\top\|_*\right).
        \end{split}
    \end{equation}

    From line 23 of Algorithm~\ref{alg:sw-amm}, we have \begin{equation}\bm{A}_i \bm{B}_i^\top = \left(\bm{I}-\bm{ZZ}^\top\right)\bm{A}_i^\prime {\bm{B}_i^\prime}^\top \left(\bm{I}-\bm{HH}^\top\right).\end{equation} As $\bm{I}-\bm{ZZ}^\top$ and $\bm{I}-\bm{HH}^\top$ are orthogonal projectors, we have
    \begin{equation*}
        \|\bm{A}_i\bm{B}_i^\top\|_* - \|\bm{A}_i^\prime {\bm{B}_i^\prime}^\top\|_* \le 0.
    \end{equation*}

    Then we continue from Eq.~\eqref{eq:amm-2k-update-21} and obtain:
    \begin{equation*}
        \begin{split}
             \ell \sum_{i=T-N+1}^{T} (\sigma_\ell^2)_i  \le & \left(\sum_{i=T-N+1}^{T} \left\|\bm{x}_i\right\|_2\left\|\bm{y}_i\right\|_2\right) + \left\|\bm{A}_{T-N} \bm{B}_{T-N}^\top \right\|_* + \sum_{i=T-N+1}^{T-1} \left(\left\|\bm{A}_{i}\bm{B}_{i}^\top\right\|_*-\left\|\bm{A}_{i}^\prime {\bm{B}_{i}^\prime}^\top\right\|_*\right) \\
            \le                                            & \left(\sum_{i=T-N+1}^{T} \left\|\bm{x}_i\right\|_2\left\|\bm{y}_i\right\|_2\right) + \left\|\bm{A}_{T-N}\bm{B}_{T-N}^\top\right\|_*                                                                                                                        \\
            \le                                            & \left(\sqrt{\sum_{i=T-N+1}^{T}\left\|\bm{x}_i\right\|_2^2}\sqrt{\sum_{i=T-N+1}^{T}\|\bm{y}_i\|_2^2}\right) + \left\|\bm{A}_{T-N}\bm{B}_{T-N}^\top\right\|_*                                                                        \\
            =                                              & \left\|\bm{X}_\text{W}\right\|_F\left\|\bm{Y}_\text{W}\right\|_F + \left\|\bm{A}_{T-N}\bm{B}_{T-N}^\top\right\|_*.
        \end{split}
    \end{equation*}

    Therefore, we get the following bound:
    \begin{equation}
        \label{eq:amm-2k-update-14}
        \text{Term 1} \le \sum_{i=T-N+1}^T (\sigma_\ell^2)_i \leq \frac{1}{\ell}\left(\left\|\bm{X}_\text{W}\right\|_F\left\|\bm{Y}_\text{W}\right\|_F + \left\|\bm{A}_{T-N}\bm{B}_{T-N}^\top\right\|_*\right).
    \end{equation}

    \paragraph{\textbf{Bounding Term 2.}}
    We state the following lemma:
    \begin{lemma}
        \label{lem:asym-power-iter}
        We have
        \begin{equation*}
            \begin{split}
                \left\|\bm{A}_{T-N} \bm{B}_{T-N}^\top\right\|_2 \leq 4\theta \leq 4 \varepsilon \left\|\bm{X}_{\text{W}}\right\|_F\left\|\bm{Y}_\text{W}\right\|_F,
            \end{split}
        \end{equation*}
        with success probability at least \(\frac{99}{100}\left(1- \frac{2}{\pi\sqrt{e}}\cdot\frac{1}{\sqrt{d \log_2 d}}\right)\). 
    \end{lemma}

    Combining Eq.~\eqref{eq:amm-2k-update-14} with Lemma~\ref{lem:asym-power-iter}, we obtain:
    \begin{equation}
        \label{eq:attp-2k-update-24}
        \begin{split}
         \text{Term 1} + \text{Term 2}   =                 & \left\|\sum_{i=T-N+1}^{T} \bm{\Delta}_i\right\|_2 + \left\|\bm{A}_{T-N} \bm{B}_{T-N}^\top\right\|_2                                                           \\
        \overset{Eq.~\eqref{eq:amm-2k-update-14}}{\leq} & \frac{1}{\ell}\left\|\bm{X}_\text{W}\right\|_F\left\|\bm{Y}_\text{W}\right\|_F + \frac{1}{\ell}\left\|\bm{A}_{T-N}\bm{B}_{T-N}^\top\right\|_* + \left\|\bm{A}_{T-N}\bm{B}_{T-N}^\top\right\|_2  \\
        \leq                                            & \frac{1}{\ell}\left\|\bm{X}_\text{W}\right\|_F\left\|\bm{Y}_\text{W}\right\|_F + 3\left\|\bm{A}_{T-N}\bm{B}_{T-N}^\top\right\|_2                                                               \\
        \overset{Lemma~\ref{lem:power-iter}}{\leq}      & \frac{1}{\ell}\left\|\bm{X}_\text{W}\right\|_F\left\|\bm{Y}_\text{W}\right\|_F + 12\varepsilon \left\|\bm{X}_\text{W}\right\|_F\|\bm{Y}_\text{W}\|_F .                                                                     \\
        \end{split}
    \end{equation}
    where the second-to-last inequality holds because $\left\|\bm{A}_{T-N}\bm{B}_{T-N}^\top\right\|_* = \sum_{i=1}^{2\ell} \sigma_i^2 \le 2\ell \sigma_1^2 = 2\ell \left\|\bm{A}_{T-N}\bm{B}_{T-N}^\top\right\|_2$.

    In the initialization of Algorithm~\ref{alg:sw-amm}, the parameter is set as $\ell = \lceil \frac{2}{\varepsilon} \rceil$. Substituting this into Eq.~\eqref{eq:attp-2k-update-24}, we obtain:
    \begin{equation}
        \label{eq:amm-a}
        \text{Term 1} + \text{Term 2} \le  \frac{25}{2}\varepsilon \left\|\bm{X}_{T-N,T}\right\|_F\left\|\bm{Y}_{T-N,T}\right\|_F^2  =  O(1)\varepsilon \left\|\bm{X}_{T-N,T}\right\|_F\left\|\bm{Y}_{T-N,T}\right\|_F.
    \end{equation}

    Finally, lines~31-33 guarantee that
    \begin{equation}
        \label{eq:amm-c}
        \left\|\hat{\bm{A}}_T\hat{\bm{B}}_T^\top - \bm{A}^*\left(\bm{B}^*\right)^\top \right\|_2= \left\|\bm{U}\min\left(\bm{\Sigma}, \sigma_\ell \bm{I}\right)\bm{V}^\top\right\|_2 = \sigma_\ell\left(\hat{\bm{A}}_T\hat{\bm{B}}_T^\top\right).
    \end{equation}
    By Weyl's inequalities, we have
    \begin{equation}
        \label{eq:amm-b}
        \begin{split}
        \sigma_\ell\left(\hat{\bm{A}}_T\hat{\bm{B}}_T^\top\right)=&\sigma_\ell\left(\hat{\bm{A}}_T\hat{\bm{B}}_T^\top - \bm{X}_{T-N,T} \bm{Y}_{T-N,T}^\top+\bm{X}_{T-N,T} \bm{Y}_{T-N,T}^\top\right)\\ 
        \le & \sigma_1\left(\hat{\bm{A}}_T\hat{\bm{B}}_T^\top - \bm{X}_{T-N,T} \bm{Y}_{T-N,T}^\top\right) + \sigma_\ell\left(\bm{X}_{T-N,T} \bm{Y}_{T-N,T}^\top\right)\\
        \overset{Eq.~\eqref{eq:amm-a}}{\le}& O(1)\varepsilon \left\|\bm{X}_{T-N,T}\right\|_F\left\|\bm{Y}_{T-N,T}\right\|_F+\frac{1}{\ell}\left\|\bm{X}_{T-N,T} \bm{Y}_{T-N,T}^\top\right\|_*\\
        =&O(1)\varepsilon \left\|\bm{X}_{T-N,T}\right\|_F\left\|\bm{Y}_{T-N,T}\right\|_F.
        \end{split}
    \end{equation}

        Finally, combining Eq.~\eqref{eq:amm-a},~\eqref{eq:amm-c} and~\eqref{eq:amm-b}, we conclude the proof:
    \begin{multline*}
        \left\|\bm{X}_{T-N,T} \bm{Y}_{T-N,T}^\top-\bm{A}^*\left(\bm{B}^*\right)^\top \right\|_2\le \left\|\bm{X}_{T-N,T} \bm{Y}_{T-N,T}^\top-\hat{\bm{A}}_T \hat{\bm{B}}_T^\top\right\|_2+\left\|\hat{\bm{A}}_T\hat{\bm{B}}_T^\top - \bm{A}^*\left(\bm{B}^*\right)^\top \right\|_2\\ \le O(1)\varepsilon \left\|\bm{X}_{T-N,T}\right\|_F\left\|\bm{Y}_{T-N,T}\right\|_F.
    \end{multline*}

    \textbf{Space Complexity analysis:}

    The simultaneous iteration algorithm guarantees that:
    \begin{equation}
        \label{eq:amm-2k-update-100}
        \begin{split}
            \left\|\bm{A}^{\prime}_{i}{\bm{B}^\prime_i}^\top\right\|_* - \left\|\bm{A}_i\bm{B}_i^\top\right\|_* & \ge 2^{j-1}\theta, \\
        \end{split}
    \end{equation}

    From Eq.~\eqref{eq:amm-2k-update-13}, we obtain:
    \begin{equation}
        \label{eq:amm-2k-update-16}
        \sum_{i=T-N+1}^{T}\left(\left\|\bm{A}_i^\prime {\bm{B}_i^\prime}^\top\right\|_* - \left\|\bm{A}_{i-1}\bm{B}_{i-1}^\top\right\|_*\right)\le \sum_{i=T-N+1}^{T} \left\|\bm{x}\right\|_2\left\|\bm{y}\right\|_2.
    \end{equation}

    Combining Eq.~\eqref{eq:amm-2k-update-100} and Eq.~\eqref{eq:amm-2k-update-16}, we obtain:
    \begin{equation*}
        \sum_{i=T-N+1}^{T} \left\|\bm{x}\right\|_2\left\|\bm{y}\right\|_2 - \left\|\bm{A}_T\bm{B}_T\right\|_* +\left\|\bm{A}_{T-N}\bm{B}_{T-N}^\top\right\|_* \ge \frac{\theta}{2} \sum_{i=T-N+1}^{T}k_i,
    \end{equation*}
    which can be reformulated as:
    \begin{equation*}
        \begin{split}
            \sum_{i=T-N+1}^{T} k_i & \leq \frac{2}{\theta} \left(\left\|\bm{X}_{T-N,T}\right\|_F\left\|\bm{Y}_{T-N,T}^\top\right\|_F + \left\|\bm{A}_{T-N}\bm{B}_{T-N}^\top\right\|_* \right) \\
            & \leq                        \frac{2}{\theta}\left(\left\|\bm{X}_{T-N,T}\right\|_F\left\|\bm{Y}_{T-N,T}\right\|_F+8\ell\theta\right),
        \end{split}
    \end{equation*}
    where the last equality holds with high probability, according to Lemma~\ref{lem:power-iter}, and that $\left\|\bm{A}_{T-N}\bm{B}_{T-N}^\top\right\|_* \le 2\ell \left\|\bm{A}_{T-N}\bm{B}_{T-N}^\top\right\|_2 \overset{w.h.p.}{\leq} 8\ell\theta$.

    By Fact~\ref{prop:ml}, plugging $\theta \le \varepsilon \left\|\bm{X}_{T-N,T}\right\|_F \left\|\bm{Y}_{T-N,T}\right\|_F < 2\theta$ into Eq.~\eqref{eq:attp-2k-update-4}, we obtain:
    \begin{equation*}
        \begin{split}
            \sum_{i=T-N+1}^{T} k_i \leq \frac{2}{\theta} \left(\frac{2\theta}{\varepsilon} + 8\ell\theta\right)=\frac{4}{\varepsilon} +16\ell = O\left(\frac{1}{\varepsilon}\right).
        \end{split}
    \end{equation*}

    The space cost of the algorithm is dominated by the sketch matrices $\bm{A}$, $\bm{B}$, and the snapshots stored in the queue $\mathcal{S}$. The sketch matrices $\bm{A}$ and $\bm{B}$ require $O((d_x+d_y)\ell)$ space. The queue $\mathcal{S}$ stores matrices $\bm{Z}_i \in \mathbb{R}^{d_x \times k_i}$, $\bm{Z}_i^\top {\bm{A}_i^\prime} {\bm{B}_i^\prime}^\top \in \mathbb{R}^{k_i \times d_y}$, $\bm{A}_i^\prime {\bm{B}_i^\prime}^\top \bm{H} \in \mathbb{R}^{d_x\times k_i}$ and $\bm{H}\in \mathbb{R}^{d_y\times k_i}$, so the space cost of $\mathcal{S}$ is
    \begin{equation*}
        O\left((d_x+d_y) \sum_{i=T-N+1}^{T} k_i\right) = O\left(\frac{d_x+d_y}{\varepsilon}\right).
    \end{equation*}

    There are $\lceil\log R\rceil$ levels, so the total space cost is
    \begin{equation*}
        O\left(\frac{d_x+d_y}{\varepsilon} \log R \right).
    \end{equation*}

    \textbf{Time Complexity analysis:}

    Similar to the proof of Proposition~\ref{prop:time-complexity}, the amortized time complexity is
    \begin{equation*}
        O\left(\frac{d_x+d_y}{\varepsilon} \log R \log d\right) . \qedhere
    \end{equation*}
\end{proof}

\section{Proof of Theorem~\ref{thm:dist}}
\label{app:dist}

\begin{reptheorem}{thm:dist}
    \pfive (Algorithm~\ref{alg:dsfd}) solves the Distributed Matrix Sketch problem  (Problem~\ref{prob:dsfd}) with communication complexity $O\left(\frac{md}{\varepsilon}\log \left\|\bm{A}\right\|_F^2\right)$ and amortized time complexity $O\left(\frac{d}{\varepsilon}\log d \right)$ per update, with success probability at least that given in Eq.~\eqref{eq:prob}. Alternatively, using the probability amplification of Algorithm~\ref{alg:amplify}, it achieves amortized time complexity $O\!\left(\frac{d}{\varepsilon}\log d \log^2 \frac{1}{\delta}\right)$ per update while guarantee success probability at least $1-\delta$, where $\delta$ is a tunable parameter that can be chosen in the interval $(0, 1/100)$.
\end{reptheorem}

\begin{proof}
    The error bound guarantee is ensured by the decomposability of the $m$ Frequent Directions sketches~\cite{liberty2013simple}, meaning that:
    \begin{lemma}[Decomposability]
        \label{lem:Decomposability}
        Given a matrix $\bm{A}\in\mathbb{R}^{d\times n}$, we decompose $\bm{A}$ into $m$ submatrices, i.e., $\bm{A}^{\top} = [\bm{A}_1,\bm{A}_2,\dots,\bm{A}_m]$, where $\bm{A}_i\in\mathbb{R}^{d\times n_i}$ and $\sum_{i=1}^m n_i = n$.
        If we construct a matrix sketch with an $\epsilon_i$ covariance error guarantee for each submatrix $\bm{A}_i$, denoted as $\bm{B}_i$, such that
        \begin{equation*}
            \left\Vert\bm{A}_i\bm{A}_i^\top-\bm{B}_i\bm{B}_i^{\top}\right\Vert_2\le \epsilon_i\cdot\left\Vert\bm{A}_i\right\Vert_F^2.
        \end{equation*}
        Then $\bm{B} = [\bm{B}_1,\bm{B}_2,\cdots,\bm{B}_m]$ is an approximation of $\bm{A}$ and the error bound is
        \begin{equation*}
            \left\Vert\bm{A}\bm{A}^{\top}-\bm{B}\bm{B}^{\top}\right\Vert_2\le\sum_{i=1}^m\epsilon_i\cdot\left\Vert\bm{A}_i\right\Vert_F^2.
        \end{equation*}
    \end{lemma}

    We let the vectors observed at site $S_i$ form the submatrix $\bm{A}_i$. In lines 19-20, the portion of the matrix $\bm{B}^\top\bm{B}$ maintained by the coordinator that comes from site $S_i$ is denoted as $\bm{B}_i^\top \bm{B}_i$. In the procedure of each site (lines 1-13 of Algorithm~\ref{alg:dsfd}), apart from replacing ``storing the dumped vector in the queue'' with ``sending it to the coordinator'', the other operations are equivalent to those in Algorithm~\ref{alg:attp}. Therefore, by Theorem~\ref{thm:attp}, we obtain:
    \begin{equation*}
        \left\|\bm{A}_i ^\top \bm{A}_i - \bm{B}_i^\top \bm{B}_i\right\|_2 \le \frac{\varepsilon}{m} \left\|\bm{A}\right\|_F^2.
    \end{equation*}
    
    From Lemma~\ref{lem:Decomposability}, we can obtain:
    \begin{equation*}
        \left\|\bm{A}^\top \bm{A} - \bm{B}^\top \bm{B}\right\|_2 \le \varepsilon \left\|\bm{A}\right\|_F^2.
    \end{equation*}

    We then bound the communication cost. We define the periods between two consecutive broadcasts of $\hat{F}$ (line 18 of Algorithm~\ref{alg:dsfd}) as a round. The $\hat{F}$ of the beginning of some round is $(1+\varepsilon)$ times the $\hat{F}$ of the beginning of the previous round. As the current $\hat{F}=\left\|\bm{A}\right\|_F^2$, we have $(1+\varepsilon)^r\le \left\|\bm{A}\right\|_F^2$. Therefore, at most $r=O\left(\log_{(1+\varepsilon)}\left\|\bm{A}\right\|_F^2\right)=O\left(\frac{1}{\varepsilon}\log \left\|\bm{A}\right\|_F^2\right)$ rounds are possible. Then, according to Lemma 3 of~\cite{ghashami2014continuous}:

    \begin{lemma}[Lemma 3 of~\cite{ghashami2014continuous}]
        After $r$ rounds, at most $O(m\cdot r)$ element update messages have been sent.
    \end{lemma}
    Instead of changing the weight using a single element, we use the Rayleigh quotient change along each corresponding estimated singular vector in $\bm{Z}$. In this setting, the update message for each element is represented by two $d$-dimensional vectors: one taken from a row of $\bm{Z}^\top \bm{C}^\top \bm{C}$ and the other from a column of $\bm{Z}$. Therefore, the total communication cost is $O\left(\frac{md}{\varepsilon}\log \left\|\bm{A}\right\|_F^2\right)$.

    For each site, the procedure is similar to that of Algorithm~\ref{alg:attp}. The amortized time complexity of Algorithm~\ref{alg:dsfd} is $O\left(\frac{d}{\varepsilon}\log d \right)$, which is the same as that of Algorithm~\ref{alg:attp}.
\end{proof}

\section{Proof of Theorem~\ref{thm:dist-sw}}
\label{app:dist-sw}

\begin{reptheorem}{thm:dist-sw}
    Algorithm~\ref{alg:dswfd} solves the Tracking Distributed Matrix Sketch over Sliding Window (Problem~\ref{prob:dswfd}) with communication complexity $O\left(\frac{md}{\varepsilon}\log \left\|\bm{A}\right\|_F^2\right)$ and amortized time complexity $O\left(\frac{d}{\varepsilon}\log d \right)$ per update, with success probability at least that given in Eq.~\eqref{eq:prob}. Alternatively, using the probability amplification of Algorithm~\ref{alg:amplify}, it achieves amortized time complexity $O\!\left(\frac{d}{\varepsilon}\log d \log^2 \frac{1}{\delta}\right)$ per update to guarantee success probability at least $1-\delta$, where $\delta$ is a tunable parameter that can be chosen arbitrarily in the interval $(0, 1/100)$.
\end{reptheorem}

\begin{proof}
    The error bound and communication cost guarantees follow from~\cite{zhang2017tracking}, similarly to those for DA2. As for the time complexity, each site in Algorithm~\ref{alg:dswfd} uses \pfive as a module, whose amortized time complexity is $O\left(\frac{d}{\varepsilon}\log d \right)$.
\end{proof}

\section{Optimize aDS-COD by our \sidscod Framework}
\label{sec:ads-cod}

Besides hDS-COD, Section~4.3 of~\cite{yao2025optimal} also proposes Adaptive DS-COD (aDS-COD) in \textit{Algorithm~6}, which optimizes both time and space complexity based on hDS-COD. However, it is explicitly pointed out that this algorithm is only suitable for datasets that ``have a large $R$ but little fluctuation within the window'', and not for the more strictly defined version in Problem~\ref{prob:sw-amm}, which requires handling arbitrary fluctuations within the window. Thus, for fairness in comparison, we do not include aDS-COD in Table~\ref{tab:comparison} as a competitor to our \sidsfd framework. Nonetheless, we find that our \sidscod framework (Algorithm~\ref{alg:sw-amm}) can replace the DS-COD component in aDS-COD, further optimizing its time complexity under conditions where \(\varepsilon\) is small. 

The high-level idea and algorithm skeleton share similarities with aDS-COD  (Algorithm~6 in~\cite{yao2025optimal}), although we replace the DS-COD module with \sidscod (Algorithm~\ref{alg:sw-amm}) to avoid frequent and time-consuming matrix factorizations. We refer to this algorithm as Adaptive AeroSketchCOD (Adaptive-AS in Algorithm~\ref{alg:adscod}).

\begin{algorithm}[h]
\SetAlgoLined
\caption{Adaptive \sidscod (Adaptive-AS)}
\label{alg:adscod}
    \SetKwProg{Fn}{Function}{:}{}

\KwIn{$\bm{X} \in \mathbb{R}^{m_x \times n},\ \bm{Y} \in \mathbb{R}^{m_y \times n}$, window size $N$, relative covariance error bound $\varepsilon$}

\textbf{Initialize:} Initialize a \sidscod (Algorithm~\ref{alg:sw-amm}) sketch $AS$ and an auxiliary sketch $AS_{\text{aux}}$ with initial threshold $\varepsilon N$ and threshold level $L=1$.\;

\Fn{\textsc{Update}($\bm{x}_i, \bm{y}_i$)}{
    \While{$AS.S[0].t + N \leq i$}{
        $AS.S.\operatorname{Dequeue}()$;
    }
    
    \If{$i\mod N == 1$}{
        $AS \leftarrow AS_{\text{aux}}$; \\
        $AS_{\text{aux}}.\operatorname{Initialize}()$;
    }
    
    $AS.\operatorname{Update}(\bm{x}_i, \bm{y}_i)$; \\
    $AS_{\text{aux}}.\operatorname{Update}(\bm{x}_i, \bm{y}_i)$;
    
    \If{$\operatorname{len}(AS.S) \geq \frac{L}{\varepsilon}$}{
        $L \leftarrow L + 1$; \\
        $AS.\theta \leftarrow 2 \cdot AS.\theta$;
    }
    \ElseIf{$\operatorname{len}(AS.S) \leq \frac{L-1}{\varepsilon}$}{
        $L \leftarrow L - 1$; \\
        $AS.\theta \leftarrow AS.\theta / 2$;
    }
}
    \Fn{\textsc{Query}()}{
        \Return $AS.\operatorname{Query}()$;
    }
\end{algorithm}

\begin{figure}[htbp]
    \centering
    \begin{subfigure}[t]{0.23\textwidth}
        \captionsetup{justification=centering}
        \includegraphics[width=\textwidth]{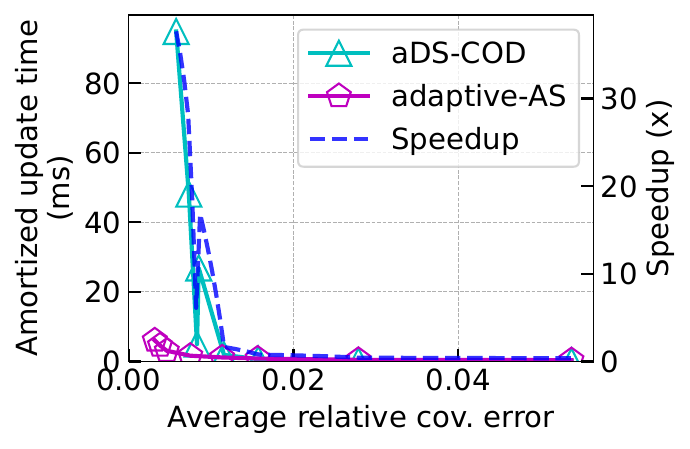}
        \caption{Amortized update time vs. average error}\label{fig:appendix-avg-err}
    \end{subfigure}
    \begin{subfigure}[t]{0.23\textwidth}
        \captionsetup{justification=centering}
        \includegraphics[width=\textwidth]{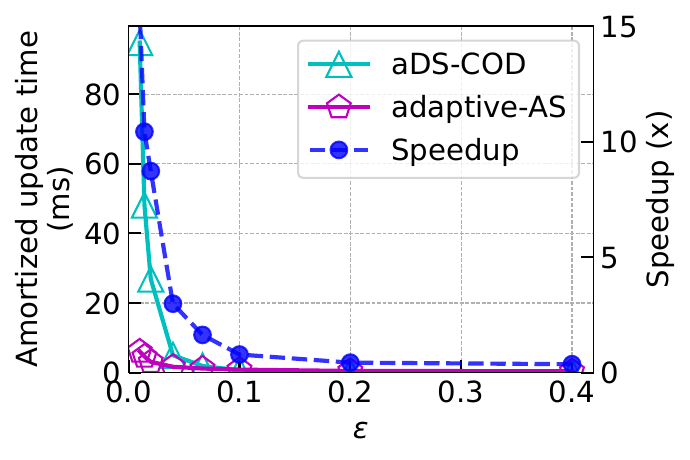}
        \caption{Amortized update time vs. $\varepsilon$}\label{fig:appendix-eps-time}
    \end{subfigure}
    \begin{subfigure}[t]{0.23\textwidth}
        \captionsetup{justification=centering}
        \includegraphics[width=\textwidth]{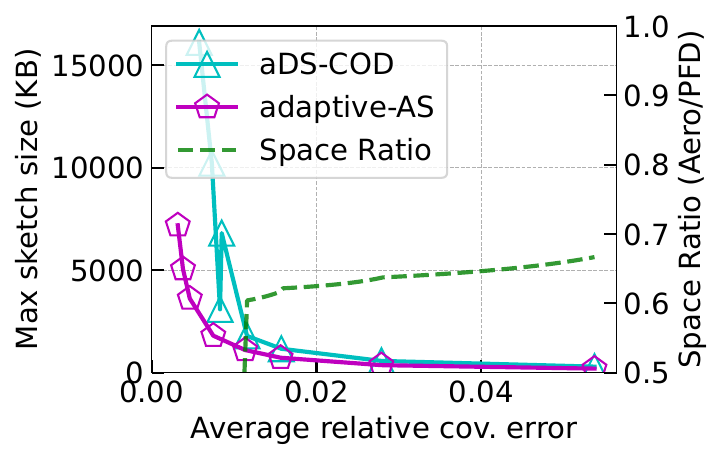}
        \caption{Memory cost vs. average error}\label{fig:appendix-size-err}
    \end{subfigure}
    \begin{subfigure}[t]{0.23\textwidth}
        \captionsetup{justification=centering}
        \includegraphics[width=\textwidth]{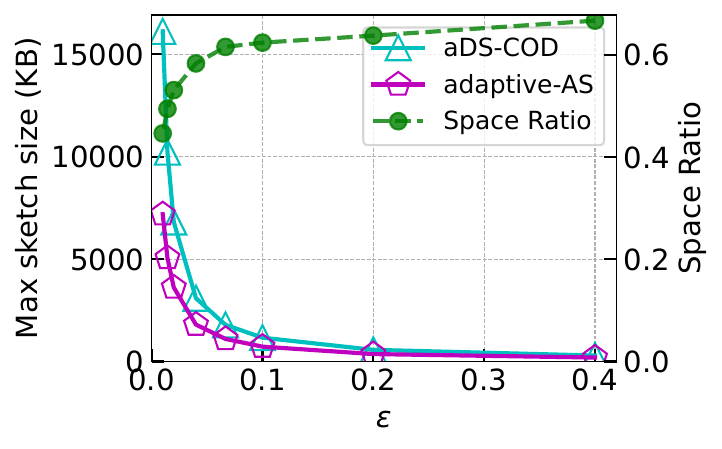}
        \caption{Memory cost vs. $\varepsilon$}\label{fig:appendix-size-eps}
    \end{subfigure}

    \caption{Adaptive \sidscod (Adaptive-AS) vs. aDS-COD}\label{fig:appendix-results}
\end{figure}

We conducted experiments on aDS-COD and Adaptive-AS using the \textbf{Uniform Random} dataset with parameters $d_x=300$, $d_y=500$, $N=5000$, and $T=10000$. Figure~\ref{fig:appendix-results} shows the relationships between the parameter $\varepsilon$, empirical error, time, and space for aDS-COD and Adaptive-AS. Figures~\ref{fig:appendix-avg-err} and~\ref{fig:appendix-eps-time} show that under the guarantee of the same approximation quality, the update time of the Adaptive-AS algorithm integrated with \sidscod is generally lower than that of aDS-COD, and this acceleration becomes more pronounced, especially when the error approaches zero. In addition, Figures~\ref{fig:appendix-size-err} and~\ref{fig:appendix-size-eps} show that, under the guarantee of the same approximation quality, Adaptive-AS can achieve a generally lower space overhead than aDS-COD.

\end{document}